\documentclass[notitlepage,superscriptaddress,nofootinbib,tightenlines,twocolumn,preprintnumbers]{revtex4-1}
\usepackage{amsmath,verbatim,latexsym,amssymb,indentfirst,mathrsfs,mathtools,amsthm,bbm,bm,url,cancel,subcaption,enumerate}
\usepackage[font=small,labelfont=bf,format=plain,justification=raggedright,singlelinecheck=false]{caption}
\usepackage{verbatim,indentfirst}
\usepackage[title,titletoc]{appendix}
\usepackage{titlesec}

\usepackage{bbm}

\usepackage[shortlabels]{enumitem}

\usepackage{graphicx}
\usepackage{epstopdf}
\newcommand{\caphead}[1]{{\bf #1}}

\renewcommand{\thesection}{\Roman{section}}
\renewcommand{\thesubsection}{\Roman{section} \Alph{subsection}}
\renewcommand{\thesubsubsection}{\Roman{section} \Alph{subsection} \arabic{subsubsection}}
\makeatletter
\def\p@subsection{}
\makeatother
\makeatletter
\def\p@subsubsection{}
\makeatother

\newtheorem{theorem}{Theorem}
\newtheorem{lemma}{Lemma}
\newtheorem{corollary}{Corollary}

\newtheorem{claim}{Claim}

\makeatletter
\newcommand\footnoteref[1]{\protected@xdef\@thefnmark{\ref{#1}}\@footnotemark}
\makeatother

\usepackage{soul}

\usepackage{color}
\usepackage[dvipsnames]{xcolor}
\usepackage[normalem]{ulem}

\usepackage{todonotes}

\newcommand{\Bell}{\mathcal{B}}
\newcommand{\tradBell}{\tilde{\Bell}}
\newcommand{\Sing}{\Psi^-}
\newcommand{\even}{{\rm even}}
\newcommand{\odd}{{\rm odd}}
\newcommand{\ent}{{\rm ent}}
\newcommand{\Even}{{\rm even}}
\newcommand{\RPos}{\mathcal{R}_z}  % Unitary defined on a Posner Hilbert space. Rotates one qubit about the z-axis.

\def\id{\mathbbm{1}}   % Identity
\newcommand{\LParen}{ \bm{(} }
\newcommand{\RParen}{ \bm{)} }

\newcommand{\p}[1]{{\mathbb P}\left(#1\right)}
\newcommand{\cX}{\mathcal{X}}
\newcommand{\cY}{\mathcal{Y}}
\newcommand{\erf}{\mathrm{erf}}
\DeclareMathOperator{\sign}{sgn}

\newcommand*{\bra}[1]{\langle #1\rvert}
\newcommand*{\ket}[1]{\lvert #1 \rangle}
\newcommand*{\braket}[2]{\langle #1 \lvert #2 \rangle}

\newcommand*{\expval}[1]{\left\langle  #1  \right\rangle}

\newcommand{\Var}[1]{\mathrm{Var}\left(#1\right)}

\newcommand{\cov}[2]{\mathrm{Cov}\left(#1,#2\right)}
\newcommand{\E}[1]{{\mathbb E}\left(#1\right)}
\DeclarePairedDelimiter\abs{\lvert}{\rvert}%

\usepackage{xr-hyper}
\usepackage{hyperref}
\usepackage{cleveref}

\makeatletter
\newcommand*{\addFileDependency}[1]{% argument=file name and extension
  \typeout{(#1)}
  \@addtofilelist{#1}
  \IfFileExists{#1}{}{\typeout{No file #1.}}
}
\makeatother
 
\newcommand*{\myexternaldocument}[1]{%
    \externaldocument{#1}%
    \addFileDependency{#1.tex}%
    \addFileDependency{#1.aux}%
}

\myexternaldocument{Suppl_mat}

\renewcommand\th{ {\rm th} }

\begin{document}
\title{Nonlinear Bell inequality for macroscopic measurements}
\author{Adam~Bene~Watts}
\email{abenewat@mit.edu}
\affiliation{Center for Theoretical Physics, 
Massachusetts Institute of Technology,
Cambridge, Massachusetts 02139, USA}
\author{Nicole~Yunger~Halpern}
\email{nicoleyh@g.harvard.edu}
\affiliation{ITAMP, Harvard-Smithsonian Center for Astrophysics, Cambridge, MA 02138, USA}
\affiliation{Department of Physics, Harvard University, Cambridge, MA 02138, USA}
\affiliation{Research Laboratory of Electronics, Massachusetts Institute of Technology, Cambridge, Massachusetts 02139, USA}
\affiliation{Center for Theoretical Physics, 
Massachusetts Institute of Technology,
Cambridge, Massachusetts 02139, USA}
\author{Aram~Harrow}
\email{aram@mit.edu}
\affiliation{Center for Theoretical Physics, 
Massachusetts Institute of Technology,
Cambridge, Massachusetts 02139, USA}
\date{\today}

%
% Abstract
%
\begin{abstract}
The correspondence principle suggests that quantum systems grow classical when large. Classical systems cannot violate Bell inequalities. 
Yet agents given substantial control can violate
Bell inequalities proven for large-scale systems.
We consider agents who have little control,
implementing only general operations suited to macroscopic experimentalists: 
preparing small-scale entanglement
and measuring macroscopic properties while suffering from noise.
That experimentalists so restricted can violate a Bell inequality
appears unlikely, in light of earlier literature.
Yet we prove a Bell inequality that such an agent can violate,
even if experimental errors have variances that scale as the system size.
A violation implies nonclassicality, given limitations on particles' interactions.  
A product of singlets violates the inequality; experimental tests are feasible for photons, solid-state systems, atoms, and trapped ions. Consistently with known results, violations of our Bell inequality cannot disprove local hidden-variables theories. By rejecting the disproof goal, we show, one can certify nonclassical correlations under reasonable experimental assumptions.
\end{abstract}
\preprint{MIT-CTP/5279}

{\let\newpage\relax\maketitle}
% \maketitle{}
% \tableofcontents

Can large systems exhibit nonclassical behaviors such as entanglement?
The correspondence principle suggests not. 
Yet experiments are pushing the quantum-classical boundary
to larger scales~\cite{Arndt_99_Wave,Leggett_02_Testing,Blencowe_04_Nanomechanical,Gerlich_11_Quantum,Wollman_15_Quantum,Kaltenbaek_16_Macroscopic,bose2017spin}:
Double-slit experiments have revealed interference
% of fullerene wave functions~\cite{Arndt_99_Wave} and 
of organic molecules' wave functions~\cite{Gerlich_11_Quantum}.
A micron-long mechanical oscillator's quantum state 
has been squeezed~\cite{Wollman_15_Quantum}.
Many-particle systems have given rise to nonlocal correlations~\cite{Yao_09_Quantum,Marinkovic_18_Optomechanical,Schmid_16_Bell}.

Nonlocal correlations are detected with Bell tests.
In a Bell test, systems are prepared, separated, and measured 
in each of many trials.
The outcome statistics may violate a Bell inequality.
If they do, they cannot be modeled with classical physics,
in the absence of loopholes.

Bell inequalities have been proved for settings that involve large scales (e.g.,~\cite{Reid_02_Violation,Mermin_80_Quantum,Drummond_83_Violations,Mermin_90_Extreme,Ardehali_92_Bell,Belinskii_93_Interference,Zukowski_02_Bell's,Collins_02_Bell,Cavalcanti_07_Bell,jeong2009failure,He_10_Bell,Stobinska_11_Bell,He_11_Entanglement,Tura_14_Detecting,Tura_15_Nonlocality,Schmid_16_Bell,Engelsen_17_Bell,Marinkovic_18_Optomechanical,Dalton_19_CGLMP,thenabadu2019testing,navascues2013testing});
see~\cite{Dalton_19_Bell,Reid_12_Entanglement} for reviews
and Supplementary Note~\ref{App_Comparison} for a detailed comparison with our results.
We adopt a different approach, considering which operations
a macroscopic experimentalist can perform easily:
preparing small-scale entanglement 
and measuring large-scale properties, in our model.
Whether such a weak experimentalist can violate a Bell inequality, 
even in the absence of noise, is unclear \emph{a priori}.
Indeed, our experimentalist can violate neither Bell's 1964 inequality~\cite{Bell_64_On,Navascues_10_Glance}, nor any previously proved macroscopic Bell inequality to which our main result does not reduce~\cite{Reid_02_Violation,Mermin_80_Quantum,Drummond_83_Violations,Mermin_90_Extreme,Ardehali_92_Bell,Belinskii_93_Interference,Zukowski_02_Bell's,Collins_02_Bell,Cavalcanti_07_Bell,jeong2009failure,He_10_Bell,Stobinska_11_Bell,He_11_Entanglement,Tura_14_Detecting,Tura_15_Nonlocality,Schmid_16_Bell,Engelsen_17_Bell,Marinkovic_18_Optomechanical,Dalton_19_CGLMP,thenabadu2019testing}.\footnote{
Navascu\'es \emph{et al.} prove a macroscopic Bell inequality that governs a similarly restricted experimentalist~\cite{navascues2013testing}.
However,~\cite{navascues2013testing} does not address noise,
with respect to which our result is robust.
See Supplementary Note~\labelcref{App_Comparison} for a detailed comparison of~\cite{navascues2013testing} with our result.}
Nevertheless, we prove a macroscopic Bell inequality
that can be violated with these operations, even in the presence of noise. 
% \{Our inequality opens the door to Bell tests in a new experimental regime 
% accessible with simpler, noisier operations.}
%%% <--NYH: This line could be helpful if we weren't restricted to PRL's page limit. Since the line basically recapitulates previous lines, it seems redundant and so fodder for cutting.
The key is the macroscopic Bell parameter's nonlinearity
in the probability distributions over measurement outcomes.
% \q{[The foregoing sentence is necessary for two reasons: 
% (i) Whenever one physicist says, ``Look; I've smashed this totally intuitive statement!''
% listening physicists respond, ``What's the trick?'' We should pre-empt the question.
% (ii) The introduction should overview the paper's key components.
% ``Nonlinear'' is in the paper's title, so the nonlinearity is key.]}

Our inequality is violated by macroscopic measurements of, 
e.g., a product of $N > 1$ singlets.
Such a state has been prepared in a wide range of platforms, including
photons~\cite{Barz_10_Heralded}, 
solid-state systems~\cite{Bernien_13_Heralded},
atoms~\cite{Laurat_07_Heralded,Hofmann_12_Heralded},
and trapped ions~\cite{Casabone_13_Heralded}.
% A photonic experiment is underway~\cite{Wong_Experiment}.
A violation of the inequality implies nonlocality if
microscopic subsystems are prepared approximately independently.
Similarly, independence of pairs of particles is assumed in~\cite{Navascues_10_Glance,Yang_11_Quantum,Navascues_16_Macroscopic}, though it may be difficult to guarantee.
%%% NYH: The foregoing sentence lends support to our assumption's reasonableness. We need for the reader to accept our restrictions ASAP.
% and Brunner/Branciard papers that I mentioned somewhere...
% Where the cited papers state the independence assumption:
% Navascues_10_Glance -- p. 884 -- "N independent pairs"
% Yang_11_Quantum -- p. 1 -- "N pairs are sent" -- independence implied
% Navascues_16_Macroscopic -- p. 442 -- "independent identical pairs"

This independence requirement prevents violations of our inequality 
from disproving local hidden-variables theories (LHVTs),
as no experimentalist restricted like ours can~\cite{Navascues_10_Glance,Yang_11_Quantum,Navascues_16_Macroscopic}.
By forfeiting the goal of a disproof, we show,
one can certify entanglement under reasonable experimental assumptions.
This certification is device-independent, requiring no knowledge of the state or experimental apparatuses, apart from the aforementioned independence.
Furthermore, our inequality is robust with respect to errors,
including a lack of subsystem independence,
whose variances scale as $N$.
Additionally, with our strategy, similar macroscopic Bell inequalities
can be derived for macroscopic systems 
that satisfy different independence assumptions.

Aside from being easily testable with platforms known to produce Bell pairs,
our inequality can illuminate whether 
poorly characterized systems harbor entanglement.
Such tests pose greater challenges but offer greater potential payoffs.
Possible applications include Posner molecules~\cite{Fisher_15_Quantum,Fisher_18_Quantum,NYH_19_Quantum,Fisher_17_Are}, tabletop experiments that simulate cosmological systems~\cite{Fifer_19_Analog},
and high-intensity beams.

The rest of this paper is organized as follows.
We introduce the setup in Sec.~\ref{sec_Setup}.
Section~\ref{sec_Bell_Ineq} contains the main results:
We present and prove the Bell inequality for macroscopic measurements,
using the covariance formulation of 
a microscopic Bell inequality~\cite{Pozsgay_17_Covariance}.
Section~\ref{sec_Discussion} contains a discussion: 
We compare quantum correlations and global classical correlations
as resources for violating our inequality,
show how to combat experimental noise,
reconcile violations of the inequality with the correspondence principle~\cite{Navascues_10_Glance,Yang_11_Quantum,Navascues_16_Macroscopic},
recast the Bell inequality as a nonlocal game,
discuss a potential application to Posner molecules~\cite{Fisher_15_Quantum,Fisher_18_Quantum,NYH_19_Quantum,Fisher_17_Are},
and detail opportunities.

%
%
% Setup
%
%
\section{Setup}
\label{sec_Setup}

Consider an experimentalist Alice who has a system $A$
and an experimentalist Bob who has a disjoint system $B$. 
Each system consists of $N$ microscopic subsystems, indexed with $i$.
The $i^\th$ subsystem of $A$ can interact with
the $i^\th$ subsystem of $B$ but with no other subsystems.
Our setup resembles that in~\cite{Navascues_10_Glance}.

Alice can measure her system with settings $x = 0, 1$,
and Bob can measure his system with settings $y = 0, 1$.
% Each microscopic subsystem reports 
%%% <-- A foundational scientist might take umbrage at the ontological attitude. I've swapped it for an operational approach.
Each measurement yields an outcome in $[0,1]$.\footnote{
% < f >
In the strategies presented explicitly in this paper,
every measurement outcome equals 0 or 1. 
But the macroscopic Bell inequality holds more generally.
% The setup is more general to emphasize that 
% the macroscopic Bell inequality also holds against 
% classical systems consisting of microscopic subsystems which, 
% when measured, report outcomes from an infinite set.
}
% < /f > 
The experimentalist observes the sum of the microscopic outcomes, 
the value of a \emph{macroscopic random variable}. 
Measuring $A$ with setting $x$ 
yields the macroscopic random variable $A_x$. 
$B_y$ is defined analogously. 

We will often illustrate with two beams of photons.
The polarization of each photon in beam $A$ is entangled with
the polarization of a photon in beam $B$ and vice versa. 
Such beams can be produced through
spontaneous parametric down-conversion (SPDC)~\cite{kwiat1995new}:
A laser beam strikes a nonlinear crystal.
Upon absorbing a photon, the crystal emits two photons
entangled in the polarization domain:
$\frac{1}{ \sqrt{2} } ( \ket{ \text{H}, \text{V} }  +  e^{i \alpha} \ket{ \text{V}, \text{H} } )$.
Horizontal and vertical polarizations are denoted by
$\ket{ \text{H} }$ and $\ket{ \text{V} }$.
The relative phase depends on some $\alpha \in \mathbb{R}$.
The photons enter different beams.
Each experimentalist measures his/her beam
by passing it through a polarizer, 
then measuring the intensity. 
The measurement setting (Alice's $x$ or Bob's $y$) determines the polarizer's angle. 
A photon passing through the polarizer yields a 1 outcome.
The intensity measurement counts the 1s. 
Supplementary Note~\labelcref{app_Feasibility} addresses concerns about the feasibility of realizing our model experimentally. Supplementary Note~\labelcref{app_Noise_Model} details the photon example.

% Three sources underlie the randomness in the $A_x$'s and $B_y$'s:
The randomness in the $A_x$'s and $B_y$'s is of three types:
\begin{enumerate}[(i)]

   \item \label{item_Randomness_i}
   \emph{Quantum randomness:}
   If the systems are quantum, outcomes are sampled according to the Born rule
during wave-function collapse.

   \item  \label{item_Randomness_ii}
   \emph{Local classical randomness:}
   Randomness may taint the preparation of each $AB$ pair of subsystems.
   % Reference: Meeting notes --> Adam, Aram - 10/29/19 --> p. 5. Each pair can be random, if classical, and in a mixed state, if quantum.
In the SPDC example, different photons enter the crystal at different locations.
Suppose that the crystal's birefringence varies over short length scales.
Different photon pairs will acquire different relative phases $e^{i \alpha}$~\cite{kwiat1995new}.
% Reference: Kwiat --> p. 2 --> LHS

   \item \label{item_Randomness_iii}
   \emph{Global classical randomness:}
   Global parameters that affect all the particle pairs can vary from trial to trial. 
   In the photon example, Alice and Bob can switch on the laser;
   measure their postpolarizer intensities several times, 
   performing several trials, during a time $T$; 
   and then switch the laser off.
   The laser's intensity affects the $A_x$'s and $B_y$'s 
   and may fluctuate from trial to trial.

\end{enumerate}
Quantum randomness and global classical randomness 
can violate our macroscopic Bell inequality.
Assuming a cap on the amount of global classical randomness, % \ref{item_Randomness_iii},
we conclude that violations imply nonclassicality. 
Local classical randomness can conceal violations 
achievable by quantum systems ideally.
Local classical randomness also produces limited correlations,
which we bound in our macroscopic Bell inequality.
We quantify classical randomness with a noise variable $r$ below.

Systems $A$ and $B$ satisfy two assumptions:
% NYH (10/25/19): Just assumption 2 is violated. Assumption 1 is relatively easy to guarantee.
\begin{enumerate}[(a)]
    \item 
    $A$ and $B$ do not interact with each other while being measured. 
    Neither system receives information about the setting with which
    the other system is measured.
    \label{assumption: noninteracting systems}
    
    \item 
    Global classical correlations are limited, 
    as quantified in Ineq.~\eqref{eq_Bound_Var}.
    \label{assumption: independent particles}
    
\end{enumerate}

\noindent
Assumption \labelcref{assumption: noninteracting systems} 
is standard across Bell inequalities. 
In the photon example, the beams satisfy~\ref{assumption: noninteracting systems} if spatially separated while passing through the polarizers and undergoing intensity measurements.

Assumption~\ref{assumption: independent particles} is the usual assumption
that parameters do not fluctuate too much between trials,
due to a separation of time scales.
Consider the photon example in item~\ref{item_Randomness_ii} above.
Let $t$ denote the time required to measure the intensity,
to perform one trial.
The trial time must be much shorter than
the time over which the global parameters drift (e.g., the laser intensity drifts):
$t \gg T$.
The greater the time scales' separation,
the closer the system comes to satisfying 
assumption~\ref{assumption: independent particles}.
Assumption~\labelcref{assumption: independent particles} 
has appeared in other studies of nonclassical correlations in macroscopic systems (e.g.,~\cite{Navascues_10_Glance, Navascues_16_Macroscopic}).
% The beams satisfy~\ref{assumption: independent particles} because the photon pairs are produced independently.

Assumptions~\ref{assumption: noninteracting systems} 
and ~\ref{assumption: independent particles}
are the conditions under which
a Bell inequality is provable for
the operations that a macroscopic experimentalist
is expected to be able to perform:
correlating small systems
%%% NYH: If we're discussing a Bell test, we can't expect an experimentalist to prepare entanglement: We're not supposed to claim that the experimentalist prepares any particular states; we're supposed to make just operational statements.
and measuring macroscopic observables.\footnote{
% < f >
Why these operations? Preparing macroscopic entanglement is difficult; 
hence the restriction to microscopic preparation control.
Given microscopic preparation control, 
if the experimentalist could measure microscopic observables,
s/he could test the microscopic Bell inequality;
a macroscopic Bell inequality would be irrelevant.
}
% < /f >
If the experimentalist can perform different operations,
different assumptions will be natural,
and our macroscopic Bell test may be extended (Sec.~\ref{sec_Discussion}).

% Reference about how the two types of classical randomness manifest in r: email chain "Randomness of types ii and iii" -- message sent by Adam on 11/16/19
We fortify our Bell test by allowing for
small global correlations and limited measurement precision.
Both errors are collected in one parameter, defined as follows.
In the absence of errors, $A_x$ and $B_y$
equal ideal random variables $A'_x$ and $B'_y$.
Each ideal variable equals a sum of independent random variables.
% The $A_{x}$ probabilities can differ from products
% due to local classical~\ref{item_Randomness_ii} 
% and global classical~\ref{item_Randomness_iii} randomness.
We model the discrepancies between ideal and actual with random variables $r$, as in
\begin{align}
   \label{eq_Error_Model}
    A_x = A_x' + r_{A_x} .
\end{align}
Our macroscopic Bell inequality is robust with respect to errors
of bounded variance:
\begin{align}
   \label{eq_Bound_Var}
    \Var{r_{A_x}} \leq \epsilon N ,
\end{align}
wherein $\epsilon > 0$. 
Errors $r_{B_y}$ are defined analogously.
They obey Ineq.~\eqref{eq_Bound_Var} with the same $\epsilon$.
Strategies for mitigating errors are discussed in Sec.~\ref{sec_Discussion}.

Our macroscopic Bell inequality depends on the covariances of 
the $A_x$'s and $B_y$'s.
The covariance of random variables $X$ and $Y$ is defined as
\begin{align}
    \cov{X}{Y} := \mathbb{E} \LParen [X- \E{X}] [Y-\E{Y}] \RParen ,
\end{align}
wherein $\E{X}$ denotes the expectation value of $X$.
One useful combination of covariances, we define as
the \emph{macroscopic Bell parameter}:\footnote{
% < f >
Calculating $\Bell$ requires knowledge of $N$,
the number of particles in each experimentalist's system.
$N$ might not be measurable precisely. 
But knowing $N$ even to within $\sqrt{N}$ suffices:
Taylor-approximating yields
$\frac{1}{N + \sqrt{N} }  
=  \frac{1}{N} \left( 1 - \frac{1}{ \sqrt{N} } \right)$.
The correction is of size $\frac{1}{ \sqrt{N} }  \ll  1$.
Furthermore, uncertainty about $N$ may be incorporated into a noise model
with which a macroscopic Bell inequality can be derived (Suppl. Note~\ref{app_Noise_Model}).}
% < /f >
\begin{align}
    \label{eq_Bell_Param}
    \Bell(A_0,A_1,B_0,B_1) & := \frac{4}{N}         
    [ \cov{A_{0}}{B_0} + \cov{A_0  }{B_1}
        \nonumber \\  
        & \hspace{0.5cm}
        + \cov{A_1}{B_0} - \cov{A_1}{B_1} ] .
\end{align}
% This parameter will appear throughout the paper. 

% Dependent on the macroscopic outcomes' means, $B(A_x,A_y,B_x,B_y)$ cannot be estimated directly. Rather, experimentalists must estimate the macroscopic outcomes' means
% and, from those, the covariance.
% We denote by $\tilde{A}_x$ 
% the experimentalists' estimate of the mean of $A_x$. 
% We assume that the estimate's accuracy is bounded as
% \begin{align}
%     \label{eq_Def_Eps2}
%     \abs{\E{A_x} -\tilde{A}_x } \leq \epsilon_2  \sqrt{N} ,
% \end{align}
% for some $\epsilon_2  \geq  0$.
% Variables $\tilde{A}_y$, $\tilde{B}_x$, and $\tilde{B}_y$
% are defined analogously and are assumed to obey analogous bounds.
% The experimentalists' estimate of $B(A_x,A_y,B_x,B_y)$,
% from the means
% $\tilde{A}_x$, $\tilde{A}_y$, $\tilde{B}_x$, and $\tilde{B}_y$,
% has the form
% \begin{align}
%     \label{eq_Noisy_Bell_Param}
%     \tilde{B}(A_x,A_y,B_x,B_y) 
%     &:= \frac{4}{N} \bigg[\E{\left(A_x - \tilde{A}_x  \right)\left(B_x - \tilde{B}_x\right)} \nonumber \\
% &\hspace{0.2cm}+ \E{\left(A_x - \tilde{A}_x  \right)\left(B_y - \tilde{B}_y\right)} \nonumber \\
% &\hspace{0.2cm} + \E{\left(A_y - \tilde{A}_y \right)\left(B_x - \tilde{B}_x\right)} \nonumber \\
% &\hspace{0.2cm}-\E{\left(A_y - \tilde{A}_y \right)\left(B_y - \tilde{B}_y \right)}\bigg].
% \end{align}

%
%
% Main results
%
%
\section{Main results}
% Nonlinear Bell inequality for macroscopic measurements and quantum violation}
\label{sec_Bell_Ineq}

We present the nonlinear macroscopic Bell inequality and sketch the proof, 
detailed in Suppl. Note~\ref{app_Proof_Bell_ineq}.
Then, we show how to violate the inequality using quantum systems.

\begin{theorem}[Nonlinear Bell inequality for macroscopic measurements] 
\label{thm_marco_bell_ineq} 
Let systems $A$ and $B$, and measurement settings $x$ and $y$, 
be as in Sec.~\ref{sec_Setup}. 
Assume that the systems are classical. 
The macroscopic random variables satisfy the macroscopic Bell inequality
\begin{align}
   \label{eq_Our_Bell_Ineq}
\Bell(A_0,A_1,B_0,B_1) &\leq 16/7 
+ 16 \epsilon + 32\sqrt{\epsilon} .
\end{align}
\end{theorem}

\begin{proof}
Here, we prove the theorem when $\epsilon = 0$,
when the observed macroscopic random variables $A_x$ and $B_y$
equal the ideal $A'_x$ and $B'_y$.
The full proof is similar but requires an error analysis (Suppl. Note~\ref{app_Proof_Bell_ineq}).

Let $a_{x}^{(i)}$ denote the value reported by the $i^\th$ $A$ particle 
after $A$ is measured with setting $x$. 
$A_x'$ and $B_y'$ equal sums of the microscopic variables:
\begin{align} 
   A'_x  =  \sum_{i=1}^N  a_x^{(i)} ,
   \quad \text{and} \quad
   B'_x  =  \sum_{i=1}^N  b_x^{(i)} .
\end{align}
Because $a_0^{(i)}$ and $b_0^{(i)}$ are independent of the other variables,
\begin{align}
    \label{eq_Bound_Help1}
    \cov{A_0'}{B_0'} 
    = \sum_{i = 1}^N  \cov{  a_{0}^{(i)}  }{  b_{0}^{(i)}  } .
\end{align}
Analogous equalities govern the other macroscopic-random-variable covariances.

Let us bound the covariances amongst the $a_x^{(i)}$'s and $b_y^{(i)}$'s.
We use the covariance formulation of 
the Bell-Clauser-Horne-Shimony-Holt (Bell-CHSH) inequality (see~\cite{Clauser_69_Proposed,Pozsgay_17_Covariance}
and Suppl. Note~\ref{app_CHSH_Backgrnd}),\footnote{
% < f >
In the original statement of Ineq.~\eqref{eq_Cov_Bell0},
the right-hand side (RHS) equals $16/7$.
The reason is, in~\cite{Pozsgay_17_Covariance},
$a_x^{(i)}, b_y^{(i)}  \in  [-1, 1]$.
We assume that each variable $\in [0, 1]$, 
so we deform the original result in two steps.
First, we translate $[-1, 1]$ to $[0, 2]$.
Translations preserve covariances.
Second, we rescale $[0, 2]$ to $[0, 1]$.
The rescaling halves each $a$ and $b$,
quartering products $ab$, the covariances, 
and the $16/7$ in Ineq.~\eqref{eq_Cov_Bell0}.
The resulting $4/7$ is multiplied by a 4 in Ineq.~\eqref{eq_Bound_Help2},
returning to $16/7$.}
% < /f >
\begin{align}
   \label{eq_Cov_Bell0}
   & \cov{a_{0}^{(i)}  }{b_{0}^{(i)}  } + \cov{a_{0}^{(i)}  }{b_{1}^{(i)}  }  
   + \cov{a_{1}^{(i)}  }{b_{0}^{(i)}  } 
   \nonumber \\ & \qquad \qquad \qquad \qquad \qquad \;
   - \cov{a_{1}^{(i)}  }{b_{1}^{(i)}  } 
   \leq 4/7 .
\end{align}
Combining Eq.~\eqref{eq_Bound_Help1} and Ineq.~\eqref{eq_Cov_Bell0} with the definition of $\Bell(A_x', A_y', B_x', B_y')$ [Eq.~\eqref{eq_Bell_Param}] gives
\begin{align}
    \label{eq_Bound_Help2}
    & \Bell(A_0', A_1', B_0', B_1') 
    = \frac{4}{N}   \sum_{i = 1}^N
    \Big[  \cov{a_0^{(i)}}{b_0^{(i)}} 
    \\ & \nonumber \quad + \cov{a_0^{(i)}}{b_1^{(i)}} 
    % \nonumber \\ & \hspace{0.5cm} 
    + \cov{a_1^{(i)}}{b_0^{(i)}} - \cov{a_1^{(i)}}{b_1^{(i)}} \Big] \\
    % % %
    \label{eq_Bound_Help3}
    & \; \leq 16/7 .
\end{align}
\end{proof}

We now show that a quantum system 
can produce correlations that violate Ineq.~\eqref{eq_Our_Bell_Ineq}.
The system consists of singlets.

\begin{theorem} \label{thm_quantum_violation}
There exist an $N$-particle quantum system 
and a measurement strategy, 
subject to the restrictions in Sec.~\ref{sec_Setup}, 
whose outcome statistics violate 
the nonlinear Bell inequality for macroscopic measurements.
The system and strategy achieve
\begin{align}
\Bell(A_0,A_1,B_0,B_1) = 2\sqrt{2}
\end{align}
in the ideal ($\epsilon = 0$) case and
\begin{align}
\label{eq_Q_Bound_Main}
\Bell(A_0,A_1,B_0,B_1) \geq 2\sqrt{2} - 16 \epsilon - 32 \sqrt{\epsilon}
\end{align}
in the presence of noise bounded as in Ineq.~\eqref{eq_Bound_Var}.
\end{theorem}

\begin{proof}
As in the proof of \Cref{thm_marco_bell_ineq}, 
we prove the result in the ideal case here.
Supplementary Note~\ref{app_Prove_Q_Violation} contains the error analysis.
Let each of $A$ and $B$ consist of $N$ qubits.
Let the $i^\th$ qubit of $A$ and the $i^\th$ qubit of $B$ form a singlet, for all $i$:
$\ket{ \Sing}  :=  \frac{1}{ \sqrt{2} } ( \ket{01} - \ket{10} )$.
We denote the $1$ and $-1$ eigenstates 
of the Pauli $z$-operator $\sigma_z$
by $\ket{0}$ and $\ket{1}$.
Let $x$ and $y$ be the measurement settings in
the conventional CHSH test (\cite{Clauser_69_Proposed},~reviewed in Suppl. Note~\ref{app_CHSH_Backgrnd}).
If the measurement of a particle yields 1, the particle effectively reports 1;
and if the measurement yields $-1$, the particle reports 0.

Measuring the $i^\th$ particle pair yields outcomes that satisfy
\begin{align}
    \label{eq_Q_Help0}
    \E{a_{0}^{(i)}} = \E{a_{1}^{(i)}} = \E{b_{0}^{(i)}} = \E{b_{1}^{(i)}} = \frac{1}{2} .
\end{align}
As shown in Suppl. Note~\ref{app_Prove_Q_Violation},
\begin{align}
    \label{eq_Q_Help}
    &\E{a_{0}^{(i)}b_{0}^{(i)}} + \E{a_{0}^{(i)}b_{1}^{(i)}}  
    \\ \nonumber
    & \hspace{0.5cm} + \E{a_{1}^{(i)}b_{0}^{(i)}} - \E{a_{1}^{(i)}b_{1}^{(i)}} 
    = 2 \sin^2(3 \pi/8) - \frac{1}{2}.
\end{align}
Combining these two equations yields
\begin{align}
    \label{eq_Q_Help2}
    &\cov{a_{0}^{(i)}}{b_{0}^{(i)}} + \cov{a_{0}^{(i)}}{b_{1}^{(i)}} \\
    &\hspace{0.1cm} + \cov{a_{1}^{(i)}}{b_{0}^{(i)}} 
    - \cov{a_{1}^{(i)}}{b_{1}^{(i)}} 
    % % %
    = 2 \sin^2(3 \pi/8) - 1. \nonumber \\
    % % %
    & \hspace{5.3cm} = 1 / \sqrt{2} .
\end{align}
Following the proof of \Cref{thm_marco_bell_ineq}, we compute
\begin{align}
    &\Bell(A_0',A_1',B_0',B_1') \\
    &\hspace{1cm} 
    = \frac{4}{N}\sum_i \Big[ 
    \cov{a_0^{(i)}}{b_0^{(i)}} + \cov{a_0^{(i)}}{b_1^{(i)}} \nonumber \\
    & \hspace{1.5cm} + \cov{a_1^{(i)}}{b_0^{(i)}} 
    - \cov{a_1^{(i)}}{b_1^{(i)}} \Big] \\
    &\hspace{1cm} = 2\sqrt{2} . \label{eq_ideal_quantum_macroscopic}
\end{align}
\end{proof}

\section{Discussion}
\label{sec_Discussion}

Six points merit analysis.
First, we discuss the equivalence of local quantum correlations
and global classical correlations
as resources for violating the macroscopic Bell inequality.
% In Sec.~\ref{sec_Error}, 
Second, we suggest strategies for mitigating experimental errors.
% Section~\ref{sec_Navascues} 
Third, we reconcile our macroscopic-Bell-inequality violation 
with the principle of macroscopic locality,
which states that macroscopic systems should behave classically~\cite{Navascues_10_Glance,Yang_11_Quantum,Navascues_16_Macroscopic}. 
% Section~\ref{sec_Game} 
Fourth, we recast our macroscopic Bell inequality 
in terms of a nonlocal game. 
% Section~\ref{sec_Posners} 
Fifth, we discuss a potential application to the Posner model of quantum cognition~\cite{Fisher_15_Quantum,Fisher_18_Quantum,NYH_19_Quantum,Fisher_17_Are}.
% Section~\ref{sec_Opportunities} 
Sixth, we detail opportunities engendered by this work.

\textbf{Violating the macroscopic Bell inequality with
classical global correlations:}
Violating the inequality~\eqref{eq_Our_Bell_Ineq}
is a quantum information-processing (QI-processing) task.
Entanglement fuels some QI-processing tasks
equivalently to certain classical resources
(e.g.,~\cite{Buhrman_01_Quantum}).
% Reference: Meeting notes --> Adam, Aram - 10/29/19 --> fingerprinting discussion
In violating the macroscopic Bell inequality, 
entanglement within independent particle pairs
serves equivalently to global classical correlations.
We prove this claim in Suppl. Note~\ref{sec_violating_nonlin_bell_ineq}.
% The reason is the macroscopic Bell parameter's nonlinearity in 
% the probabilities according to which 
% the measurement outcomes are distributed:
% Let $\mathcal{P}_1$ and $\mathcal{P}_2$ denote probability distributions
% that fail to violate the macroscopic Bell inequality.
% Linearly combining the distributions---selecting 
% from $\mathcal{P}_1$ with a probability $p$
% and from $\mathcal{P}_2$ with a probability $1 - p$---can violate the inequality. 
% $p$ serves as a global classical correlation.
This result elucidates entanglement's power in QI processing.

\textbf{Two strategies for mitigating experimental imperfections:}
% \label{sec_Error}
Imperfections generate local classical~\ref{item_Randomness_ii}
and global classical~\ref{item_Randomness_iii}
randomness, discussed in Sec.~\ref{sec_Setup}.
Local classical randomness can conceal 
quantum violations of the macroscopic Bell inequality,
making the macroscopic Bell parameter $\Bell$~\eqref{eq_Bell_Param}
appear smaller than it should.
Global classical randomness can lead classical systems 
to violate the inequality.
These effects can be mitigated in two ways.

First, we can reduce the effects of local classical randomness on $\Bell$
by modeling noise more precisely than in Sec.~\ref{sec_Setup}.
A macroscopic Bell inequality tighter than Ineq.~\eqref{eq_Our_Bell_Ineq}
may be derived.
We illustrate in Suppl. Note~\ref{app_Noise_Model},
with noise that acts on the microscopic random variables
$a_x^{(i)}$ and $b_y^{(i)}$ independently.
Second, we can mitigate global classical randomness
by reinitializing global parameters between trials. 
In the photon example, the laser can be reset between measurements.
% Some apparatuses may age irreversibly.
% For example, the crystal used in SPDC might wear down.
% But typical aging time scales are expected
% to far exceed a trial's time scale.

%
%
%
\textbf{Reconciliation with the principle of macroscopic locality:}
% \label{sec_Navascues}
%
% Reference: Meeting notes --> Adam, Aram - 10/22/19 --> p. 1-2
%
Macroscopic locality has been proposed as an axiom
for distinguishing quantum theory from other nonclassical probabilistic theories~\cite{Navascues_10_Glance,Yang_11_Quantum,Navascues_16_Macroscopic}
(see~\cite{reid2001proposal, reid2001new} for a more restrictive proposal).
%
%%% Adam addition. 
% The principle of macroscopic local realism posits that macroscopic systems can be described by local hidden variables. Quantum mechanics predicts that Bell tests can be performed with macroscopic systems, violating this principle \cite{reid2001proposal, reid2001new}.  
% Macroscopic locality is a weakening of this principle which has been proposed as an axiom for distinguishing quantum theory from other nonclassical probabilistic theories~\cite{Navascues_10_Glance,Yang_11_Quantum,Navascues_16_Macroscopic}.
%
Suppose that macroscopic properties 
of $N$ independent quantum particles 
are measured with precision $\sim \sqrt{N}$.
The outcomes are random variables that obey a probability distribution $P$.
A LHVT can account for $P$,
according to the principle of macroscopic locality.

The violation of our macroscopic Bell inequality
would appear to violate the principle of macroscopic locality.
But experimentalists cannot guarantee 
the absence of fluctuating global parameters,
no matter how tightly they control the temperature, laser intensity, etc.
Some unknown global parameter
could underlie the Bell-inequality violation,
due to the inequality's nonlinearity (Suppl. Note~\ref{subsec:Nonlinearity}).
This parameter would be a classical, and so local, hidden variable.
Hence violating our macroscopic Bell inequality does not disprove LHVTs.
Rather, a violation signals nonlocal correlations
under reasonable, if not airtight, assumptions about the experiment
(Sec.~\ref{sec_Setup}).

\textbf{Nonlocal game:}
% \label{sec_Game}
The macroscopic Bell inequality gives rise to a nonlocal game.
Nonlocal games quantify what quantum resources can achieve
that classical resources cannot.
The CHSH game is based on the Bell-CHSH inequality
(\cite{Clauser_69_Proposed,CHTW04,Preskill_01_Lecture} and Suppl. Note~\ref{app_CHSH_Backgrnd}):
Players Alice and Bob agree on a strategy;
share a resource, which might be classical or quantum;
receive questions $x$ and $y$ from a verifier;
operate on their particles locally;
and reply with answers $a_x$ and $b_y$.
If the questions and answers satisfy 
$x \wedge y = a + b \pmod{2}$, the players win. 
Players given quantum resources can win more often than classical players can.

Our macroscopic game (Suppl. Note~\ref{app_Nonlocal_Game}) 
resembles the CHSH game but differs in several ways:
$N$ Alices and $N$ Bobs play. The verifier aggregates the Alices' and Bobs' responses, but the verifier's detector has limited resolution. The aggregate responses are assessed with a criterion similar to the CHSH win condition. After many rounds of the game, the verifier scores the player's performance. The score involves no averaging over all possible question pairs $xy$. Players who share pairwise entanglement (such that each Alice shares entanglement with only one Bob and vice versa) can score higher than classical players.

\textbf{Toy application to Posner molecules:}
% \label{sec_Posners}
%
% Reference: Meeting notes --> Adam - 7/25/19
Fisher has proposed a mechanism by which
entanglement might enhance coordinated neuron firing~\cite{Fisher_15_Quantum}.
Phosphorus nuclear spins, he argues, can retain coherence for long times
when in Posner molecules Ca$_9$(PO$_4$)$_6$~\cite{Onuma_98_Posner,Ayako_99_Posner,Dey_10_Posner,Treboux_00_Posner,Kanzaki_01_Posner,Yin_03_Posner,Swift_17_Posner}.
(We call Posner molecules ``Posners'' for short.)
He has argued that Posners might share entanglement.
Fisher's work has inspired developments in 
quantum computation~\cite{NYH_19_Quantum,Freedman_18_Quantum},
chemistry~\cite{Fisher_18_Quantum,Swift_17_Posner},
and many-body physics~\cite{Li_19_Measurement,Chan_19_Unitary,Skinner_19_Measurement}.
The experimental characterization of Posners has begun.
If long-term coherence is observed, entanglement in Posners should be tested for.

How could it be? Posners tumble randomly in their room-temperature fluids.
In Fisher's model, Posners can undergo 
the quantum-computational operations detailed in~\cite{NYH_19_Quantum},
not the measurements performed in conventional Bell tests.
Fisher sketched an inspirational start to an entanglement test in~\cite{Fisher_17_Are}.
% Reference for the above line: Fisher_17_Are -- p. 8 -- bottom LHS -- "add calcium- indicators (molecules that fluoresce when calcium ions bind), pour half of the fluid into a second test-tube, drop a bit of hydrochloric acid into each test tube to encourage the Posner molecules to disassociate, and detect the emitted fluorescence from the calcium-indicators."
Concretizing the test as a nonlocal game was proposed in~\cite{NYH_19_Quantum}.
% Reference for the above line: NYH_19_Quantum -- p. 41 of published version
We initiate the concretization in Suppl. Note~\ref{app_Posner}.
Our Posner Bell test requires microscopic control
but proves that Posners can violate a Bell inequality, in principle, 
in Fisher's model.
Observing such a violation would require more experimental effort than 
violating our inequality with photons.
But a Posner violation would signal never-before-seen physics:
entanglement amongst biomolecules.

\textbf{Opportunities:}
% \label{sec_Opportunities}
This work establishes six avenues of research.
First, violations of our inequality can be observed experimentally.
Potential platforms include photons~\cite{Barz_10_Heralded}, 
solid-state systems~\cite{Bernien_13_Heralded},
atoms~\cite{Laurat_07_Heralded,Hofmann_12_Heralded},
and trapped ions~\cite{Casabone_13_Heralded}.
% A photonic experiment is now underway~\cite{Wong_Experiment}.
These systems could be conscripted relatively easily
but are known to generate nonclassical correlations.
More ambitiously, one could test our macroscopic Bell inequality
with systems whose nonclassicality needs characterization.
Examples include the cosmic microwave background (CMB) and Posner molecules.
% Reference about CMB: Meeting notes --> Cosmology group - 11/1/19
Detecting entanglement in the CMB faces difficulties:
Some of the modes expected to share entanglement
have such suppressed amplitudes, they cannot be measured~\cite{Martin_17_Obstructions}.
Analogs of cosmological systems, however, can be realized in tabletop experiments~\cite{Fifer_19_Analog}.
Such an experiment's evolution can be paused.
% Reference: Fifer_19_Analog -- p. 4 -- bottom LHS: "Our proposed system allows us to stop the expansion"
Consider pausing the evolution before, 
or engineering the evolution to avoid, the suppression.
From our Bell test, one might infer about entanglement in the CMB.
A Posner application would require the elimination of microscopic control
from the Bell test in Suppl. Note~\ref{app_Posner}, opportunity two.

Third, our macroscopic Bell inequality may be generalized to systems that violate the independence requirement in Sec.~\ref{sec_Setup}. Examples include squeezed states, as have been realized with, e.g., atomic ensembles and SPDC~\cite{jing2019split,fadel2018spatial}.
The assumptions in Sec.~\ref{sec_Setup} would need to modified
to accommodate the new setup.
% One would have to identify an appropriate classical system 
% with which to compare the quantum system,
% as we compare our quantum systems with
% classical systems whose particles form approximately independent pairs.
If an experimental system violated the new inequality 
while satisfying the appropriate assumptions, 
one could conclude that the system was nonclassical.
We illustrate such a modification and violation 
in Suppl. Note~\ref{app_Noise_Model}, with a photonic system.
Tailoring our results to a high-intensity pump
appears likely to enable experimentalists to witness entanglement
in systems that violate a common coincidence assumption:
Bell tests tend to require low intensities,
so that only one particle reaches each detector per time window~\cite{Larsson_14_Loopholes}.
The coincidence of a particle's arriving at detector $A$ 
and a particle's arriving at detector $B$
implies that these particles should be analyzed jointly.
High-intensity pumps violate the 
one-particle-per-time-window coincidence assumption.
Tailoring Suppl. Note~\ref{app_Noise_Model},
using Gaussian statistics, appears likely to expand Bell tests
to an unexplored, high-intensity regime.

Fourth, which macroscopic Bell parameters $\Bell$ can 
probabilistic theories beyond quantum theory realize?
Other theories can support correlations
unrealizable in quantum theory~\cite{Popescu_94_Quantum,Janotta_14_Generalized}.
These opportunities can help distinguish quantum theory
from alternative physics
while illuminating the quantum-to-classical transition.

Fifth, our macroscopic Bell parameter is nonlinear in
the probabilities of possible measurements' outcomes
(Suppl. Note~\ref{subsec:Nonlinearity}).
We have proved that a nonlinear operation---photodetection---can violate the inequality.
% Reference: Photodetection is non-Gaussian according to p. 3 of https://arxiv.org/abs/1110.3234.
Can Gaussian operations~\cite{Weedbrook_12_Gaussian}?
The answer may illuminate the macroscopic Bell inequality's limits.

Sixth, certain Bell inequalities have applications to \emph{self-testing}~\cite{Supic_19_Self}. A maximal violation of such an inequality
implies that the quantum state had a particular form.
Whether covariance Bell inequalities can be used in self-testing merits investigation.

% Fourth, macroscopic measurements might distinguish
% nonclassical theories from classical more easily than in our prescription.
% A simpler game (Sec.~\ref{sec_Bell_Ineq} and Suppl. Note~\ref{app_Nonlocal_Game}) 
% might be constructed or a wider quantum-classical gap, proved.
% Such a game or construction could facilitate experimental tests of our results.

%
% Acknowledgements
%
% \section*{Acknowledgements}
\begin{acknowledgments}
  The authors thank Isaac Chuang,  Peter Drummond, Matteo Fadel, Alan Guth, David I. Kaiser, Christopher McNally, Miguel Navascués,  Margaret Reid, Jeffrey H. Shapiro, Robert W. Spekkens, Franco Wong, and John Wright for helpful discussions. 
  We acknowledge inspiration from talks by Matthew Fisher,
  and we thank Elizabeth Crosson for sharing code.
  ABW was funded by NSF grant CCF-1729369.  NYH is
  grateful for an NSF grant for the Institute for Theoretical Atomic,
  Molecular, and Optical Physics at Harvard University and the
  Smithsonian Astrophysical Observatory.  AWH was funded by NSF grants
  CCF-1452616, CCF-1729369, PHY-1818914 and ARO contract
  W911NF-17-1-0433. 
\end{acknowledgments}

% To switch for two-column to one-column formatting here, using the command 
\onecolumngrid

% Number subsections in the appendices as in the main text,
% except skip the capital Roman numerals.
\renewcommand{\thesection}{\Alph{section}}
\renewcommand{\thesubsection}{\Alph{section} \arabic{subsection}}
\renewcommand{\thesubsubsection}{\Alph{section} \arabic{subsection} \roman{subsubsection}}

% Label the equations in Appendix L as L1, L2, ...
\makeatletter\@addtoreset{equation}{section}
\def\theequation{\thesection\arabic{equation}}

\titleformat{\section}{\centering\bf}{Supplementary Note~\thesection:}{0.3em}{\textbf}

\begin{appendices}

\section{Comparisons with other Bell inequalities}
\label{App_Comparison}

This supplementary note compares the nonlinear Bell inequality for macroscopic measurements (Theorem~\ref{thm_marco_bell_ineq})
with other Bell inequalities, in three ways.
Supplementary Note~\ref{app_Comparison_Philosophy} contrasts our inequality with
a range of other Bell inequalities for large-scale systems.
The most familiar Bell inequality, the CHSH inequality, 
governs microscopic systems.
Since the CHSH inequality is linear,
Supplementary Note~\ref{subsec:Nonlinearity} explains how our Bell inequality is nonlinear.
Supplementary Note~\ref{subsec:Reduction} describes how
our Bell inequality reduces to the microscopic CHSH inequality
when evaluated on microscopic systems 
whose observables have certain eigenvalues.

\subsection{Comparison of the nonlinear Bell inequality for macroscopic measurements \\ with other large-scale Bell inequalities}
\label{app_Comparison_Philosophy}

From a theoretical standpoint, quantum mechanics allows for a wide range of macroscopic Bell tests. As observed in \cite{reid2001proposal,reid2001new}, one can scale up a microscopic Bell test by taking a cue from Schr\"{o}dinger's cat: One begins with a microscopic Bell state.  Each subsystem is entangled with a macroscopic system. The macroscopic systems are then evolved unitarily and are measured.  The result is a Bell test performed with macroscopic measurements on macroscopic entangled systems. This strategy is impractical: Entangling a microscopic system with a particular macroscopic system is difficult, and % macroscopic systems decohere quickly. 
large-scale entanglement decoheres quickly.
%%% NYH: ``Macroscopic systems decohere quickly'' is a strike against the practicality of all macroscopic Bell tests, including ours.
Multiple solutions to these challenges have been proposed. 
In some works, many-particle systems or far-apart systems are entangled.
However, the degrees of freedom that violate a Bell inequality are microscopic~\cite{Marinkovic_18_Optomechanical,Yao_09_Quantum}.
Other large-scale Bell inequalities require states (commonly GHZ states) 
that have yet to be prepared at the macroscale \cite{Mermin_90_Extreme,Ardehali_92_Bell,Belinskii_93_Interference,Cavalcanti_07_Bell,Cavalcanti_10_Macroscopically,He_10_Bell} 
or measurements that may be hard to realize \cite{Tura_14_Detecting,Collins_02_Bell,Dalton_19_CGLMP}.
Experimental demonstrations of Bell tests based on the latter approach~\cite{Tura_14_Detecting} have been attempted
but have been limited to witness inequalities.
(A witness inequality relies on a quantum mechanical description of the state 
and correct calibration of the measurement apparatus \cite{Schmid_16_Bell,Engelsen_17_Bell}.)
More-concrete, experiment-centered proposals focus on finding systems that 
can support macroscopic entangled states 
and measurements that violate Bell's original inequality \cite{jeong2009failure,thenabadu2019testing}. 
While technology is approaching the point at which these tests may be feasible, 
they have yet to be demonstrated experimentally.
For other large-scale Bell inequalities,
the amount by which quantum states violate the inequalities 
vanishes (a violation becomes exceptionally unlikely) in the large-system limit~\cite{Mermin_80_Quantum,Reid_02_Violation,Drummond_83_Violations}.

We adopt a fundamentally different approach, inspired by resource theories.
\emph{Resource theories} are simple, general models,
developed in quantum information theory,
for situations in which some operations are easy to perform
and others are impractical~\cite{Chitambar_19_Quantum}.
% States form a hierarchy, one can prove, defined by which states can transform into which under the allowed operations.
Resource theories have proven critical for understanding
quantum information, from entanglement to quantum error correction, channels, thermodynamics, and more. 
Our resource-theorety-inspired approach manifests as follows.
Instead of wrangling a system into supporting macroscopic entanglement,
we ask which operations a macroscopic experimentalist can perform fairly easily. 
After characterizing these operations, we derive a Bell inequality 
violable by these operations 
and identify conditions under which a violation implies nonclassicality.
The operations are 
(i) the preparation of % many copies of 
small-scale entangled states and 
(ii) measurements of coarse-grained properties of 
macroscopic systems.
% the ensemble of states. 
These conditions limit the experimentalist's ability 
as much as possible while still letting him/her create and measure 
a bipartite macroscopic state that contains some entanglement:
Macroscopic entangled states are difficult to prepare and control,
requiring the limitation (i).
Given the ability to prepare microscopic entanglement,
if the experimentalist could measure microscopic observables,
s/he could violate the microscopic Bell inequality,
obviating the need for a macroscopic Bell inequality.
Hence the limitation (ii).\footnote{
% < f >
Different operations are ascribed to the classical experimentalists
in the most famous resource theory,
the resource theory for pure bipartite entanglement~\cite{Horodecki_09_Quantum}.
There, experimentalists can perform 
local operations and classical communication (LOCC).
Why do we not restrict our experimentalist to LOCC?
Using LOCC, experimentalists can communicate.
Communication is prohibited during a Bell test.
As the goal differs between the Bell test and the resource-theory agent,
so do the allowable operations differ.}
% < /f >
An experimentalist subject to our restrictions cannot violate 
the large-scale Bell inequalities described in this supplementary note's first paragraph.
A violation of the macroscopic Bell inequality implies nonclassicality
if the system satisfies the conditions in Sec.~\ref{sec_Setup}.
Supplementary Note~\ref{app_Feasibility} elaborates on 
the conditions' experimental feasibility.

Applying a resource-theory approach to Bell tests 
is valuable for two reasons.
First, the approach begins with states and measurements
that are inherently amenable to experimental realization.
The macroscopic Bell test is therefore simple, natural,
and implementable with tools present in many labs
(e.g., Supplementary Note~\ref{app_Noise_Model}).\footnote{
% < f >
Of course, realizing our Bell test with any particular platform
requires that our abstract theory be tailored to the platform.
The tailoring may require experiment-specific modifications and assumptions.
See Supplementary Note~\ref{app_Noise_Model} and~\cite{Wong_Experiment} for examples.} 
% < /f >
Second, the resource-theory approach highlights 
how surprising our result is:
One should not expect a macroscopic % , fairly classical 
experimentalist to violate a Bell test, even in the absence of noise.
For example, a similar experimentalist can perform
quantum computing with liquid-state nuclear magnetic resonance (NMR):
NMR degrees of freedom (nuclear spins) are mostly independent,
and measurements are coarse-grained,
though microscopic entangled states cannot be prepared.
NMR with up to $\approx 12$ spins can be modeled
with a classical local hidden-variables theory~\cite{Menicucci_02_Local}.
That a similarly restricted experimentalist 
can violate a Bell inequality is therefore surprising.
In another example, the restrictions on our state and measurements
have been conjectured to preclude
violations of the ordinary Bell inequality~\cite{Navascues_10_Glance}.
Adopting a resource-theoretic approach, therefore,
leads to a macroscopic Bell inequality 
whose existence is not obvious \emph{a priori}.
%%% NYH: I agree that the following statements are true and that they have some relationship with the foregoing sentences. However, the following sentences seem to me not to be very relevant to the appendix's main point. So, since we present the following ideas elsewhere, I think that the appendix would be most coherent if it didn't contain the following lines.
% Indeed we are only able to construct a Bell test at the cost of an additional assumption; our bound requires the same independence assumption holds on classical systems as we place on our macroscopic system. Our constructed Bell test is also non-linear, a restriction we elaborate on in Supplementary Note \ref{subsec:Nonlinearity}.

% Rather than describing a particular quantum system that supports macroscopic entangled states, we begin by describing a large class of ``near-classical" quantum states and measurements, and construct a Bell test using these. 
%% NYH: The term ``near-classical'' confused me.

A paper related to ours merits close comparison~\cite{navascues2013testing}. In Sec.~6 of~\cite{navascues2013testing}, the authors consider a setup similar to ours, including macroscopic measurements of a quantum system. 
Certain outcomes, they show, cannot be explained by 
a classical model in which 
measuring the microscopic systems yields discrete outcomes.
The inequality they derive has a similar functional form to ours at $\epsilon = 0$, and our inequality reduces to theirs in the special case that $\epsilon = 0$ and that microscopic measurement outcomes take discrete values $0$ or $1$.
Their approach and ours differ in three ways.
First, the initial assumptions differ:
They assume that measurement outcomes are discrete,
whereas we assume microscopic measurements give a value in the interval $[0,1]$.
Second, our macroscopic Bell inequality accommodates
small experimental imperfections:
The $\epsilon$-dependent terms render our inequality 
robust with respect to bounded errors.
We expect that, using our strategy, one can imbue 
the bound in~\cite{navascues2013testing} with similar robustness.
Third,~\cite{navascues2013testing} adopts a foundational approach,
whereas we emphasize experimental platforms.
Reference~\cite{navascues2013testing} occupies a corpus of works
about replacing the postulates of quantum theory.
% with physically motivated principles.
Quantum theory rests on \emph{ad hoc} rules,
such as the association of physical systems with Hilbert spaces.
Replacing these rules with physically motivated principles
forms a subfield of  the foundations of quantum theory~\cite{Hardy_01_Quantum}.
Such principles include the impossibility of signaling more quickly
than light can travel.
Navascu\'es \emph{et al.} have proposed 
the \emph{principle of macroscopic locality},
the ability of local hidden-variables theories to recapitulate
all macroscopic measurements' outcome statistics~\cite{Navascues_10_Glance,Yang_11_Quantum,navascues2013testing,Navascues_16_Macroscopic}.
In contrast with this foundational approach, 
we focus more on experiments and applications:
Supplementary Note~\ref{app_Noise_Model} details 
how our Bell test can be realized with photonic SPDC, and
Supplementary Note~\ref{app_Posner} sketches a Posner-molecule implementation.
Furthermore, as stated in the introduction, 
our goal is practical: to certify entanglement 
under reasonable experimental assumptions.

Another inequality relevant to macroscopic systems,
and violated in quantum theory,
is the Leggett-Garg inequality~\cite{Leggett_85_Quantum}.
The Leggett-Garg inequality codifies the principle of macroscopic realism,
which consists of two parts:
First, every macroscopic object occupies a definite state.
Second, one can ascertain the state via 
a measurement that fails to disturb the state.
A Leggett-Garg inequality bounds correlations 
between events spread across time.
In contrast, a Bell inequality (such as 
the nonlinear Bell inequality for macroscopic measurements)
bounds correlations between objects spread across space.

\subsection{Nonlinearity of the macroscopic Bell test}
\label{subsec:Nonlinearity}

Theorem~\ref{thm_marco_bell_ineq} is said to be a nonlinear Bell inequality.
Here, we clarify the nonlinearity's nature.

Let $A_0$, $A_1$, $B_0$, and $B_1$ denote random variables
that represent outcomes of measurements of quantum systems.
For $x = 0, 1$, let $p_{A_x}(a)$ denote the probability that $A_x = a$.
Define $p_{B_y}(b)$, for $y = 0, 1$, analogously. 
Let $p_{A_x,B_y}(a,b)$ denote the joint probability that  
$A_x = a$ and $B_y = b$. As shorthand, we denote the set of all such joint distributions
$p := \{ p_{A_x,B_y}(a,b) \}_{x, y = 0, 1}$.
[To ensure consistency with quantum theory,
we assume that some joint distributions do not exist, e.g.,
$p_{A_0,A_1}(a_0, a_1)$.
This technicality does not affect the main discussion.]

Consider measuring quantum observables in each of many trials,
averaging over outcomes, and combining the averages.
The combination is a function of 
% the distributions $p$ coming from these measurements.
the set $p$ of joint distributions over
measurement pairs' possible outcomes.
Examples include the traditional Bell parameter,
\begin{align}
    \label{eq_Trad_Bell}
    \tradBell(p) 
    & := \mathbb{E}_p 
    \left(A_0B_0 + A_0B_1 + A_1B_0 - A_1B_1 \right) \\
    &= \sum_{a,b} p_{A_0, B_0}(a,b) ab  
    + \sum_{a,b} p_{A_0, B_1}(a,b) ab  
    + \sum_{a,b} p_{A_1, B_0}(a,b) ab  
    - \sum_{a,b} p_{A_1, B_1}(a,b) ab .
\end{align}
Critically, the Bell parameter depends on 
each underlying probability distribution linearly:
Let $p'_{A_x}(a)$, $p'_{B_y}(b)$, and $p'_{A_x, B_y} (a, b)$
denote probability distributions defined analogously to
$p_{A_x} (a)$, etc. 
Define the set $p'$ of distributions analogously to $p$. 
Let $w \in [0, 1]$ denote an independent probability weight.
Define the set $q$ as consisting of mixtures of 
the distributions in $p$ and $p'$:
\begin{align}
    \label{eq_Def_q}
    q_{A_x,B_y}(a,b) 
    := w p_{A_x,B_y}(a,b) 
    +  (1 - w) p'_{A_x,B_y}(a,b)   
\end{align}
for all $x,y \in \{0,1\}$. The traditional Bell parameter~\eqref{eq_Trad_Bell} is linear in that
\begin{align}
    \label{eq_TradBell_Linear}
    \tradBell(q) = w \tradBell(p) + (1 - w) \tradBell(p') .
\end{align}
This equation can be verified by direction computation and follows from the expectation value's linearity. 

Like the traditional Bell parameter, 
the macroscopic Bell parameter~\eqref{eq_Bell_Param} 
is a function of the underlying probability distributions:
\begin{align}
    \Bell(p) & := \frac{4}{N}         
     \text{Cov}_p  \left({A_{0}}, {B_0}\right) 
    +\text{Cov}_p \left({A_0  },{B_1}\right)
    + \text{Cov}_p\left({A_1},{B_0}\right) 
    - \text{Cov}_p\left({A_1},{B_1}\right) ] \\
    &= \frac{4}{N} \big[
        \mathbb{E}_p  (A_0 B_0) 
        - \mathbb{E}_p(A_0)  \mathbb{E}_p(B_0) 
        +\mathbb{E}_p  (A_0 B_1) 
        - \mathbb{E}_p(A_0)  \mathbb{E}_p(B_1)  
        \\ \nonumber &\hspace{45pt} 
        +\mathbb{E}_p (A_1 B_0) 
        - \mathbb{E}_p(A_1)  \mathbb{E}_p(B_0)
        -\mathbb{E}_p (A_1B_1) 
        + \mathbb{E}_p(A_1)\mathbb{E}_p(B_1) 
        \big]\\
        \label{eq_Nonlin_Bell_Expand}
        &= \frac{4}{N} \Bigg[ 
        \sum_{a,b}p_{A_0, B_0}(a,b)ab 
        - \sum_a p_{A_0}(a) a \sum_b p_{B_0}(b)b  
        + \sum_{a,b}p_{A_0, B_1 }(a,b)ab 
        - \sum_a p_{A_0 }(a) a \sum_b p_{B_1}(b)b  \\
        &\hspace{35pt} \nonumber
        + \sum_{a,b}p_{A_1, B_0}(a,b)ab 
        - \sum_a p_{A_1}(a) a \sum_b p_{B_0}(b)b  
        - \sum_{a,b}p_{A_1, B_1}(a,b)ab 
        + \sum_a p_{A_1}(a) a \sum_b p_{B_1}(b)b  \Bigg] .
\end{align}
$\Bell$ depends on the underlying probability distributions 
nonlinearly.
The nonlinearity comes from the terms quadratic in $p_{A_x,B_y}$ 
in Eq.~\eqref{eq_Nonlin_Bell_Expand}. 
Define $p'$, $q$, and $w$ as in and around Eq.~\eqref{eq_Def_q}.
In contrast with the traditional Bell parameter 
[Eq.~\eqref{eq_TradBell_Linear}],
\begin{align}
    \Bell(q) \neq  w \Bell(p) + (1 - w) \Bell(p') .
\end{align}
The inequality in \Cref{thm_marco_bell_ineq} is a bound on $\Bell$, 
so we call the result a nonlinear Bell inequality. 

This nonlinearity impacts experiments as follows. 
Suppose that an experimentalist wishes to observe
a violation of the traditional Bell inequality.
The experimentalist performs measurements
over some time interval,
computes the average in the Bell parameter~\eqref{eq_Trad_Bell},
and checks that the value exceeds the classical bound
with the desired certainty.
The measurement outcomes are drawn according to
underlying probability distributions
that may shift over time.
But, if the Bell inequality is violated,
at least some of those distributions were nonclassical;
the system had, at some time, 
nonclassical measurement-outcome statistics.
Consider, in contrast, violating the macroscopic Bell inequality.
The experimentalist must perform enough measurements
that the $\Bell$ estimate violates the inequality with sufficient confidence.
Furthermore, the measurements' outcomes
must be drawn from nearly the same underlying distributions.
(The ``nearly'' comes from the inequality's $\epsilon$ terms.)
This second requirement stems from the nonlinearity of $\Bell$:
Violating the macroscopic Bell inequality with 
a mixture $q = w p + (1 - w) p'$ of distributions
does not guarantee that any particular distribution $p$ from the mixture
violates the inequality.
Relatedly, violating the nonlinear inequality 
cannot falsify hidden variables that vary with time.
However, a nonlinear Bell inequality can falsify theories
that involve time-independent local hidden variables.
Furthermore, a nonlinear inequality can certify entanglement
under realistic experimental assumptions.

\subsection{Relationship with the microscopic Bell inequality}
\label{subsec:Reduction}

The macroscopic Bell parameter $\Bell$ [Eq.~\eqref{eq_Bell_Param}]
reaches $2 \sqrt{2}$ when evaluated on a product of singlets
[Theorem~\ref{thm_quantum_violation}].
So does the traditional Bell parameter $\tradBell$ [Eq.~\eqref{eq_Trad_Bell}]
reach $2 \sqrt{2}$ when evaluated on a singlet.
We explain the parallel here:
$\Bell$ reduces to $\tradBell$ under certain conditions.

% The macroscopic Bell parameter and the traditional Bell inequality saturate at the same value, $2\sqrt{2}$. 
%%% NYH: We haven't proved that the greatest value of $\Bell$ achievable with quantum systems is $2 \sqrt{2}$. Also, the statement that we aim to make isn't that the inequality saturates at some value. An inequality saturates when the LHS equals the RHS. 

The macroscopic Bell parameter $\Bell$ was designed 
to accommodate many-body systems.
However, $\Bell$ can be evaluated on $N = 1$ of 
the microscopic systems described in Sec.~\ref{sec_Setup}.
In this case,\footnote{
$\Bell$ assumes the following value also when evaluated on
a product of identical microscopic systems.
The reason is, each total variance equals
a sum of independent variables' covariances.}
\begin{align}
    \label{eq_Macro_Reduce1}
    4[\cov{a_0}{b_0} + \cov{a_0}{b_1} + \cov{a_1}{b_0} - \cov{a_1}{b_1}] .
\end{align}
%%% NYH: In the main text, we defined A_0, A_1, B_0, and B_1 as sums of microscopic random variables a_0^(j), etc. These microscopic variables were defined as lying between 0 and 1. So saying, here, that A_0, etc. lie between 0 and 1 could cause confusion, especially if the foregoing footnote is to be true.
The random variables $a_0$, $a_1$, $b_0$ and $b_1$ 
assume values in $[0,1]$. 
Define random variables $a'_0$, $a'_1$, $b'_0$, and $b'_1$
by subtracting $1/2$ from each original variable
and multiplying by 2. For example, $a'_0 := 2(a_0 - 1/2)$.
The subtraction preserves the covariances, 
and the macroscopic Bell parameter~\eqref{eq_Macro_Reduce1} becomes
\begin{align}
    &\cov{2a_0}{2b_0} + \cov{2a_0}{2b_1} 
      + \cov{2a_1}{2b_0} - \cov{2a_1}{2a_1} \\
    \label{eq_Macro_Reduce2}
    &\hspace{40pt}= \cov{a_0'}{b_0'} + \cov{a_0'}{b_1'} + \cov{a_1'}{b_0'} - \cov{a_1'}{a_1'} .
\end{align}

Suppose that each of $a_0'$, $a_1'$, $b_0'$, and $b_1'$ averages to 0. 
This criterion is satisfied by 
the random variables defined by the outcomes of
the traditional Bell test's measurements.
The macroscopic Bell parameter~\eqref{eq_Macro_Reduce2} reduces to
\begin{align}
    &\E{a_0' b_0'} - \E{a_0'}\E{b_0'} 
    + \E{a_0'b_1'} - \E{a_0'}\E{b_1'} \nonumber  \\
    &\hspace{60pt}
    + \E{a_1' b_0'} - \E{a_1'}\E{b_0'}  
    - \E{a_1'b_1'} + \E{a_1'}\E{b_1'} \\
    &\hspace{40pt}
    = \E{a_0' b_0' + a_0' b_1' + a_1' b_0' - a_1' b_1'},
\end{align}
the original Bell parameter. 

The above argument, coupled with Tsirelson's bound~\cite{cirel1980quantum}, 
implies that the macroscopic Bell parameter, $\Bell$, maximizes at $2\sqrt{2}$ 
if each microscopic measurement
(of which the macroscopic measurement is a coarse-graining)
is equally likely to output a 0 or a 1.
% when measurements are made on a quantum system whose microscopic subsystems are equally likelly to report $0$ or $1$ outcomes. 
%%% NYH: Discussing microscopic subsystems that report outcomes is ontological. I've exchanged the ontological language for operational language. Strictly speaking, a foundational thinker might protest, we don't know that microscopic subsystems exist or report anything. But we can observe outcomes of measurements.
% The analysis in \cite{Pozsgay_17_Covariance} gives a stronger version of this result, 
%%% NYH: The line above suggests that Pozsgay et al. analyze our macroscopic Bell parameter.
This result can be strengthened via a technique in~\cite{Pozsgay_17_Covariance}:
$\Bell$ maximizes, at $2 \sqrt{2}$,
whenever the measurements are performed on 
a quantum system that satisfies the assumptions of Sec.~\ref{sec_Setup}.

\section{Experimental Feasibility} 
\label{app_Feasibility}

The proof of Theorem~\ref{thm_quantum_violation} involves
a product of $N$ independent singlets (and can be extended to $N$ copies of other entangled states).
This state, we proved, (i) satisfies the assumptions in Sec.~\ref{sec_Setup}
and (ii) violates the macroscopic Bell inequality.
Here, we address three possible concerns about preparing 
such a state, or a similar state, experimentally.

%These assumptions were designed to mimic ``near-classical" quantum states already present in nature. Here we address some potential concerns about the feasibility of reproducing these states in the laberatory.

First, macroscopically many singlets must be prepared and measured.
One might suppose that the singlets must be stored, for some time,
in a macroscopically large quantum memory. But not all singlets need to be prepared simultaneously. One can produce many singlets in sequence, then measure for some time interval (Supplementary Note~\ref{app_Noise_Model}). This strategy requires no memory.
% Quantum memories, including photonic memories, 
% are under development~\cite{teja2019photonic,heshami2016quantum,brennen2015focus}.
% However, constructing a quantum memory large enough to store macroscopically many singlets may appear daunting. 
%Fortunately, our singlets need not all be stored in a common memory.
Moreover, even if singlets are produced simultaneously, they do not need to be stored in a common quantum memory. In Supplementary Note~\ref{app_Posner} we discuss a possible violation of the macroscopic Bell inequality using Posner molecules. 
In that approach, $N$ pairs of Posner molecules act as $N$ independent (small) quantum memories. 
%The Posners are then measured collectively. 
This allows for simultaneous storage of $N$ singlets, without requiring a macroscopically large quantum memory. %In this possible violation singlets are produced simultaneously, but each singlet is stored independently in an pair of biomolecules. }

%And in some architectures, singlets can be created and measured quickly enough that it negates the need for a quantum memory altogether. If a SPDC source is driven with a high enough intensity laser, signal and idler beams will both contain a macrsocopic number of photons -- exactly the setup required for our inequality. 

Second, one might question how easily singlets can be prepared.
Singlets have been prepared with photons~\cite{Barz_10_Heralded}, 
solid-state systems~\cite{Bernien_13_Heralded},
atoms~\cite{Laurat_07_Heralded,Hofmann_12_Heralded},
and trapped ions~\cite{Casabone_13_Heralded}.
Furthermore, other states can violate 
the macroscopic Bell inequality~\eqref{eq_Our_Bell_Ineq}.
% :
% Many copies of any highly entangled state should suffice.
% \q{[What's ``should suffice'' supposed to mean? 
% Many copies of any highly entangled state suffice?
% Or we expect them to suffice, but we're not sure?]}
The measurements would need to be tailored to the state,
and the tailored measurements' experimental feasibility
would need to be checked.
Supplementary Note~\ref{app_Posner} illustrates that 
nonsinglet states can violate the inequality:
We consider performing a macroscopic Bell test with Posner molecules. Under the assumptions in~\cite{Fisher_15_Quantum},
a product of singlets cannot be measured
as described in the proof of \Cref{thm_quantum_violation}. 
We therefore show how to prepare an entangled Posner state
that behaves like a singlet in certain ways.
We then identify feasible (according to the model in~\cite{Fisher_15_Quantum}) measurements that violate the macroscopic Bell inequality.

Third, in Sec.~\ref{sec_Setup}, each pair of particles
is initially assumed to be independent of the other particles.
In natural macroscopic systems, microscopic subsystems tend to interact.
The independence assumption may therefore seem artificial.
The resolution to this quandary has three components.
First, our results accommodate small violations
of the independence assumption,
as quantified with the $A'_x$ and $B'_y$ in Sec.~\ref{sec_Setup}.
Second, even while some microscopic degrees of freedom interact,
other degrees of freedom can remain isolated.
Examples include the polarizations of photons produced by SPDC
(Supplementary Note~\ref{app_Noise_Model})
and the nuclear spins in liquid-state NMR.
Third, the macroscopic Bell inequality may be modified
as described at the end of Sec.~\ref{sec_Discussion},
to govern particles that interact in a limited, known manner.
Supplementary Note~\ref{app_Noise_Model} presents a similar modification:
There, the macroscopic system consists of
a Poisson-distributed number of microscopic systems.
Consider modeling this state with the
fixed-number-of-singlets model in Sec.~\ref{sec_Setup}.
The noise variable $r_{A_x}$ [Eq.~\eqref{eq_Error_Model}] 
would have 
%too large 
a variance [Eq.~\eqref{eq_Bound_Var}]
large enough to preclude a macroscopic Bell test.
We resolve this problem by incorporating the Poisson distribution
into our assumptions, deriving a new Bell inequality.
If a system violates the new Bell inequality while satisfying the assumptions,
the system is nonclassical.
This resolution will be detailed further
when tailored to an experimental platform in the upcoming experimental collaboration~\cite{Wong_Experiment}.

\section{Example noise model, for spontaneous parametric down-conversion of photons}
\label{app_Noise_Model}

The bounds presented in \Cref{thm_marco_bell_ineq,thm_quantum_violation} are worst-case bounds. They hold for any noise that satisfies the variance bounds of Eq.~\ref{eq_Bound_Var}. However, experimental assumptions can often constrain noise further. 
A noise-specific analysis can lead to bounds that 
separate classical from nonclassical correlations
when the general bounds cannot.
We illustrate with noise that acts on the microscopic random variables independently. 

For concreteness, we analyze errors that occur when photon beams are produced via SPDC (Sec.~\ref{sec_Setup}). We do this for two reasons. First, we hope to demonstrate that a macroscopic Bell test is physically viable. Second, ideas in this analysis may generalize to other physical setups. This section will provide a template for device-specific noise analysis. 
A more sophisticated version of this analysis,
tailored to a platform used in SPDC experiments~\cite{kim2008applications}, 
will appear in an upcoming experimental paper~\cite{Wong_Experiment}.

We begin by reviewing the setup. Photon beams are produced when a laser shines on a nonlinear crystal. The crystal down-converts some fraction of the incident photons:
Upon absorbing one photon, the crystal emits two. 
The two photons travel in different paths,
and their polarizations become maximally entangled.
If this process occurs frequently enough,\footnote{
% < f >
Coincidence rates of $\approx 10$ per second were reported in~\cite{kwiat1995new}.}
% < /f >
two distinct beams of photons, whose polarizations form Bell states, result.

This process can involve two sources of randomness.
First, imagine placing a perfect-efficiency detector
right next to the crystal, in the path of one of the beams.
The detector's clicking rate is expected to obey a Poisson distribution.
The distribution governs the number of photon pairs 
produced by the crystal per unit time.
This number is a random variable 
whose randomness is global and classical, 
or is of type~\ref{item_Randomness_iii} (Sec.~\ref{sec_Setup}).
The Poisson randomness could be removed via heralded photon pairs~\cite{Barz_10_Heralded,pittman2003heralded,wagenknecht2010experimental,wang2020dipole}, though we do not pursue heralding here.
Second, photons can be lost between
being produced in the crystal and being measured by 
Alice's or Bob's post-polarization detector.
Dust on the polarizer could absorb a photon, for example,
or the detector could have subunit efficiency.
This randomness is local and classical, or of type~\ref{item_Randomness_ii}.
The randomness prevents the photons from 
being in a product of singlets.
For example, the Poisson distribution makes the state mixed.
We therefore cannot directly apply the proof of Theorem~\ref{thm_quantum_violation},
to conclude that the photons will violate Theorem~\ref{thm_marco_bell_ineq}.
Rather, we model the randomness of both types.
Using the model, we prove tighter analogs of \Cref{thm_marco_bell_ineq,thm_quantum_violation}.
If an experimental system violates the tighter macroscopic Bell inequality
while obeying the assumptions in the model,
the system is nonclassical.

\subsection{Model}
\label{app_Photon_Model}

Our first step in tightening the classical bound (\Cref{thm_marco_bell_ineq}) is to identify the most general form that the macroscopic random variables $A_0$, $A_1$, $B_0$, and $B_1$ can assume. 
Let $M$ be the number of photons in the laser beam that strikes the crystal (per unit time). We call these ``incident photons.'' Assume that each incident photon down-converts with probability $\lambda$, independently of the other possible down-conversions.\footnote{
We are assuming that the possible down-conversions are independent. This assumption can be approximately satisfied if the time scale over which global parameters change $\gg$ the down-conversion time scale. Small deviations from this assumption can be accommodated with the bound in Supplementary Note~\ref{app_Proof_Bell_ineq}.} 
The total number of photon pairs produced is represented by a random variable
$\sum_{i=1}^M  e_i$.
The $e_i$'s are independent Bernoulli random variables,
each with mean $\lambda$. 
Down-conversion is improbable~\cite{Bock_16_Highly}, 
% Reference for the foregoing line: “The conversion efficiency of SPDC is typically very low, with the highest efficiency obtained is on the order of 4 pairs per 106 incoming photons for PPLN in waveguides. [3]” 
%%% (https://en.wikipedia.org/wiki/Spontaneous_parametric_down-conversion#cite_note-3) 
%%% Ref. 3 in the Wikipedia article: https://www.osapublishing.org/oe/abstract.cfm?uri=oe-24-21-23992 
so $\lambda$ is small: $\lambda \ll 1$.

The random variable $a_x^{(i)}$ describes the value that would be reported 
if the $i^\th$ incident photon participated in a down-conversion event
and the resultant photon in Alice's beam were measured with measurement setting $x$. 
The random variable $b_y^{(i)}$ is defined analogously. 
The $i^\th$ possible down-conversion event can add a photon to Alice's beam.
Suppose that Alice measures with setting $x$.
Whether that photon is lost before detection
is represented by $l_{a,x}^{(i)}$. 
This random variable equals $0$ if the photon is lost and equals $1$ otherwise. 
The random variable $l_{b,y}^{(i)}$ is defined analogously.
These variables govern the macroscopic random variables:
\begin{align}
    A_x = \sum_{i=1}^M e_i a_x^{(i)} l_{a,x}^{(i)} \, ,
    \quad \text{and} \quad
    B_y = \sum_{i=1}^M e_i b_y^{(i)} l_{b,y}^{(i)} .
\end{align}

How the total particle number, $N$, should be defined is ambiguous.
Several possibilities suggest themselves. We choose a definition that leads to strong bounds: In the ideal quantum experiment in Sec.~\ref{sec_Bell_Ineq}, each microscopic random variable has a probability $1/2$ of reporting 1 and a probability $1/2$ of reporting 0. Hence 
$\E{A_0} = \E{A_1} = \E{B_0} = \E{B_1} = N/2$. 
We turn this observation into a definition:
\begin{align}
    N = \E{A_0} + \E{B_0} . \label{eq_N_alt_def} 
\end{align}

Proceeding from definitions to bounds, we compute 
the macroscopic random variables' covariances. 
For all $x, y \in \{0, 1\}$,
\begin{align}
    \cov{A_x}{B_y} 
    &= \sum_{i=1}^M
    \cov{e_i a_x^{(i)} l_{a,x}^{(i)}}{e_i b_y^{(i)} l_{b,y}^{(i)}} \\
    % % %
    &= \sum_{i=1}^M  \left[
    \E{e_i  a_x^{(i)} l_{a,x}^{(i)} b_y^{(i)} l_{b,y}^{(i)}} 
    -\E{e_i  a_x^{(i)} l_{a,x}^{(i)}}\E{e_i b_y^{(i)} l_{b,y}^{(i)}}  \right] \\
    % % %
    &= \sum_{i=1}^M   \left[
    \lambda \E{a_x^{(i)} l_{a,x}^{(i)} b_y^{(i)} l_{b,y}^{(i)}} 
    - \lambda^2 \E{a_x^{(i)} l_{a,x}^{(i)}}\E{b_y^{(i)} l_{b,y}^{(i)}}  \right] \\
    % % %
    &= \lambda  \sum_{i=1}^M  
    \E{a_x^{(i)} l_{a,x}^{(i)} b_y^{(i)} l_{b,y}^{(i)}}
    + O( M \lambda^2)   .
    \label{eq_small_lambda_bound}
\end{align}
The macroscopic random variables have averages of the form
\begin{align}
\E{A_x} = \sum_{i=1}^M \E{ e_i a_x^{(i)} l_{a,x}^{(i)}} 
= \lambda\sum_{i=1}^M \E{ a_x^{(i)} l_{a,x}^{(i)}}.
\end{align} 
Substituting into Eq.~\eqref{eq_Bell_Param}, we form the macroscopic Bell parameter:
\begin{align}
    \label{eq_Bell_Ineq_Noise_App}
    \Bell(A_0,A_1,B_0,B_1) 
    = \frac{4 \sum_{i=1}^M \left[
    \lambda \E{a_0^{(i)} 
    l_{a,0}^{(i)} b_0^{(i)} l_{b,0}^{(i)}
    +  a_0^{(i)} l_{a,0}^{(i)} b_1^{(i)} l_{b,1}^{(i)}
    +  a_1^{(i)} l_{a,1}^{(i)} b_0^{(i)} l_{b,0}^{(i)}
    -  a_1^{(i)} l_{a,1}^{(i)} b_1^{(i)} l_{b,1}^{(i)}}\right] 
    + O( M \lambda^2)  }{
    \lambda  \sum_{i=1}^M 
    \E{ a_0^{(i)} l_{a,0}^{(i)}   +  b_0^{(i)} l_{b,0}^{(i)} }  } .
\end{align}
We will bound the RHS under the assumption that the microscopic variables are classical. Then, we will show that the bound can be violated with a quantum state.

\subsection{Classical bound}
 
We bound the numerator of Ineq.~\eqref{eq_Bell_Ineq_Noise_App}
using a general inequality.
For any variables $a_0, a_1, b_0, b_1  \in  \{0,1\}$,
\begin{align} 
   \label{eq_Class_Num_Help}
   a_0 b_0 + a_0 b_1 + a_1 b_0 - a_1 b_1 
   \leq a_0 + b_0.
\end{align} 
Let $a_x = a_x^{(i)}  l_{a,x}^{(i)}$ and
$b_y = b_y^{(i)}  l_{b, y}^{(i)}$.
Applying Ineq.~\eqref{eq_Class_Num_Help} to the numerator of
Ineq.~\eqref{eq_Bell_Ineq_Noise_App} yields
\begin{align}
\E{a_0^{(i)} l_{a,0}^{(i)} b_0^{(i)} l_{b,0}^{(i)}
+  a_0^{(i)} l_{a,0}^{(i)} b_1^{(i)} l_{b,1}^{(i)}
+  a_1^{(i)} l_{a,1}^{(i)} b_0^{(i)} l_{b,0}^{(i)}
-  a_1^{(i)} l_{a,1}^{(i)} b_1^{(i)} l_{b,1}^{(i)}} 
\leq \E{ a_0^{(i)} l_{a,0}^{(i)}  +  b_0^{(i)} l_{b,0}^{(i)} } .
\end{align}
We substitute into the numerator in Ineq.~\eqref{eq_Bell_Ineq_Noise_App},
then cancel the $\E{ \ldots }$ in the denominator.
The classical macroscopic random variables satisfy\footnote{
% < f >
The correction in Ineq.~\eqref{eq_classical_downconversion} is small when
Alice and Bob measure as dictated in 
the ``Quantum violation of the classical bound'' subsection:
% when the game is played following the quantum strategy of Sec.~\ref{sec_Bell_Ineq}.
The correction decomposes as
$M \lambda^2 / N  =  \left( \frac{M \lambda}{N}  \right)  \lambda$.
The final $\lambda \ll 1$ by assumption.
$M$ denotes the number of photons in the laser beam that hits the crystal,
$\lambda$ denotes the probability that a given laser-beam photon down-converts,
and $N$ denotes the number of photons in Alice or Bob's beam. Hence $N \approx \lambda M$. }
% < /f >
\begin{align} 
    \label{eq_classical_downconversion}
    \Bell (A_0,A_1,B_0,B_1) \leq 4 + O(M\lambda^2/N).
\end{align}
% Reason for the /N: The denominator of Eq.~\eqref{eq_Bell_Ineq_Noise_App} is N by definition of the macroscopic Bell parameter.

We can understand this result as follows:
The randomness in $N$ serves as noise.
It raises the macroscopic Bell bound~\eqref{eq_classical_downconversion}
above the main-text macroscopic Bell bound~\eqref{eq_Our_Bell_Ineq}
and even above the quantum bound~\eqref{eq_Q_Bound_Main}.
In this setting, however, a quantum bound 
lies above the classical~\eqref{eq_classical_downconversion}.
If a system violates Ineq.~\eqref{eq_classical_downconversion}
and satisfies the assumptions in Supplementary Note~\ref{app_Photon_Model},
the system is nonclassical.

\subsection{Quantum violation of the classical bound} 

We can relax our assumptions, because experiments will replace this calculation. Once experimentalists observe covariances that violate Ineq.~\eqref{eq_Our_Bell_Ineq} or Ineq.~\eqref{eq_classical_downconversion}, they can conclude that the particles are nonclassical, if the global correlations are small enough
to be unlikely to have caused the violation. The experimentalists need not worry about precisely why the violation occurred.

We therefore simplify by assuming that 
$l_{a,0}$, $l_{a,1}$, $l_{b,0}$, and $l_{b,1}$ 
obey Bernoulli distributions with the same mean, $\gamma$. 
The macroscopic Bell parameter becomes 
\begin{align} 
        \label{eq_Bell_Param_w_assumptions}
        \Bell(A_0,A_1,B_0,B_1) 
        = \frac{4 \sum_{i=1}^M \left[
        \gamma  \E{a_0^{(i)} b_0^{(i)} +a_0^{(i)} b_1^{(i)} +a_1^{(i)}  b_0^{(i)} -a_1^{(i)}  b_1^{(i)}} 
        + O( \lambda)  \right] }{
        \sum_{i=1}^M \E{ a_0^{(i)} + b_0^{(i)} }} .
\end{align}
If the experimentalists follow the quantum strategy in Sec.~\ref{sec_Bell_Ineq}, the microscopic random variables satisfy [Eq.~\eqref{eq_Q_Help}] 
\begin{align}
\E{a_0^{(i)} b_0^{(i)} +a_0^{(i)} b_1^{(i)} +a_1^{(i)}  b_0^{(i)} -a_1^{(i)}  b_1^{(i)}} 
= 2 \sin^2(3 \pi/8) - 1/2
\end{align}
and [Eq.~\eqref{eq_Q_Help0}]
\begin{align}
\E{ a_x^{(i)}  + b_x^{(i)}  } = 1 .
\end{align}
We substitute into the numerator and denominator of Eq.~\eqref{eq_Bell_Param_w_assumptions}. 
The quantum strategy achieves a macroscopic Bell parameter of
\begin{align}
    \Bell (A_0,A_1,B_0,B_1) 
    =  2 \gamma   [ 4 \sin^2 (3 \pi / 8) - 1]  +  O (\lambda) .
\end{align}
A quantum system can violate the classical bound~\eqref{eq_classical_downconversion} if 
\begin{align}
    \label{eq_Gamma_Condn}
    \gamma > \frac{2  +  O(M\lambda^2/N) }{  
                             4 \sin^2 (3 \pi / 8)  -  1  +  O (\lambda) }
    \approx  0.828 + O (M\lambda^2/N).
\end{align}
Photons can violate the noise-specific macroscopic Bell inequality~\eqref{eq_classical_downconversion} if
$\gtrsim 83\%$ of the photon pairs created arrive at the detectors. 
%\nicole[I strongly suggest finding out the loss rates achieved in experiments and comparing.]}

A similar condition arises in the standard Bell test: 
The standard test suffers from a \emph{detection loophole} if the detector misses too many incident photons. As with the detection loophole, a Bell test remains possible here even if too many photons are lost [even if the system disobeys Ineq.~\eqref{eq_Gamma_Condn}]. Formulating the Bell test would require a more-detailed noise model.
% <-- Details: "We don't make any assumptions about the form when we give the classical bound. If the loss rate was very high, and we enforced that when proving our classical bound (i.e. enforced that l_{ax} was small) then we would get a different (I think lower) bound then the 16/3 that we prove. Assuming the loss rate was the same for all players and measurement settings would also help us (though this might be a questionable assumption!)."
% <-- Reference: email chain "Noise-model app, ctd." --> first message sent by Adam on 10/30/19 

%
%
%
%
%
\section{Proof of nonlinear Bell inequality for macroscopic measurements}
\label{app_Proof_Bell_ineq}

We now prove \Cref{thm_marco_bell_ineq} in full generality,
building on the proof in Sec.~\ref{sec_Bell_Ineq}.
The added analysis introduces robustness against 
weak global classical correlations 
[randomness of type~\ref{item_Randomness_iii}, according to Sec.~\ref{sec_Setup}].

\begin{proof}

First, we review notation.
Second, we bound the observed Bell correlator $\Bell(A_0,A_1,B_0,B_1)$,
using (i) the ideal Bell correlator $\Bell(A'_0,A'_1,B'_0,B'_1)$
and (ii) the bound~\eqref{eq_Bound_Var} on global correlations.

Recall the definitions given in Sec.~\ref{sec_Bell_Ineq}.
$A_0$, $A_1$, $B_0$, and $B_1$ represent the macroscopic random variables observed by the experimentalists. 
$A_0'$, $A_1'$, $B_0'$, and $B_1'$ represent the random variables 
that the experimentalists would measure 
if all global parameters were fixed to their ideal values.
In the photon example, the laser's intensity, the laser-crystal alignment, etc.
would remain constant across trials.
Equation~\eqref{eq_Error_Model} relates
the measured variables to the ideals
via the error variable $r$.
Inequality~\eqref{eq_Bound_Var} bounds the error's variance.

We aim to bound the observed correlator $\Bell(A_0,A_1,B_0,B_1)$ 
in terms of the ideal correlator $\Bell(A_0',A_1',B_0',B_1')$ and $\epsilon$. Algebraic manipulation gives
\begin{align}
   \label{eq_Noisy_Cov1}
   \cov{A_x}{B_y}  & =
   \cov{A_x' + r_{A_x}}{B_y' + r_{B_y}}  \\
   \label{eq_Noisy_Cov1b}
   & = \cov{A_x'}{B_y'} + \cov{A_x'}{  r_{B_y}  } + \cov{  r_{A_x}  }{B_y'} 
   % \hspace{1.3cm} 
   + \cov{r_{A_x}}{r_{B_y}} .
\end{align}
Random variables $X$ and $Y$ have a covariance $\cov{X}{Y}$
bounded in terms of the variables' variances:
Let $\mathcal{X}  :=  X - \expval{X}$ and
$\mathcal{Y}  :=  Y  -  \expval{Y}$.
The original variables have the covariance
\begin{align}
   \cov{X}{Y}
   & =  \E{  \mathcal{X}  \mathcal{Y}  } 
   \leq  \sqrt{  \E{ \mathcal{X}^2 }  \E{ \mathcal{Y}^2 }  } 
   \label{eq_Cov_Var}
   =  \sqrt{  \Var{X}  \Var{Y}  } .
\end{align}
The bound follows from the Cauchy-Schwarz inequality.
We apply Ineq.~\eqref{eq_Cov_Var} to each of
the final three covariances in Eq.~\eqref{eq_Noisy_Cov1b}:
\begin{align}
\abs{\cov{A_x}{B_y} - \cov{A_x'}{B_y'}} 
%% % %
& = \Big| \cov{A_x'}{r_{B_y}} + \cov{  r_{A_x}  }{B_y'} 
 + \cov{r_{A_x}}{r_{B_y}} \Big| \\ % &\hspace{1cm}
& \leq \sqrt{  \Var{A_x'}  \Var{r_{B_y}}  } + \sqrt{\Var{r_{A_x}}\Var{B_y'}} 
+ \sqrt{\Var{r_{A_x}}\Var{r_{B_y}}} \\
% % %  
& \leq (\epsilon + 2\sqrt{\epsilon})N \label{eq_Noisy_Cov3}.
\end{align}
The final inequality follows from Ineq.~\eqref{eq_Bound_Var} and
$\Var{A_x'} \leq N$. 
This latter inequality holds because $A_x'$ equals 
a sum of $N$ independent terms $a_x^{(i)}$.
Each $a_x^{(i)} \in \{0,1\}$ and so has variance $\leq 1$. 
We combine Ineq.~\eqref{eq_Noisy_Cov3} with Eq.~\eqref{eq_Bell_Param}
and the triangle inequality to conclude that
\begin{align}
   \label{eq_Error_Result1}
   \abs{ \Bell(A_0,A_1,B_0,B_1)  -  \Bell(A_0', A_1', B_0', B_1')} 
   \leq 16\epsilon + 32\sqrt{\epsilon}.
\end{align}

%%% Please do not erase the following.
% NYH: We're not missing a factor of four? I'm finding
% [Edit (11/13/19): I haven't changed $x$'s and $y$'s to 0s and 1s
% in the following green text.]
% \begin{align}
%   \abs{ \Bell(A_x,A_y,B_x,B_y)  -  \Bell(A_x', A_y', B_x', B_y')}
%   % % %
%   & =  \frac{4}{N}
%   \abs{  \cov{A_x}{B_x}  +  \cov{A_x}{B_y}  +  \cov{A_y}{B_x}  -  \cov{A_y}{B_y}
%   \\ & \qquad
%   -  \left(  \cov{A'_x}{B'_x}  +  \cov{A'_x}{B'_y}  +  \cov{A'_y}{B'_x}  
%               -  \cov{A'_y}{B'_y}  \right)  }
%   \nonumber \\ & 
%   % % %
%   =  \frac{4}{N}
%   \abs{  \left(  \cov{A_x}{B_x}  -  \cov{A'_x}{B'_x}  \right)
%              +  \left(  \cov{A_x}{B_y}  -  \cov{A'_x}{B'_y}  \right)
%              \\ & \qquad
%              +  \left(  \cov{A_y}{B_x}  -  \cov{A'_y}{B'_x}  \right)
%              - \left(  \cov{A_y}{B_y}  -  \cov{A'_y}{B'_y}  \right)  }
%   \nonumber \\ &
%   % % %
%   \leq  \frac{4}{N}  \big(
%   |  \cov{A_x}{B_x}  -  \cov{A'_x}{B'_x}  |
%   +  |  \cov{A_x}{B_y}  -  \cov{A'_x}{B'_y}  |
%   \\ & \qquad
%   +  |  \cov{A_y}{B_x}  -  \cov{A'_y}{B'_x}  |
%   +  |  \cov{A_y}{B_y}  -  \cov{A'_y}{B'_y}  |   \big)
%   \nonumber \\ &
%   % % %
%   \leq  \frac{4}{N}  \times 4  \times
%   ( \epsilon  +  2  \sqrt{\epsilon}  )  N \\ 
%   % % %
%   & =  16 \epsilon  +  32  \sqrt{\epsilon}  .
% \end{align}
% }
According to the sketch of the proof of \Cref{thm_marco_bell_ineq} 
[Ineq.~\eqref{eq_Bound_Help3}],
\begin{align}
    \label{eq_Bound_Help3_App}
    \Bell(A_0', A_1', B_0', B_1') \leq 16/7.
\end{align}
Combing Ineqs.~\eqref{eq_Error_Result1} and~\eqref{eq_Bound_Help3_App} gives
\begin{align}
   \label{eq_Bound_Help4_App}
    \Bell(A_0, A_1, B_0, B_1) \leq 16/7 + 16\epsilon + 32\sqrt{\epsilon},
\end{align}
the desired result. 
% Reference: email chain "Deleting absolute-value signs" -- messages sent on 11/18/19 and 11/19/19
\end{proof}

\section{Background: CHSH Game}
\label{app_CHSH_Backgrnd}

Before describing the CHSH game, we establish a more general framework for two-party nonlocal games. Nonlocal games illustrate how players given quantum resources can outperform players given only classical resources.
A two-party nonlocal game involves two players, Alice and Bob. They share some resource---typically, classical shared randomness or a quantum state. 
They cannot communicate after agreeing on the strategy they will follow.
The game begins when a verifier sends Alice a question, or symbol, $x$, and sends Bob a question $y$.
Using only the questions and possibly measurements of the shared resource, the players respond with symbols $a$ and $b$. 
The verifier substitutes $x$, $y$, $a$, and $b$ into a function.
If the function's value satisfies some predetermined criterion, the players win the game.

Every nonlocal game has a list of winning response pairs $ab$ 
for every question pair $xy$. 
The players aim to maximize their probability of responding with a winning $ab$, knowing the winning response lists and the distribution from which the questions are drawn. The maximal win probability is called the game's \emph{value}. 
% A game is sometimes specified by its set of question pairs and no explicit distribution. In such cases, questions are assumed to be drawn uniformly randomly.

The CHSH game is described as follows in this language.
Questions $x$ and $y$ are drawn from $\{ 0 , 1 \}$. 
Winning responses are $a, b \in \{0,1\}$ such that 
\begin{align}
   \label{eq_CHSH_Win_Criterion}
   x \wedge y = a + b \pmod{2} .
\end{align}
The $\wedge$ denotes the logical AND.
Table~\ref{table_CHSH} summarizes the winning response pairs.

\begin{table}[ht]
    \centering
    \begin{tabular}{c | c | c}
%    & \multicolumn{2}{c}{Player 1} \\ \hline
   $y \backslash x$ & 0         & 1 \\ \hline
    0                   & 00 11     & 00 11 \\
    1                   & 00 11     & 01 10 \\
    \end{tabular}
    \caption{\caphead{Diagrammatic specification of the CHSH game:} Each column corresponds to one possible value of Alice's question, $x$, and each row corresponds to one possible value of Bob's question, $y$. Each cell contains the winning response pairs $ab$.}
    \label{table_CHSH}
\end{table}

The CHSH game illustrates the separation between 
what players can achieve when sharing only classical resources
and what players can achieve when sharing entanglement.
Suppose that $x$ and $y$ are selected uniformly randomly.
Players given only classical resources have a probability $\leq 3/4$ 
of winning a random round.
Players who measure a shared entangled state
have a probability $\sin^2 (3 \pi / 8)  \approx  0.854$.
Both facts are proved below. 

\begin{theorem} \label{thm: max classical value for CHSH}
A classical strategy based on shared randomness can win the CHSH game with probability at most $3/4$.
\end{theorem}

\begin{proof}
Classical players can achieve the value (the optimal win probability) 
with a deterministic strategy. 
We prove this claim with a fairly standard minimax argument:
Let $\omega$ denote the game's value.
Assume that some randomness-based strategy achieves $\omega$.
Let $r$ denote the random seed. By assumption,

\begin{align}
    \mathbb{E}_r \left( \mathbb{E}_{a,b} 
    \left( \sum_{a, b \, : \, x \wedge y = a + b} P(a, b  \,  |  \, x, y, r) \right) \right)
    = \omega .
\end{align}
$P(a, b  \, |  \, x, y, r)$ denotes the probability that the players respond with $a$ and $b$, conditioned on the questions $x$ and $y$ and on the random seed $r$. 
Some value $r_0$ of $r$ maximizes the inner expectation value,
by the average's convexity.
Fixing $r = r_0$ results in a deterministic strategy that achieves the game's value, $\omega$, as claimed.

Restricted to deterministic strategies, the players have few options.
Given a question $i$, a player must respond with some fixed output. 
Define $a_i$ as Alice's response to the question $i$,
and define $b_i$ as Bob's response to $i$.
In the CHSH game, the winning responses satisfy $a_i,b_i \in \{0,1\}$, 
\begin{align}
a_0 + b_0 &= 0 \pmod{2}, \\ \nonumber
a_1 + b_0 &= 0 \pmod{2}, \\ \nonumber
a_0 + b_1 &= 0 \pmod{2}, \quad \text{and} \\ \nonumber
a_1 + b_1 &= 1 \pmod{2} .
\end{align} 
Linear algebra over $\mathbb{F}_2$ shows that these equations cannot all be satisfied simultaneously. Hence any deterministic classical strategy must lose on at least one of the four question pairs. Such a strategy wins a random round with probability $\leq 3/4$.  
\end{proof}
\noindent
Constructing a deterministic classical strategy that achieves a win probability of $3/4$ is straightforward~\cite{Preskill_01_Lecture,CHTW04}.
% Reference: CHTW04 -- p. 5
% Reference: John's lecture notes -- section that begins on p. 20
The construction shows that the CHSH game's classical value is 3/4. 

Next, we construct a quantum strategy that has a superclassical probability $>3/4$ of winning the CHSH game. Our presentation is nonstandard but will prove useful later. We begin by reviewing notation and facts about maximally entangled two-qubit states.

Let $\ket{\Sing} 
:= \frac{1}{\sqrt{2}}\left(\ket{01} - \ket{10}\right)$ 
denote the singlet and 
$\ket{\Sing(\theta)} 
:= \frac{1}{\sqrt{2}}\left(\ket{01} - e^{i \theta} \ket{10}\right)$.
We denote the operator that rotates one qubit 
about the $z$-axis through an angle $\theta$ by
$R_z\left(\theta \right) :=  e^{ -i \theta \sigma_z / 2}$.  

\begin{lemma}
Rotations compose as
\begin{enumerate}
    \item $[  R_z(\theta_1) \otimes R_z( -\theta_2)  ]  \ket{\Sing} 
    =  (\text{phase}) \ket{\Sing(\theta_1 + \theta_2)}.$ \label{eq:ThetasAdd}
\end{enumerate}
Consider preparing a pure two-qubit state $\ket{\psi}$,
then measuring $\sigma_x \otimes \sigma_x$.
A classical two-bit string results.
Let $P_{XX}\left(S | \psi \right)$ denote the string's probability
of being in the set $S$.
If $S_\even := \{00, 11\}$ and $S_\odd := \{01, 10\}$,
\begin{enumerate}[resume]

    \item $P_{XX}(S_\even \; | \Sing ) = 0$, 
    \label{eq:XXevens}
    
    \item $P_{XX}\LParen S_\even \; | \Sing(\pi) \RParen = 1$,
    \quad and
    \label{eq:XXodd}
    
    \item $P_{XX}\LParen S_\even \; | \Sing(\theta) \RParen = \sin^2(\theta/2)$. 
    \label{eq:XXarbitrary}
    
\end{enumerate}
\end{lemma}

\begin{proof}
Identities \ref{eq:ThetasAdd}--\ref{eq:XXodd} can be verified by direct calculation. 
To prove identity \ref{eq:XXarbitrary}, we note that $\ket{\Sing(\theta)}$ equals a linear combination of $\ket{\Sing}$ and $\ket{\Sing(\pi)}$. Furthermore, 
$\braket{\Sing}{\Sing(\pi)} =0$. Hence 
$\ket{\Sing(\theta)} = \alpha\ket{\Sing} + \beta\ket{\Sing(\pi)}$, and
\begin{align}
    P_{XX} \LParen S_\even \; | \Sing(\theta) \RParen 
    = \abs{\alpha}^2 
    &= \abs{\braket{\Sing}{\Sing(\theta)}}^2 \\
    &= \frac{1}{4}\abs{1 - \exp(i\theta)}^2\\
    &= \frac{1}{2}\left[ 1 - \cos(\theta)\right] = \sin^2(\theta/2).
\end{align}
\end{proof}

These facts underlie a strategy for winning the CHSH game, 
using quantum resources,
with a greater probability than is achievable with only classical resources.
The quantum strategy consists of the following steps:
\begin{enumerate}
    \item Alice and Bob prepare $\ket{\Sing}$, and each player takes one qubit. The players agree on how each will generate a response, given any possible question.
    \item Upon receiving question $i$, Alice rotates her qubit with $R_z(\theta_i)$. Upon receiving question $i$, Bob rotates his qubit with $R_z(-\theta_i)$. The rotation angle $\theta_i$ depends on the question and the strategy.
    \item Each player measures his/her qubit's $\sigma_x$. The outcome is sent to the verifier as a response.
\end{enumerate}
We now identify angles $\theta_i$ that lead to a superclassical probability 
of winning the CHSH game.

\begin{lemma}
A quantum strategy with rotation angles 
$\theta_{0} =  -3\pi / 8$ and  $\theta_{1}  =  9\pi /8$
%$\theta_{B,0}  =  -3\pi / 8$, and  $\theta_{B,1} = -9\pi /8$
wins the CHSH game with probability $\sin^2 (3\pi/8) \approx 0.854$.
\end{lemma}

\begin{proof}
We verify the claim computationally. 
Upon receiving the question pair $00$, the players win with a probability 
\begin{align}
   P_{XX}\LParen S_\even \; | \Sing(3 \pi / 4)  \RParen
   =  \sin^2(- 3 \pi / 8) . \label{eq_CHSH_win00}
\end{align}
Upon receiving 01 or 10, the players win with a probability
\begin{align}
   P_{XX}\LParen S_\even \; | \Sing (3 \pi / 4)  \RParen
   =  \sin^2( 3 \pi / 8) . \label{eq_CHSH_win01_10}
\end{align}
Finally, upon receiving 11, the players win with a probability
\begin{align} 
   P_{XX} \LParen  S_\odd \; | \Sing( 9 \pi / 4 )  \RParen 
   &= 1- P_{XX} \LParen  S_\even  \; | \Sing( 9 \pi / 4 ) \RParen \\
   &= 1 - \sin^2 ( 9 \pi / 8 )  \\
   &= \sin^2 (3 \pi / 8). \label{eq_CHSH_win11}
\end{align}
Hence the players have a total win probability,
averaged over the possible question pairs, of $\sin^2 (3 \pi / 8)$.
\end{proof}
\noindent
A wide range of rotation angles can achieve superclassical win probabilities.
For example, $\theta_0 = \pi/2$ and $\theta_1 = 3\pi/4$ lead to a win probability of $\approx 0.802$.

\section{Details: Quantum violation of the nonlinear Bell inequality for macroscopic measurements}
\label{app_Prove_Q_Violation}

Here, we complete the proof of Theorem~\ref{thm_quantum_violation}.
In Supplementary Note~\ref{app_Eq_In_Q_Violation}, we prove Eq.~\eqref{eq_Q_Help}.
Supplementary Note~\ref{app_Q_Violate_Error} shows that the proof in the main text
is robust with respect to small experimental imperfections.

\subsection{Proof of Eq.~\eqref{eq_Q_Help}}
\label{app_Eq_In_Q_Violation}

Let $\ket{\Sing}  :=  \frac{1}{ \sqrt{2} } ( \ket{01} - \ket{10} )$ denote the singlet.
We simplify notation by omitting superscripts from the microscopic responses,
e.g., $a_x^{(i)}$.
Recall that $a_x, b_y  \in  \{0, 1\}$.
Consequently, for any $x,y \in \{0,1\}$, each expectation value $\E{a_x b_y}$ 
contains only one nonzero term,
the term in which $a_{x}  =  b_{y}  =  1$. 
These equalities are satisfied when Alice's $i^\th$ microscopic system 
and Bob's $i^\th$ microscopic system output 1s.
Consequently, $\E{ a_{x} b_{y} }$ equals the probability that 
$a_{x} = b_{y} = 1$:
\begin{align}
   \label{eq_Prob_Agree_11}
   \E{ a_{x} b_{y} }  =  \p{ a_{x} {=} b_{y} {=} 1 } .
\end{align}

We calculate $\E{ a_x b_y}$ when Alice and Bob follow the CHSH strategy outlined in Supplementary Note~\ref{app_CHSH_Backgrnd}. The probability that Alice and Bob send responses $11$, given questions $xy$, is 
\begin{align}
    \label{eq_Thm2_App_1}
    \p{ a_{x} {=} b_{y} {=} 1 } 
    = \abs{\bra{--} R_z(\theta_x + \theta_y) \ket{ \Sing }}^2.
\end{align}
Since $ZZ\ket{ \Sing} = - \ket{\Sing}$,
\begin{align}
   \label{eq_Thm2_App_2}
   \abs{\bra{--} R_z(\theta_x + \theta_y) \ket{ \Sing }}^2
   % % %
   & = \abs{\bra{--}  R_z(\theta_x + \theta_y)  \left(ZZ\right)  \ket{ \Sing }}^2 \\
   % % %
   & = \abs{\bra{--} \left(ZZ\right) R_z(\theta_x + \theta_y) \ket{ \Sing }}^2 \\
   % % %
   \label{eq_Thm2_App_3}
   & = \abs{\bra{++} R_z(\theta_x + \theta_y) \ket{ \Sing }}^2 .
\end{align}
Hence the LHS of Eq.~\eqref{eq_Thm2_App_1} decomposes in terms of
Eq.~\eqref{eq_Thm2_App_3} and the LHS of Eq.~\eqref{eq_Thm2_App_2}:
\begin{align}
   \label{eq_Thm2_App_4}
   \p{ a_{x} {=} b_{y} {=} 1 }
   & = \frac{1}{2}\left[
          \abs{\bra{++} R_z(\theta_x + \theta_y) \ket{ \Sing }}^2 
         +  \abs{\bra{--} R_z(\theta_x + \theta_y) \ket{ \Sing }}^2\right]  \\
    &= \frac{1}{2} P_{XX}\LParen S_\even \; | \Sing(\theta_x + \theta_y)  \RParen.
\end{align}
Substituting in from Eq.~\eqref{eq_CHSH_win00} gives
\begin{align}
\p{ a_0 {=} b_0 {=} 1 }
&= \frac{1}{2}P_{XX}\LParen S_\even \; | \Sing(3 \pi / 4)  \RParen \\
   \label{eq_Thm2_App_5}
&= \frac{1}{2}\sin^2(3\pi/8).
\end{align}
Similarly, by Eq.~\eqref{eq_CHSH_win01_10},
\begin{align}
\p{ a_1 {=} b_1 {=} 1 }
= \p{ a_0 {=} b_1 {=} 1 }
&= \frac{1}{2}P_{XX}\LParen S_\even \; | \Sing(3 \pi / 4)  \RParen \\
&= \frac{1}{2}\sin^2(3\pi/8).
   \label{eq_Thm2_App_6}
\end{align}
Finally, Eq.~\eqref{eq_CHSH_win11} implies that
\begin{align}
\p{ a_1 {=} b_1 {=} 1 }
&= \frac{1}{2} P_{XX}\LParen S_\even \; | \Sing(9 \pi / 4)\RParen \\
&= \frac{1}{2}\sin^2(9 \pi/8) \\
&= \frac{1}{2} - \frac{1}{2} \sin^2(3\pi/8).
   \label{eq_Thm2_App_7}
\end{align}
Combining Equations~\eqref{eq_Thm2_App_5},~\eqref{eq_Thm2_App_6},
and~\eqref{eq_Thm2_App_7} with Eq.~\eqref{eq_Prob_Agree_11} yields 
\begin{align}
\E{a_0b_0} + \E{a_1b_0}+\E{a_0b_1}-\E{a_1b_1} 
&= \p{a_0{=}b_0{=}1} + \p{a_1{=}b_0{=}1}
     +\p{a_0 {=} b_1 {=} 1}-\p{a_1 {=} b_1 {=} 1} \\
&= 2 \sin^2(3\pi/8) - \frac{1}{2}.
\end{align}

\subsection{Analysis of experimental error in quantum violation of the macroscopic Bell inequality}
\label{app_Q_Violate_Error}

In the sketch of the proof of Theorem~\ref{thm_quantum_violation}, we showed that
[Eq.~\eqref{eq_ideal_quantum_macroscopic}]
\begin{align}
\Bell(A_0',A_1',B_0',B_1') = 2\sqrt{2}  .
\label{eq_ideal_quantum_macroscopic2}
\end{align}
The macroscopic random variables $A_0'$, $A_1'$, $B_0'$, $B_1'$ were produced by noise-free measurements of perfectly prepared Bell states. 

Noise can taint the setup, as discussed in Sec.~\ref{sec_Setup}. To recap, we define 
\begin{align}
    A_x = A_x' + r_{A_x} .
\end{align}
The random variable $r_{A_x}$ represents noise whose variance is bounded:
$\Var{r_{A_x}} \leq \epsilon N$.
$A_x$ represents the macroscopic outcome of a measurement
made in the presence of noise.
$B_y$ and $r_{B_y}$ are defined analogously. 

In Supplementary Note~\ref{app_Proof_Bell_ineq}, we showed that [Eq.~\eqref{eq_Error_Result1}]
\begin{align}
   \label{eq_Error_Result1_Again}
   \abs{ \Bell(A_0,A_1,B_0,B_1)  -  \Bell(A_0', A_1', B_0', B_1')} 
   \leq 16\epsilon - 32\sqrt{\epsilon}.
\end{align}
Rearranging gives 
\begin{align}
    \Bell(A_0,A_1,B_0,B_1)  
    \geq   \Bell(A_0', A_1', B_0', B_1') - 16\epsilon - 32\sqrt{\epsilon} .
\end{align}
Substituting in from Eq.~\eqref{eq_ideal_quantum_macroscopic2} gives 
\begin{align}
\Bell(A_0,A_1,B_0,B_1) \geq 2\sqrt{2} - 16\epsilon - 32\sqrt{\epsilon} ,
\end{align}
the desired result.

\section{Equivalence of local quantum correlations and global classical correlations
as resources for violating the macroscopic Bell inequality} 
\label{sec_violating_nonlin_bell_ineq}

%The following theorem qualitatively bounds the power of entanglement to produce correlations in macroscopic systems.
%Macroscopic systems formed from independent pairs of entangled particles 
%can access a set of correlations larger than 
%the set of correlations accessible by classical systems formed from independent pairs of classically correlated particles, but 
% % % Aram's diagram could be helpful here.

We formalize the discussion in Sec.~\ref{sec_Discussion} with a theorem.
To state the theorem cleanly and to avoid confusion with $A_x$ and $B_y$, we introduce experimentalists Carol and Dan. 
Each has a system of $N$ particles.
Carol measures with settings $x = 0, 1$,
and Dan measures with settings $y = 0, 1$.
The macroscopic outcomes are the values of random variables 
$C_x$ and $D_y$. 

Like Alice and Bob, Carol and Dan obey assumption~\ref{assumption: noninteracting systems} in Sec.~\ref{sec_Setup}.
But Carol and Dan's systems can share global correlations,
violating assumption~\ref{assumption: independent particles}.
We assume that Carol and Dan's measurements
suffer from no other errors.

\begin{theorem}
\label{thm_uncorrelated_classical_strat}
Carol and Dan can, with $2N \gg 1$ particles, produce correlations that satisfy 
\begin{align}
    \Bell(C_0, C_1, D_0, D_1) = 2N.
\end{align}
\end{theorem}

\begin{proof}
Carol and Dan can implement a probabilistic strategy,
flipping an unbiased coin.
If the coin falls heads-up, they fix their particles to output 1s, 
regardless of measurement settings.
If the coin falls tails-up, all particles are fixed to output 0s. 
A straightforward calculation gives 
\begin{align}
\cov{C_0}{D_0} 
&= \frac{ N^2 }{2} - \left(\frac{N}{2} \right)^2 
= \frac{ N^2 }{ 4 } .
\end{align}
Similar calculations describe the other covariances, so
\begin{align}
    \Bell(C_0, C_1, D_0, D_1) 
    &= \frac{4}{N}\left(\frac{3N^2}{4} - \frac{N^2}{4}\right) 
    = 2N .
\end{align}
\end{proof}

\section{Macroscopic CHSH game}
\label{app_Nonlocal_Game}

We develop a macroscopic analog of the CHSH game, 
reviewed in Supplementary Note~\ref{app_CHSH_Backgrnd}. 
Just as the CHSH game is built on the Bell-CHSH inequality, 
the macroscopic CHSH game is built on 
the macroscopic Bell inequality proven in \Cref{thm_marco_bell_ineq}.  
The macroscopic CHSH game differs from the microscopic CHSH game in three ways:
\begin{enumerate}

   \item
   The macroscopic game is multiplayer. It involves $2N$ players, $N$ Alices and $N$ Bobs, who cannot communicate with each other. Each receives a question and responds. However, the verifier aggregates the Alices' responses and aggregates the Bobs' responses. We also place an important restriction on the players. We assume each Alice plays the game independently from all players except one Bob, and vice versa. This means each Alice can share randomness or entanglement with at most one Bob.
   The game's Alices play the same role as the microscopic particles in the main text. In the main text's photon-beam example, photons serve as the game's Alices and Bobs, and the beams serves as the game's aggregate Alice and aggregate Bob.

   \item
   The game is multiround; several question-and-answer sessions take place. We assume the players lack memories, following the same strategy in every round.\footnote{
This requirement might seem strong from a nonlocal-games perspective. 
However, it is natural from the perspective of the macroscopic Bell test,
presented in the main text, equivalent to our nonlocal game.
We illustrate with the photon beams introduced in Sec.~\ref{sec_Setup}.
To perform the macroscopic Bell test,
one evaluates the macroscopic Bell parameter~\eqref{eq_Bell_Param}
after running multiple trials. 
Multiple trials manifest, in the photon-beam example, as
sequential measurements of Alice's beam's intensity
and of Bob's beam's intensity.
Alice's sequential measurements are
measurements of independent sets of photons.
The photons' independence is equivalent to the players' amnesia in the nonlocal game.}
   The verifier evaluates the players' performance after analyzing all the rounds' outcomes. 
   
   \item 
   The verifier assigns to the players a score in $[0,1]$, rather than a win or a loss. 
\end{enumerate}

The rest of this appendix is organized as follows.
We formulate the game in ~\ref{sec_Define_Macro_CHSH}.
In Supplementary Note~\ref{sec_Bound_Class_Macro_CHSH}, 
we upper-bound the score achievable by 
players given only classical systems.
We show how to violate the bound, using quantum systems,
in Supplementary Note~\ref{sec_Q_Macro_CHSH}.

\subsection{Definition of the macroscopic CHSH game}
\label{sec_Define_Macro_CHSH}

The macroscopic CHSH games is a multiround nonlocal game played with $N$ memoryless Alices and $N$ memoryless Bobs.
In every round, the verifier randomly picks a question pair $xy$ from the set $\{00,01,10,11\}$. 
The question $x$ is sent to every Alice, 
and the question $y$ is sent to every Bob. 
Each Alice responds with one bit, as does each Bob. 
The verifier keeps a transcript of the questions and responses.
After all the rounds, the verifier scores the game as follows:
\begin{enumerate}
    \item 
    The verifier calculates the average number $\overline{A}_x$ 
    of Alices who answer 1 to question $x$
    and the average number $\overline{B}_y$ 
    of Bobs who answer 1 to question $y$,
    for all questions $x, y \in \{0, 1\}$. \label{item_calc_average}
    
    \item 
    The verifier assesses each round, using the following procedure.  
    Label the round's questions $x$ and $y$. 
    Let  $\boldsymbol{A_x}$ denote the number of Alices who reply $1$ to $x$,
    and let $\boldsymbol{B_y}$ denote the number of Bobs
    who reply 1 to $y$. 
    ($\boldsymbol{A_x}$ and $\boldsymbol{B_y}$ 
    are values of random variables.)
    The verifier checks whether $\boldsymbol{A_x}$ and $\boldsymbol{B_y}$
    satisfy two criteria, motivated below:
    \begin{enumerate}
    
        \item \label{item_Macro_Win_Condn1}
        If either number of 1s lies too close to the mean,
        the players lose the round:
        $\abs{\boldsymbol{A_x} - \overline{A}_x} < \sqrt{N}$,
        or $\abs{\boldsymbol{B_y} - \overline{B}_y}  < \sqrt{N}$.
        
        \item  \label{item_Macro_Win_Condn2}
        Otherwise, the verifier checks whether
        \begin{align}
            \sign(\boldsymbol{A_x} - \overline{A}_x) 
            \sign(\boldsymbol{B_y} - \overline{B}_y) 
            = (-1) ^{x  \wedge  y} \pmod{2}  .
        \end{align}
        If this equation is true, the players win the round. If not, they lose. 
    \end{enumerate}
    
    \item The verifier assigns the players a score for the entire game:
    The verifier identifies the question pair $xy = x_0 y_0$
    on which the players won least frequently.
    % on which the players have won the smallest fraction of rounds. 
    The fraction of $x_0 y_0$ rounds on which the players won
    becomes their score. 
\end{enumerate}

A few comments about this game are in order. First, we discuss the single-round win conditions, oppositely the order in which they are presented. 
Consider assigning the aggregated Alices a 1
if far more than the average number of constituent Alices respond with 1s
and assigning the aggregated Alices a 0
if far fewer than the average number respond with 1s.
Assign the aggregated Bobs a 1 or a 0 analogously.
Condition \ref{item_Macro_Win_Condn2} confirms that
the aggregated Alices and aggregated Bobs
satisfy the CHSH win condition~\eqref{eq_CHSH_Win_Criterion}
in one round.
% Condition \labelcref{item_Macro_Win_Condn2}) resembles the win condition~\eqref{eq_CHSH_Win_Criterion} for the original CHSH game, but applied to the aggregated Alices and the aggregated Bobs. 
% If we interpret significantly higher than average number of individual Alices (resp. Bobs) answering $1$ as an aggregate $1$, and significantly fewer as an aggregate $0$, the single round win condition checks that the aggregate responses win the original CHSH game. 

Condition~\ref{item_Macro_Win_Condn1} ensures that the players fail a round if their responses lie too close to the average responses. In the absence of this condition, the macroscopic CHSH game would reduce to the microscopic game: 
Imagine eliminating condition~\ref{item_Macro_Win_Condn1}
and aggregating responses via 
$\sign(A_x - \mathbb{A}_x)$ and $\sign(B_y - \mathbb{B}_y)$.
The players could follow a strategy according to which
$N - 1$ Alices (Bobs) responded deterministically.
The final Alice's (Bob's) response would determine whether the number of 1s received were higher or lower than the average, determining the aggregate response. A microscopic response would control the macroscopic response.

$\sqrt{N}$ was chosen for the following reason.
Each $A_x$ and $B_y$ is a sum of 
independent and identically distributed (i.i.d.) random variables.
Consider the limit as $N \to \infty$.
Consider the probability that $A_x$ or $B_y$ assumes a value
$N^{1/2 + \epsilon}$ away from its mean.
This probability vanishes for all $\epsilon > 0$,
by the central limit theorem.
Hence fluctuations $\sim \sqrt{N}$ are 
the largest---most easily visible---fluctuations that can occur.
The verifier must be able to detect these largest fluctuations
and need not resolve finer fluctuations.
% 11/19/19: The analysis does not change substantially if we demand that the verifier be able to resolve fluctuations of size (const.)sqrt{N}, in accordance with the standard deviation's equaling (const.)sqrt{N} + (higher-order corrections) according to the central limit theorem.
A similar criterion is introduced in \cite{Navascues_16_Macroscopic}.

Second, we elucidate how the macroscopic Bell inequality's nonlinearity manifests in the macroscopic CHSH game.
The inequality and the game distinguish
classical randomness from pairwise entanglement 
(entanglement shared by each Alice with exactly one Bob and vice versa),
without violating the principle of macroscopic locality~\cite{Navascues_10_Glance,Yang_11_Quantum,Navascues_16_Macroscopic}.
The inequality succeeds by depending on probabilities nonlinearly (Sec.~\ref{sec_Discussion}).
Therefore, also the game should involve nonlinearity.
Each strategy specifies a set of four conditional probability density functions (PDFs),
$\p{ a, b  \,  |  \,  xy {=} 00 }$,  $\p{ a, b  \,  |  \,  xy {=} 01 }$,  
$\p{ a, b  \,  |  \,  xy {=} 10 }$,  and $\p{ a, b  \,  |  \,  xy {=} 11 }$.
The score is a function of the four PDFs and is nonlinear in each PDF.
The reason is step~\ref{item_calc_average}:
The verifier calculates average aggregate responses,
then compares the actual aggregate responses with the averages.

This use of averages implies that the players should lack memories:
Suppose that the players had only classical resources
but had memories.
The players could use different strategies in different rounds.
Mixing strategies would change the averages,
allowing players to win rounds
that they would lose if they followed either strategy consistently.
%The classical players could raise their score
%by mixing classical strategies.

% Second, we discuss non-linearity. Ultimately, our goal is to distinguish classical players from players that share pairwise entanglement (so each Alice shares entanglement with one Bob, and vice versa) through the score they achieve on the macroscopic CHSH game. By the principle of macrsocopic locality \cite{Navascues_16_Macroscopic}, this should only possible if the score depends on the player's strategy in some nonlinear way. In the macroscopic CHSH game, the nonlinearity is introduced in step \cref{item_calc_average}. The average responses $\bar{A}_x$ and $\bar{B}_y$ are computed by looking at players responses over many rounds. Mixing two different strategies can change this average, and allow players to win on rounds they would lose following either strategy by itself. This also explains why the players must be memoryless. If players could change strategies depending on the round, players could mix two different classical strategies, increasing their overall score. 

%
%
%
\subsection{Upper bound on the score achievable by classical players of the macroscopic CHSH game}
\label{sec_Bound_Class_Macro_CHSH}

% Our next goal is to show quantum players can achieve a higher score on the macroscopic CHSH game than is achievable classically. To do this, we note that 
We bound the classical players' score as follows.
The random variables 
$A_0$, $A_1$, $B_0$, and $B_1$ are distributed according to a multivariate Gaussian, by the central limit theorem. 
If these variables have limited variances and covariances,
they are unlikely to satisfy the win criteria,
(\ref{item_Macro_Win_Condn1}) and (\ref{item_Macro_Win_Condn2}).
We prove this fact for $x \wedge y = 0$ in Lemma~\ref{lem_gaussian_tails}
and for $x \wedge y = 1$ in Corollary~\ref{cor_neg_gaussian_tails}.
The proofs consist of technical calculations regarding
tails of multivariate Gaussians.
Combining the lemma and corollary with \Cref{thm_marco_bell_ineq} leads to Theorem~\ref{th_Macro_CHSH_Class}:
Every classical strategy has  small covariance on at least on question pair.
Hence the score achievable by classical players obeys an upper bound.  
% After proving the classical bound, we prove that a quantum strategy violates it
% (Theorem~\ref{th_Q_Macro_CHSH}):
% the strategy for violating the microscopic CHSH game (Supplementary Note~\ref{app_CHSH_Backgrnd}). 

\begin{lemma} \label{lem_gaussian_tails}
Let $\cX$ and $\cY$ denote random variables distributed according a multivariate Gaussian with variances
$\sigma_\cX^2 \leq N/4$ and 
$\sigma_\cY^2  \leq N/4$ and with covariance
$\cov{\cX}{\cY} \leq N/7$. 
The probability that $\cX$ and $\cY$ both far exceed their means
is small: 
\begin{align}
   \label{eq_gaussian_tails}
   \p{\cX - \E{\cX} \geq \sqrt{N} \wedge \cY - \E{\cY} \geq \sqrt{N}} 
   \leq 0.0051.
\end{align}
\end{lemma}

\begin{proof}
The proof is computational.
For ease of notation, we shift $\cX$ and $\cY$ so that each has mean 0. Let $\cY(x')$ denote the random variable $\cY$ conditioned on the event $\cX = x'$. We expand the probability in Eq.~\eqref{eq_gaussian_tails}:
 \begin{align}
     \p{\cX \geq \sqrt{N} \wedge \cY \geq \sqrt{N}} 
     &= \p{\cX \geq \sqrt{N}}
     \p{\cY  \geq \sqrt{N} \big| \cX \geq \sqrt{N} } \\
     &=  \int_{\sqrt{N}} ^\infty 
     \p{\cX = x'} 
     \p{\cY(x') \geq \sqrt{N}} dx' .
     \label{eq_decomp_joint_prob} \\
     &= \int_{\sqrt{N}} ^\infty  
     \int_{\sqrt{N}} ^\infty 
     \p{\cX = x'} 
     \mathbb{P}  \LParen  \cY(x') = y' \RParen  dx' dy' .
     \label{eq_gaussian_tails2}
\end{align}
 The theory of multivariate Gaussians implies that $\cY(x')$ is distributed according to a Gaussian with variance 
 \begin{align}
     \sigma_{\cY(x')} 
     = \sqrt{   \sigma_\cY^2 - \frac{\cov{\cX}{\cY}^2}{\sigma_\cX^2} } 
     \label{eq_cond_var}
 \end{align}
 and mean
 \begin{align}
      \mathbb{E} \LParen  \cY(x')  \RParen 
      &= \frac{x' \cov{\cX}{\cY}}{\sigma_\cX^2}. 
 \end{align}
We substitute the probabilities' Gaussian forms into Eq.~\eqref{eq_gaussian_tails2}.
If $\erf (z)  :=  \frac{2}{ \sqrt{\pi} }  \int_0^z  dt  \:  e^{ - t^2 }$
denotes the error function,
\begin{align}
    \p{\cX \geq \sqrt{N} \wedge \cY \geq \sqrt{N}} 
    &= \int_{\sqrt{N}} ^\infty 
    \frac{1}{\sqrt{2 \pi \sigma_\cX^2} }
    \exp   \left(   -\frac{ (x')^2}{2 \sigma_\cX^2}   \right)
    \int_{\sqrt{N}} ^\infty 
    \frac{1}{\sqrt{2 \pi \sigma_{\cY(x')}^2} }
    \exp   \left(   -\frac{  [ y' - \mathbb{E} \LParen  \cY(x')  \RParen ]^2  }{  
                                   2 \sigma_{ \cY(x') }^2  }   \right)  
    dy'  dx'\\
    &= \int_{\sqrt{N}} ^\infty  
    \frac{1}{\sqrt{2 \pi \sigma_\cX^2} }
    \exp\left(-\frac{(x')^2}{2 \sigma_\cX^2}\right)
    \frac{1}{2} \left[1 - 
    \erf \left(\frac{\sqrt{N} - \mathbb{E}  \LParen  \cY(x')  \RParen }{
                         \sqrt{2} \sigma_{\cY(x')}} \right)\right] dx' \\
    &= \int_{1} ^\infty  
    \frac{1}{\sqrt{2 \pi \sigma_\cX^2/N} }
    \exp\left(-\frac{(x')^2}{2 \sigma_\cX^2/N}\right)
    \frac{1}{2} \left[1 - \erf \left(\frac{1 - \mathbb{E}  \LParen  \cY(x')  \RParen}{\sqrt{2}  \: \sigma_{\cY(x')}/\sqrt{N}} \right)\right] dx' .
    \label{eq_gaussian_tails3}
%    &=\int_{\sqrt{N}} ^\infty \int_{\sqrt{N}} ^\infty \frac{1}{2 \pi\sqrt{ \sigma_x^2\sigma_\cY^2 - \cov{\cX}{\cY}^2}}
%    \exp\left(-\frac{(x')^2}{2 \sigma_x^2} -
%    \frac{(y' - x \cov{\cX}{\cY}/\sigma_\cX^2 )^2}{2    \sigma_\cY^2 - \cov{\cX}{\cY}^2/\sigma_\cX^2} \right) dx' dy'
\end{align}

 By assumption, $\sigma_{\cX} \leq \sqrt{N} /2$, 
 $\sigma_{\cY} \leq \sqrt{N} /2$, 
 and  $\cov{\cX}{\cY} \leq \sqrt{N / 7}$. 
 By the Cauchy-Schwarz inequality, 
 $ \sqrt{ \sigma_\cX^2  \sigma_\cY^2 } 
 \geq \cov{\cX}{\cY}$
 [Ineq.~\eqref{eq_Cov_Var}].
 We numerically optimize the probability~\eqref{eq_gaussian_tails3}
 subject to these constraints~\cite{BeneWatts_19_Code}.
 The probability maximizes when $\sigma_\cX$, $\sigma_\cY$, and $\cov{\cX}{\cY}$ assume their maximum possible values: 
 $\sigma_{\cX} = \sigma_{\cY} = \sqrt{N}/2$, and $\cov{\cX}{\cY} = N/7$. 
The maximum probability lies slightly below $0.0051$. 
\end{proof}

\begin{corollary} 
\label{cor_neg_gaussian_tails}
Let $\cX$ and $\cY$ denote random variables distributed according to 
a multivariate Gaussian with variances
$\sigma_\cX^2 \leq N/4$ and 
$\sigma_\cY^2  \leq N/4$ and with covariance
$\cov{\cX}{\cY} \leq -N/7$. 
The probability that $\cX$ far exceeds its mean
while $\cY$ lies far below its mean is small:
\begin{align}
\p{\cX - \E{\cX} \geq \sqrt{N} \wedge \cY - \E{\cY} \leq -\sqrt{N}} 
\leq  0.0051 .
\end{align}
\end{corollary}
\begin{proof}
Apply \Cref{lem_gaussian_tails} to the random variables $\cX$ and $-\cY$. 
\end{proof}

\begin{theorem}
\label{th_Macro_CHSH_Class}
Classical players can achieve an average score
of at most 0.0102 in the macroscopic CHSH game,
if the number $2N$ of players is sufficiently large.
\end{theorem}

\begin{proof}
By the multivariate central limit theorem, the random variables 
$A_0$, $A_1$, $B_0$, and $B_1$ 
come to obey multivariate Gaussian distributions in the large-$N$ limit. 
Let $xy$ denote an arbitrary question pair. 
If $x \wedge y = 0$, the players can win in two ways:
(i) The number $A_x$ of Alices who respond 1
lies far above the mean number $\overline{A}_x$ who respond 1.
Meanwhile, the number $B_y$ of Bobs who respond 1
lies far above the mean number $\overline{B}_y$. That is,
\begin{align}
    A_x - \overline{A}_x \geq \sqrt{N} 
    \; \wedge  \;
    B_y - \overline{B}_y \geq \sqrt{N} .
\end{align}
(ii) $A_x$ lies far below its average,
while $B_y$ lies far below its average:
\begin{align}
        A_x - \overline{A}_x \leq -\sqrt{N} 
        \; \wedge  \; 
        B_y - \overline{B}_y \leq -\sqrt{N} .
\end{align}
The players' probability of winning via (i)
was bounded in Lemma~\ref{lem_gaussian_tails}.
Their probability of winning via (ii) is the same, by the multivariate Gaussian's symmetry.
Hence the players' total probability of winning on 
$xy \, : \,  x \wedge y = 0$ is
\begin{align}
    \p{A_x - \E{A_x} 
    \geq \sqrt{N} \wedge  B_y - \E{B_y} \geq \sqrt{N}} + \p{A_x - \E{A_x} \leq -\sqrt{N} \wedge  B_y - \E{B_y} \leq -\sqrt{N}} \\
    = 2 \p{A_x - \E{A_x} \geq \sqrt{N} \wedge  B_y - \E{B_y} \geq \sqrt{N}} .\label{eq_game_win_even}
\end{align}
The second line follows from the multivariate Gaussian's symmetry. 

Now, suppose that $x \wedge y = 1$.
The players win if the number $A_x$ of Alices who reply 1
lies far above/below its mean
while the number $B_y$ of Bobs who reply 1
lies far below/above its mean:
\begin{align}
    A_x - \overline{A}_x \geq \sqrt{N} 
    & \;  \wedge  \; 
    B_y - \overline{B}_y \leq -\sqrt{N} 
    \, , \text{ or} \\
    A_x - \overline{A}_x \leq -\sqrt{N} 
    & \;  \wedge  \;  
    B_y - \overline{B}_y \geq \sqrt{N}
\end{align}
% equals 1, while the other equals 0.
These two events have equal probabilities of occurring.
Hence the players have a win probability of
\begin{align}
    2 \p{A_x - \E{A_x} \geq \sqrt{N} \wedge  B_y - \E{B_y} \leq -\sqrt{N}}.\label{eq_game_win_odd}
\end{align}

We now invoke the macroscopic Bell inequality.
For simplicity, we have not defined noise or classical global correlations
in the game.
We therefore set $\epsilon = 0$ in \Cref{thm_marco_bell_ineq}.
The theorem, with the definition~\eqref{eq_Bell_Param}, implies that 
\begin{align}
    \cov{A_0}{B_0} + \cov{A_0}{B_1} + \cov{A_1}{B_0} - \cov{A_1}{B_1} \leq 4N/7.
\end{align}
A minimimax argument gives
\begin{align}
    \label{eq_Min_Covs}
    \min \left\{\cov{A_0}{B_0}, \cov{A_0}{B_1} , \cov{A_1}{B_0}, 
    - \cov{A_1}{B_1} \right\}  \leq N/7 .
\end{align}
Therefore, some question pair $xy$ satisfies either 
$x \wedge y = 0$ and $\cov{A_x}{B_y} \leq N/7$ 
or $x \wedge y = 1$ and $\cov{A_x}{B_y} \geq -N/7$. 
Furthermore, each $A_x$ and each $B_y$ is 
a sum of $N$ independent random variables, 
each of which has a variance of $\leq 1/4$. Hence
\begin{align}
    \label{eq_Var_Bound_For_Lemma}
    \Var{A_0}, \Var{A_1}, \Var{B_0}, \Var{B_1} \leq N/4.
\end{align}
Consider the $(A_x , B_y)$ that achieves the minimization in
Ineq.~\eqref{eq_Min_Covs}.
It satisfies the assumptions in Lemma~\ref{lem_gaussian_tails}
and Corollary~\ref{cor_neg_gaussian_tails}.
% by Ineqs.~\eqref{eq_Min_Covs} and~\eqref{eq_Var_Bound_For_Lemma}.
By the lemma and Eq.~\eqref{eq_game_win_even},
and by the corollary and Eq.~\eqref{eq_game_win_odd}, 
this $(A_x , B_y)$ satisfies the win conditions with probability 
$\leq 2 \times 0.0051  =  0.0102$.
The score equals the minimum, over all $xy$ pairs,
of the probability that the players win on $xy$.
Hence the score $\leq  0.0102$, as claimed.
\end{proof}

\subsection{Superclassical score in the macroscopic CHSH game}
\label{sec_Q_Macro_CHSH}

We have upper-bounded the score achievable by classical players
of the macroscopic CHSH game.
Now, we show that players can violate this bound, given quantum resources.
The proof is constructive; we exhibit a superclassical strategy.
It is built on the strategy shown, in Supplementary Note~\ref{app_CHSH_Backgrnd},
to win the microscopic CHSH game with a superclassical probability.

\begin{theorem}
\label{th_Q_Macro_CHSH}
Players given quantum resources can achieve a score of $\geq 0.0150$ in the macroscopic CHSH game. 
\end{theorem}

\begin{proof}
As in the proof of \Cref{thm_quantum_violation}, 
each Alice-Bob pair adopts the conventional CHSH strategy
(Supplementary Note~\ref{app_CHSH_Backgrnd}):
Each pair shares a singlet.
When Alice measures any observable, she has
a probability $1/2$ of obtaining $+1$,
and responding $1$ to the verifier,
and a probability $1/2$ of obtaining $-1$,
and responding $0$.
The same is true of Bob.
Hence the aggregated Alice responses 
and the aggregated Bob responses obey
\begin{align}
    \E{A_0} = \E{A_1} = \E{B_0} = \E{B_1} = N/2
\end{align}
and 
\begin{align}
    \Var{A_0} = \Var{A_1} &= \Var{B_0} = \Var{B_1} = N/4.
\end{align}
According to 
% the analysis of \Cref{thm_quantum_violation}  <-- The first equation referred to is in the proof of Theorem 1.
Equations~\eqref{eq_Bound_Help1},~\eqref{eq_Prob_Agree_11},~\eqref{eq_Thm2_App_5},~\eqref{eq_Thm2_App_6}, and~\eqref{eq_Thm2_App_7},
$A_0$, $A_1$, $B_0$, and $B_1$ satisfy also
\begin{align}
    \cov{A_0}{B_0} & = \cov{A_0}{B_1} = \cov{A_1}{B_0} = -\cov{A_1}{B_1} \\
    % % %
    & =  \sum_{i = 1}^N  \cov{ a_0^{(i)} }{  b_0^{(i)}  } \\
    % % %
    & =  \sum_{i = 1}^N  \left[
    \E{ a_0^{(i)} }  \E{  b_0^{(i)}  }    \right] \\
    % % %
    & =  \sum_{i = 1}^N  
    \left[  \p{  a_0^{(i)}  =  b_0^{(i)}  =  1  }
             -  \frac{1}{2}  \cdot  \frac{1}{2}  \right]  \\
    % % %
    &= N  \left[ \frac{1}{2} \sin^2 \left(  \frac{3 \pi}{8}  \right) - \frac{1}{4}  \right]\\
    &= N/(4\sqrt{2}) .
\end{align}

We can numerically bound the mean score averaged over instances of the (multiround) game. The computation bounding the mean win probability on the question pair $00$ is given here. The other calculations are similar and produce identical results.

By the multivariate central limit theorem, the joint distribution over the random variables $A_0$ and $B_0$ is a multivariate Gaussian with means $N/2$ and covariance matrix 
\begin{align}
    \Sigma = 
    \begin{pmatrix} 
    \frac{N}{4} & \frac{N}{4\sqrt{2}} \\
    \frac{N}{4\sqrt{2}} & \frac{N}{4} 
    \end{pmatrix}.
\end{align}
The players win if $A_0$ and $B_0$ both lie above or both lie below their means
by at least $\sqrt{N}$. This event occurs with probability
\begin{align}
    \p{ A_0  \geq  N/2+\sqrt{N}   
         \:  \wedge  \:
         B_0  \geq  N/2+\sqrt{N} }
    +  \p{  A_0  \leq  N/2-\sqrt{N}
              \:  \wedge  \:
              B_0  \leq  N/2-\sqrt{N}  }
    \geq 0.0150.
\end{align}
The bound was computed numerically~\cite{BeneWatts_19_Code}.
\end{proof}

\section{Toy application to Posner molecules}
\label{app_Posner}

% Reference: Meeting notes --> Adam - 7/25/19

In Sec.~\ref{sec_Discussion}, we introduced
a potential application of the macroscopic Bell inequality
to Posner molecules. 
Posners are beginning to be characterized experimentally.
If they are found to retain coherences, entanglement should be tested for~\cite{Fisher_17_Are}.
How can it be, since the operations conjectured to be performable on Posners
differ from the operations used in conventional Bell tests~\cite{NYH_19_Quantum}?
We begin answering that question here, though further work is needed.
We construct a partially macroscopic Bell test implementable with
the operations conjectured to be performable on Posners~\cite{Fisher_15_Quantum,NYH_19_Quantum}.
The test relies on macroscopic intensity measurements
but microscopic manipulations of Posners.

The background needed to understand this supplementary note can be found in the following places.
First, information about Posners appears in~\cite{NYH_19_Quantum}, 
particularly in Sections~2.1, 3.1, 3.2, 3.4, and 3.7.
Second, useful background appears in
this paper's Sec.~\ref{sec_Discussion}, Supplementary Note~\ref{app_CHSH_Backgrnd},
and Supplementary Note~\ref{subsec:PosnerPreliminaries}.
Third, calculations involving Posner states were performed with code 
that was originally written by E.~Crosson for~\cite{NYH_19_Quantum} 
and was repurposed with her permission~\cite{BeneWatts_19_Code}.

The rest of this supplementary note is structured as follows.
The first two sections offer a warmup:
Supplementary Note~\ref{subsec:PosnerPreliminaries}
overviews the tools used.
Supplementary Note~\ref{subsec:PosnerTest} introduces a strategy for winning the CHSH game with a superclassical probability, using a finely controlled system of a few Posners.  
We sketch a many-Posner Bell test in Supplementary Note~\ref{sec_Posner_Test_Actual}.
In Supplementary Note~\ref{sec_Posner_Analysis}, we analyze the sketch,
identify its shortcomings, and discuss opportunities for sharpening it.

\subsection{Preliminaries needed for the Posner Bell test}
\label{subsec:PosnerPreliminaries}

In Sec.~\ref{sec_Ops_For_Posner_Test}, we discuss the operations needed to perform our Posner Bell test. In Sec.~\ref{sec_Micro_Posner_4_Facts}, we introduce four facts that underlie the analysis of the Posner Bell test.

\subsubsection{Operations needed to implement the Posner Bell test}
\label{sec_Ops_For_Posner_Test}

These operations can be performed in principle, 
if Fisher conjectures correctly about Posner biochemistry~\cite{Fisher_15_Quantum}.
That is, these operations can be implemented
within the Posner model of quantum computation,
or with \emph{Posner operations}, defined in Sec.~3.4 of~\cite{NYH_19_Quantum}.
However, some operations require impractical microscopic control.
%Notably, our test requires the ability to rotate $\tau$ values of a Posner, a feature which is not present in the 
%We require three main classes of operations, corresponding roughly to setup, rotation, and finally measurement stages of the Bell test.
\begin{enumerate}

    \item \label{item:PhosphateStatePrep} 
    \textbf{Preparation of singlets of phosphate nuclear spins:} 
    We assume that phosphates' phosphorus nuclear spins can be prepared in the singlet state, 
$\ket{\Psi^{-}} = \frac{1}{\sqrt{2}}  \left(  \ket{01} - \ket{10}  \right)$
(see Sections 2.1 and 3.4 of \cite{NYH_19_Quantum}).
Such a singlet has been conjectured to form when
the enzyme pyrophosphatase hydrolyzes a diphosphate
into two phosphates~\cite{Fisher_15_Quantum,Fisher_18_Quantum}.

 %   By waiting for this singlet to decohere and then aligning spins in a magnetic field, we also assume it is possible to produce arbitrarely many phosphates aligned in the $\ket{0}$ state.\footnote{Because of thermal noise this operation may not be possible. I think it should be able to recover a biological Bell test without this operation, but I'm leaving it in for now because it simplifies much of the math.} 
 
    \item \label{item:PosnerStatePrep}
    \textbf{Controlled Posner formation:} 
    We assume that Posners can be formed with phosphates laid out in arbitrary arrangements, subject to the restrictions of the Posner's geometry (see Sections~3.1.2--3.1.4, 3.4, and 3.7 of~\cite{NYH_19_Quantum}).
This assumption may seem unreasonable. 
We can mitigate the unreasonableness slightly, because the assumption is required only for setting up the Posner Bell test.
The ability to detect and postselect on Posners in desired geometries would suffice.
    
    In particular, we assume the ability to create two Posners that share six singlets in a geometrically symmetric arrangement. This ``six singlets shared'' state is presented in Sec.~3.7 of \cite{NYH_19_Quantum} and reproduced here in Fig.~\ref{fig_entangledBindingSixSinglets}. 
We denote this state by $\ket{\psi^\ent_{AB}}$, wherein $A$ and $B$ label the Posners. % When we write the state using the common eigenbasis, the triples are given in the order $\ket{1,2,3} \ket{7,8,9} \ket{4,5,6} \ket{10,11,12}$ to more easily describe the entangled triples. When writing the state in the spin basis, subscript labels are used to indicate the spatial location of the states described.

    \item \label{item:Rotate Tau} 
    \textbf{$\tau$ rotations:} 
    A Posner has a sixfold rotational symmetry. Consider the operator that represents a rotation about the symmetry axis. The operator has eigenvalues $\tau = 0, \pm 1$.
    The Posner Bell test requires the ability to map some states of a Posner to another $\tau$ sector. 
    Such a rotation can be accomplished via multiple mathematical operations,
    we found via direct calculation~\cite{BeneWatts_19_Code}.
    One operation is a rotation of one of the Posner's qubits with 
    the Pauli $z$-operator, $\sigma_z$. 
    We focus on this implementation for concreteness.
    
    We assume the existence of an operator $\RPos(\theta)$ 
    that represents the desired rotation.
    For specificity, we assume that the rotated qubit is
    the qubit labeled 1 in~\cite{NYH_19_Quantum} 
    (see Fig.~\ref{fig_entangledBindingSixSinglets} in the present paper).
    % $R(\theta) = R_z(\theta)_1 \otimes I_{2-6}$ when written in the spin basis. 
    When we need to distinguish between Posners, 
    we use the notation $\RPos(\theta)_A$ to denote the operator 
    $\RPos(\theta)$ applied to Posner $A$. 
    
    $\RPos(\theta)$ can be effected physically, e.g., as described in
    Sections~3.4 and~3.8 of~\cite{NYH_19_Quantum}:
    Let phosphate 1 form a Posner $A$ with other phosphates
    and with calcium ions.
    Applying a magnetic field to $A$ will rotate all six qubits,
    with $[ R_z(\theta) ]^{\otimes 6}$.
    If the pH rises, Posner $A$ will likely hydrolyze, or break apart.
    The pH can then be lowered.
    Phosphate 1 can find new phosphates with which to form a Posner $B$.
    Posner $B$ will have undergone $\RPos(\theta)_B$.
    Other means of effecting $\RPos (\theta)$ may be possible.
    
    \item \label{item:Posner Binding} 
    \textbf{Posner-binding measurement:}
    Suppose that Posners $A$ and $B$ approach each other 
    such that their symmetry axes are parallel and point oppositely each other.
    We call this arrangement the \emph{prebinding orientation},
    following~\cite{NYH_19_Quantum}.
    Quantum-chemistry calculations suggest that the Posners can bind together~\cite{Swift_17_Posner}.
    This measurement projects Posners $A$ and $B$
    into the subspace labeled by $\tau_A + \tau_B = 0$ (if the Posners bind)
    or onto the orthogonal subspace (if the Posners do not). 
    Following \cite{NYH_19_Quantum}, we denote by $\Pi_{AB}$ the projector onto the Posner-binding subspace. 
    
    We assume that an experimentalist can observe the number of such bindings. A method is proposed in \cite{Fisher_17_Are}: Calcium indicators are added to the Posner-containing test tube. 
Bound-together Posners would move slowly, becoming susceptible to
attack by hydrogen ions H$^+$ and magnesium ions Mg$^{2+}$.
These ions could outcompete the positively charged calcium ions Ca$^{2+}$
in binding to the negatively charged phosphate ions PO$_4^{3-}$.
The invaders would hydrolyze the Posners, breaking the molecules into
their constituent ions.
The calcium indicators would bind to the calcium ions Ca$^{2+}$, then fluoresce.
An experimentalist could detect the fluorescence.

\end{enumerate} 

Our proposal sharpens the inspirational sketch, in~\cite{Fisher_17_Are}, of a test for entanglement between Posners.
There, Posners in different test tubes were imagined to share entanglement.
Posners in each test tube would bind, and each test tube would fluoresce.
The intensities of the test tubes' fluorescence were imagined
to exhibit correlations.
We add that, to infer that Posners shared entanglement,
one must observe not just any correlations between the intensities.
Some correlations produceable with entanglement
can be recapitulated with classical resources.
We begin constructing a means of observing nonclassical correlations,
using the macroscopic Bell inequality (Theorem~\ref{thm_marco_bell_ineq}).

% Figure: 2 Posners that share 6 singlets
\begin{figure}[hbt]
\centering
\includegraphics[width=.4\textwidth, clip=true]{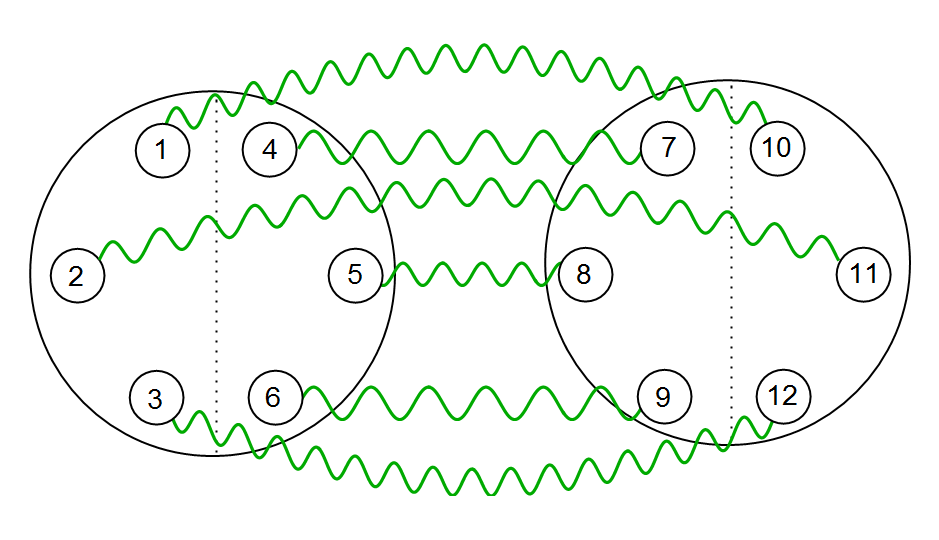}
\caption{\caphead{Symmetric ``six singlets shared'' configuration
of two Posner molecules:} 
Reprinted from~\cite{NYH_19_Quantum}, with permission from Elsevier.}
% Reference about permission: https://www.elsevier.com/about/policies/copyright/permissions
%% <-- stored in "Notes" folder --> Bell test --> Reproducing fig. from AoP public'n
\label{fig_entangledBindingSixSinglets}
\end{figure}
\subsubsection{Four facts that underlie the Posner Bell test}
\label{sec_Micro_Posner_4_Facts}

First, all the phosphorus nuclear spins in the state $\ket{\psi^\ent_{AB}}$ 
are in copies of $\ket{\Sing}$.
Each phosphate in Posner $A$ is entangled with
the phosphate at the corresponding position in Posner $B$.
As shown in Supplementary Note~\ref{app_CHSH_Backgrnd},
\begin{align}
    [ R_z(\theta_1)  \otimes  R_z(-\theta_2) ]  \ket{\Sing}
    = [ R_z(\theta_1 + \theta_2)  \otimes  \id ]  \ket{\Sing}.
\end{align}
The Posner analog (operation~\ref{item:Rotate Tau} in Supplementary Note~\ref{sec_Ops_For_Posner_Test}) has the form
\begin{align} \label{eq:movingZrotationsOnPos}
    [ \RPos(\theta_1)_A  \otimes  \RPos(-\theta_2)_B ]  
    \ket{\psi^\ent_{AB}} 
    =  \RPos(\theta_1 + \theta_2)_A    \ket{\psi^\ent_{AB}}.
\end{align}

Second, $\ket{\psi^\ent_{AB}}$ is a superposition of states 
in which the Posners' $\tau$ values sum to zero.
This fact was first pointed out in~\cite{NYH_19_Quantum}:
Consider Posners that occupy the state $\ket{\Psi^\ent_{AB}}$
and the prebinding orientation.
The Posners were observed to have
a unit probability of binding (under the assumptions of Fisher's model).

Third, suppose that Posners $A$ and $B$ occupy the state
$\ket{\psi^\ent_{AB}}$,
while Posners $A'$ and $B'$ occupy the state
$\ket{\psi^\ent_{A'B'}}$.
Suppose that $A$ assumes the prebinding orientation with $A'$
and that $B$ assumes the prebinding orientation with $B'$.
A pair's binding is represented with a bit value 0,
and a pair's not binding is represented with a 1.
The two pairs' bits have even parity in two cases,
if both pairs bind or both pairs fail to bind.
If the bits have even parity, the four-Posner state is projected with
\begin{align}
    \Pi_\Even 
    = \Pi_{AA'}\Pi_{BB'} + (\id - \Pi_{AA'})(\id - \Pi_{BB'}).
\end{align}

The binding-measurement outcomes have even parity, we claim:
\begin{align}
    \label{eq_Even_Parity}
    \abs{\Pi_\even  
    \ket{\psi^\ent_{AB}}   \ket{\psi^\ent_{A'B'}}  }^2 = 1 .
\end{align}
Either both pairs bind or both fail to bind.
Equation~\eqref{eq_Even_Parity} can be checked computationally~\cite{BeneWatts_19_Code}
and with the following logic:
In $\ket{\psi^\ent_{AB}}$, $\tau_A + \tau_B = 0$.
In $\ket{\psi^\ent_{A'B'}}$, $\tau_{A'} + \tau_{B'} = 0$.
Hence $\tau_{\rm total} = \tau_A + \tau_{A'} + \tau_B + \tau_{B'}  =  0$. 
In contrast, any state in the image of the projector 
$\mathbbm{1} - \Pi_\Even 
= \Pi_{AA'}(\mathbbm{1} - \Pi_{BB'}) + (\mathbbm{1} - \Pi_{AA'})\Pi_{BB'}$ 
has a $\tau_{\rm total} \neq 0$. 

Fourth, we continue to consider the four-Posner state 
$\ket{\psi^\ent_{AB}} \ket{\psi^\ent_{A'B'}}$.
Consider rotating a qubit with $\RPos(\theta_1)_A$.
The state of Posners $A$ and $B$ is rotated out of 
the $\tau_A + \tau_B = 0$ sector.
This claim can be checked via direct calculation~\cite{BeneWatts_19_Code}.
The even-parity-binding probability,
$\abs{\Pi_\Even
\RPos(\theta_1)_A \ket{\psi^\ent_{AB}}
\ket{\psi^\ent_{A'B'}}  }^2$,
decreases as the rotation angle $\theta_1$ grows. 
We solve for this relationship numerically~\cite{BeneWatts_19_Code}
and present the results in Fig.~\ref{fig:EvenBindingProbability}.

\begin{figure}[ht]
    \centering
    \includegraphics[width=.6\textwidth, clip=true]{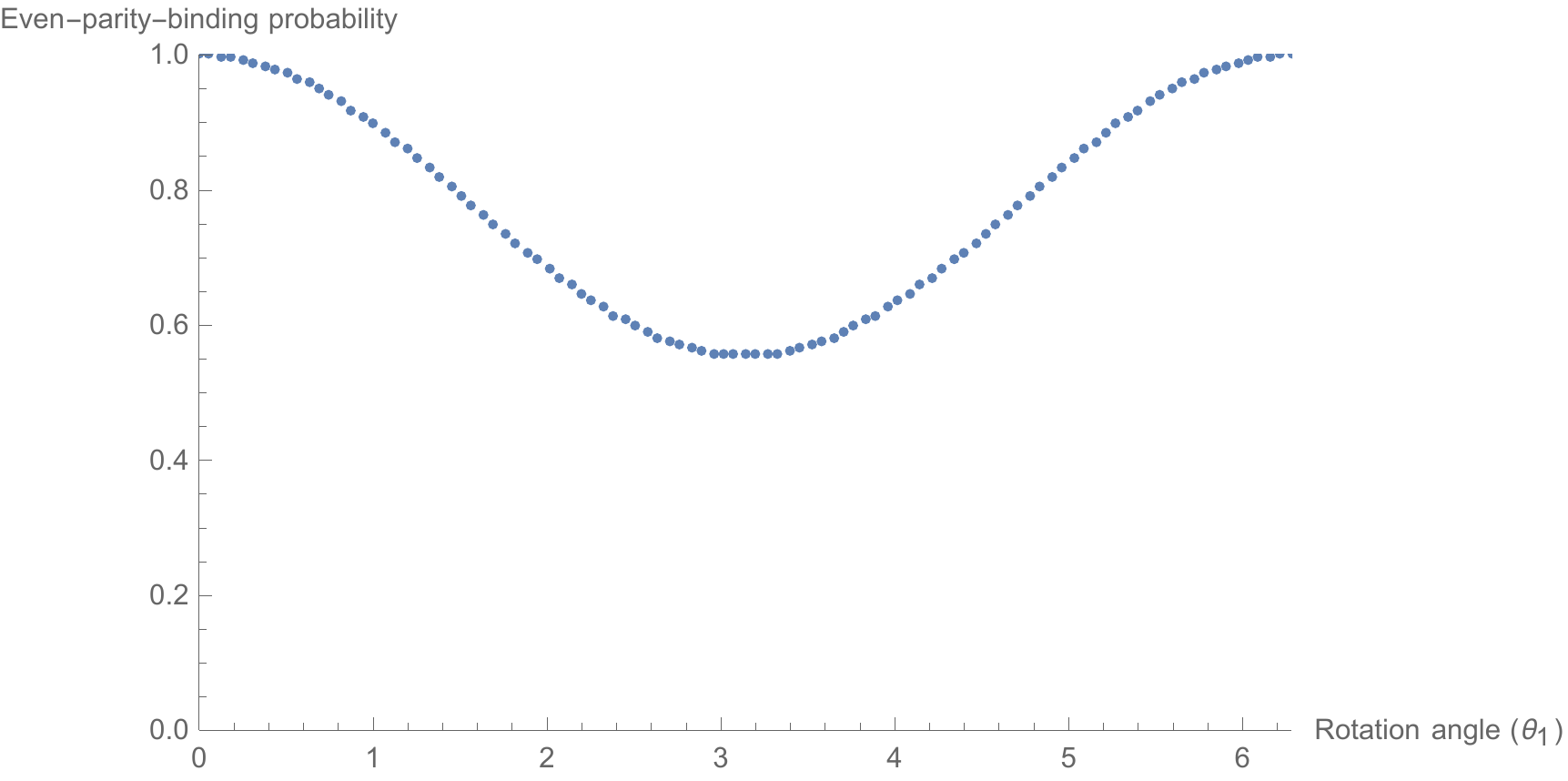}
    \caption{\caphead{Probability of even-parity binding vs. rotation angle:}
    Probability 
    $\left\lvert \Pi_\Even
\RPos(\theta_1)_A \ket{\psi^\ent_{AB}}
\ket{\psi^\ent_{A'B'}}  \right\rvert^2$
that Posner $A$ binds to Posner $A'$ while $B$ binds to $B'$,
after a qubit in $A$ is rotated through an angle $\theta_1$.}
    \label{fig:EvenBindingProbability}
\end{figure}

\subsection{Nonlocal game for a system of few Posners}
\label{subsec:PosnerTest}

We define the game in Sec.~\ref{sec_Micro_Posner_Procedure}.
The game is adapted from the original CHSH game
(Supplementary Note~\ref{app_CHSH_Backgrnd}):
As in the original game, the Posner game's players 
use singlets and rotations.
But Posner-holding players cannot measure $\sigma_x \otimes \sigma_x$.
So they share 12 singlets, rather than one,
and perform Posner-binding measurements,
rather than measuring Pauli operators.
Section~\ref{sec_Micro_Posner_Analyze} shows that
Posners can win the game with a superclassical probability,
in principle, if Fisher's model is correct.
This result is not obvious: The quantum operations 
undergone by Posners in Fisher's model are nonstandard
and might not contain a universal gate set~\cite{NYH_19_Quantum}.
Our proof rests, however, on fine control over the operations
conjectured to be implementable.

\subsubsection{Procedure for the Posner Bell test}
\label{sec_Micro_Posner_Procedure}
 
We consider two experimentalists, Alice and Bob.
Each can perform the operations described in Supplementary Note~\ref{subsec:PosnerPreliminaries}. 
Before the test begins, the players prepare two copies of 
the two Posner-state $\ket{\psi^\ent}$.
Posner $A$ shares entanglement with Posner $B$,
and $A'$ shares entanglement with $B'$.
Alice takes $A$ and $A'$, while Bob takes $B$ and $B'$.
Alice and Bob separate, and each experimentalist receives a question from a verifier. The experimentalists aim to produce responses that win the CHSH game 
(Supplementary Note~\ref{app_CHSH_Backgrnd}) with a superclassical probability. 

Alice and Bob's procedures are almost identical, but we describe Alice's approach first. If Alice receives the question $0$, she rotates Posner $A$ with
$\RPos(-\pi/8)_A$. If she receives a $1$, she rotates with $\RPos(3\pi/8)_A$. Afterward, she performs the Posner-binding measurement between Posners $A$ and $A'$. If the Posners bind, she sends the verifier a 1. If the Posners fail to bind, she sends a 0.

Bob's procedure is similar. 
However, he rotates Posner $B$ with $\RPos(\pi/8)_B$, 
given the question $y = 0$,
and rotates with $\RPos(-3\pi/8)_B$, given $y = 1$.
His binding measurement is on Posners $B$ and $B'$.

\subsubsection{Analysis of the Bell test for a few Posners}
\label{sec_Micro_Posner_Analyze}

If Posners share entanglement, we now show, 
they can outperform classical resources in the CHSH game, in principle,
under Fisher's assumptions.

\begin{claim}
Alice and Bob can win the nonlocal game for a few Posners 
with a superclassical probability of $\geq 79.5\%$,
given fine control over the operations implementable in principle in Fisher's model.
\end{claim}

\begin{proof}
Assume that Alice rotates Posner $A$ with an angle $\theta_1$, 
while Bob rotates $B$ with an angle $\theta_2$. By the test's construction, the experimentalists' probability of sending the verifier an even-parity response pair is
\begin{align}
    \abs{\Pi_\Even
    \left[ \RPos(\theta_1)_A  \otimes  \RPos(\theta_2)_B \right]
    \ket{\psi^\ent_{AB}} \ket{\psi^\ent_{A'B'}} }^2 
    = \abs{\Pi_\Even
    \left[ \RPos(\theta_1 - \theta_2)_A  \right]
    \ket{\psi^\ent_{AB}}  \ket{\psi^\ent_{A'B'}} }^2.
\end{align}
The equality follows from Eq.~\eqref{eq:movingZrotationsOnPos}. 
We evaluate this probability for all possible question pairs
in Table~\ref{tab:EvenBindingProbs}~\cite{BeneWatts_19_Code}.

\begin{table}[ht]
    \centering
    \begin{tabular}{c|c|c}
        Questions 
        &$\theta_1 - \theta_2$ 
        & Even-parity-response probability\\ \hline
        00 &$-\pi/4$  & 0.934 \\
        01 & $\pi/4$  & 0.934  \\
        10 & $\pi/4$  & 0.934  \\
        11 & $3\pi/4$ & 0.620 
    \end{tabular}
    \caption{The probability that players provide 
    an even-parity response pair $ab$
    to the possible question pair $xy$ 
    in the few-Posner CHSH game.
    The first bit in the ``questions'' column labels the question $x$ sent to Alice.
    The second bit labels the question $y$ sent to Bob.}
    \label{tab:EvenBindingProbs}
\end{table}

In each round of the CHSH game, 
the verifier selects the question pair uniformly randomly.
The winning response pairs for the question pairs $00$, $01$, and $10$ 
have even parity.
The winning response pair for the question pair $11$ has odd parity. 
Therefore, the overall win probability for the few-Posner game is
\begin{align}
    \frac{3}{4} \times 0.934 + \frac{1}{4} \times (1-0.620) > 0.795. 
\end{align}
\end{proof}
The angles used by Alice and Bob (Sec.~\ref{sec_Micro_Posner_Procedure})
differ from the angles used in the conventional CHSH game
(Supplementary Note~\ref{app_CHSH_Backgrnd}).
The reason is, in the conventional game, Alice and Bob share a singlet, 
$\ket{ \Sing }  :=  \frac{1}{ \sqrt{2} } ( \ket{01} - \ket{10} )$.
In the Posner game, Alice and Bob share entangled a pair of Posners.
For the game's purposes, the Posners' $\tau$ degrees of freedom
resemble qubits in the maximally entangled state
$\ket{ \Phi^+ }  :=  \frac{1}{ \sqrt{2} } ( \ket{00}  +  \ket{11} )$.
If Alice and Bob shared a copy of $\ket{ \Phi^+ }$
in the conventional CHSH game,
they would use the rotations in Supplementary Note~\ref{sec_Micro_Posner_Procedure}.
This small-scale Posner Bell test informs our macroscopic Posner Bell test, introduced next.

\subsection{Sketch of macroscopic Posner Bell test}
\label{sec_Posner_Test_Actual}

We envision a three-dimensional tank of aqueous fluid.
A glass plate coincides with the $xy$-plane.
Along the $x$-axis is a slit in which pyrophosphatase enzymes are lodged.
Diphosphate ions are poured above the $xy$-plane
and waft downward on currents propelled from above.

The diphosphates can traverse the plane only through the enzymes.
Suppose that an enzyme hydrolyzes a diphosphate,
breaking the diphosphate into two phosphates (PO$_4^{3-}$).
The phosphorus ($^{31}$P) nuclear spins form a singlet,
$\frac{1}{ \sqrt{2} } ( \ket{ \uparrow \downarrow } - \ket{ \downarrow \uparrow } )$,
according to the dynamical selection rule
posited by Fisher and Radzihovsky~\cite{Fisher_15_Quantum,Fisher_18_Quantum}.
We assume that, during the hydrolyzation and release,
the enzyme changes shape,
such that the diphosphate can enter the fluid below the $xy$-plane.

Below the $xy$-plane, two currents sweep fluid away from the $x$-axis.
One current sweeps toward the $+y$-axis, where Alice collects the fluid.
The other current sweeps toward the $-y$-axis, where Bob collects the fluid.
Alice and Bob share singlets.
Alice and Bob add calcium ions, Ca$^{2+}$, to their fluids.
Posners form.
Each players divides his/her fluid among 3 test tubes.
Alice holds test tubes $A_1$, $A_2$, and $A_3$.
Bob holds test tubes $B_1$, $B_2$, and $B_3$.

A verifier sends a question $x = 0, 1$ to Alice and a question $y = 0, 1$ to Bob
(Supplementary Note~\ref{app_Nonlocal_Game}).
Alice and Bob agreed, before collecting the Posners,
to follow the strategy for winning the few-Posner Bell test
with a superclassical probability (Supplementary Note~\ref{subsec:PosnerTest}): 
Upon receiving $x = 0$, Alice applies a magnetic field to test tube 1. The field implements the rotation $R_z(-\pi/8)$ on all the Posners in the test tube. 
If Alice receives $x = 1$, she applies a field that implements $R_z(3\pi/8)$. 
Bob's strategy is similar but involves opposite signs:
Given $y = 0$, he applies a magnetic field that implements $R_z(\pi/8)$ on all the qubits in his 1 test tube. Given $y = 1$, he applies $R_z(-3\pi/8)$. 

Each player lowers the pH in test tubes $1$ and $2$.
The H$^+$ ions outcompete positively charged Ca$^{2+}$ ions
in binding to the negatively charge phosphorus ions (PO$_4^{3-}$).
The protons hydrolyze the Posners.
Each player mixes his/her $1$ and $2$ test tubes, then raises the pH.
New Posners form. 
Some Posners contain phosphorus nuclei whose spins are rotated
relative to the spins of the Posner's other phosphorus nuclei.

Each player mixes his/her combined $1$ and $2$ test tubes with test tube $3$,
whose qubits have not been rotated. 
Then, each player adds calcium indicators to the test tube
and lowers the pH.
Posners are hoped to approach each other
in the prebinding orientation.
In Fisher's model, some fraction of these Posners will bind.
The fraction depends on the entanglement
and on the binding of Posners in the other player's test tube.
The bound-together Posners move slowly,
forming easy targets for H$^+$ ions.
The ions hydrolyze the bound-together Posners, flooding the test tubes with calcium.
The calcium binds to the calcium indicators, which fluoresce.

Alice and Bob measure the fluorescence's intensity. After completing many trials, 
they compute the covariances between their intensities, 
then estimate the macroscopic Bell parameter
[Eq.~\eqref{eq_Bell_Param}]. 
A superclassical value (Sec.~\ref{thm_marco_bell_ineq}) 
certifies entanglement between Posners,
if the experiment satisfies the assumptions in Sec.~\ref{sec_Setup}.
The experiment sketched here likely does not satisfy the assumptions.
We delineate reasons, and opportunities for improving the sketch,
in the next section.

\subsection{Analysis of sketch of Posner Bell test}
\label{sec_Posner_Analysis}

Much work remains to be done 
to shore up the Posner Bell test theoretically
and to ensure its experimental feasibility.
First, the macroscopic Bell inequality needs extending. 
Theorem~\ref{thm_marco_bell_ineq} relies on
each particle's interacting with, at most, one other particle.
Each Posner can share entanglement with
up to six other Posners.
The extension from one to six requires a change to the inequality
but maintains interactions' locality.

Second, much could go awry during an implementation of 
the protocol in Sec.~\ref{sec_Posner_Test_Actual}.
A not-necessarily-complete list of loopholes include the following:
(i) Enzymes might release separated phosphates 
into the fluid above the glass plane.
(ii) A current could sweep two entangled-together phosphates
toward the same test tube.
(iii) The phosphates will assume random locations in the Posners.
The ``six singlets shared'' states 
(Fig.~\ref{fig_entangledBindingSixSinglets})
might form rarely.
(iv) The entangled phosphates spend time outside
the Posners conjectured to protect coherence.
The phosphorus nuclear spins might decohere 
before Alice and Bob can complete their trial.
Time scales must be estimated and compared, as in
Section~3.8 and App.~K of~\cite{NYH_19_Quantum}. 
(v) Posners in the same test tube might have a low probability
of assuming the prebinding orientation.
The overall binding rate might therefore be too low.
(vi) H$^+$ ions can hydrolyze not only bound-together Posners,
but also individual Posners.
Individual Posners' hydrolyzation will add noise to
the intensity measurements.
\end{appendices}

%
% Bibliography
%
\bibliographystyle{apsrev4-1}
\bibliography{Antman_Bib}

\begin{thebibliography}{86}%
\makeatletter
\providecommand \@ifxundefined [1]{%
 \@ifx{#1\undefined}
}%
\providecommand \@ifnum [1]{%
 \ifnum #1\expandafter \@firstoftwo
 \else \expandafter \@secondoftwo
 \fi
}%
\providecommand \@ifx [1]{%
 \ifx #1\expandafter \@firstoftwo
 \else \expandafter \@secondoftwo
 \fi
}%
\providecommand \natexlab [1]{#1}%
\providecommand \enquote  [1]{``#1''}%
\providecommand \bibnamefont  [1]{#1}%
\providecommand \bibfnamefont [1]{#1}%
\providecommand \citenamefont [1]{#1}%
\providecommand \href@noop [0]{\@secondoftwo}%
\providecommand \href [0]{\begingroup \@sanitize@url \@href}%
\providecommand \@href[1]{\@@startlink{#1}\@@href}%
\providecommand \@@href[1]{\endgroup#1\@@endlink}%
\providecommand \@sanitize@url [0]{\catcode `\\12\catcode `\$12\catcode
  `\&12\catcode `\#12\catcode `\^12\catcode `\_12\catcode `\%12\relax}%
\providecommand \@@startlink[1]{}%
\providecommand \@@endlink[0]{}%
\providecommand \url  [0]{\begingroup\@sanitize@url \@url }%
\providecommand \@url [1]{\endgroup\@href {#1}{\urlprefix }}%
\providecommand \urlprefix  [0]{URL }%
\providecommand \Eprint [0]{\href }%
\providecommand \doibase [0]{http://dx.doi.org/}%
\providecommand \selectlanguage [0]{\@gobble}%
\providecommand \bibinfo  [0]{\@secondoftwo}%
\providecommand \bibfield  [0]{\@secondoftwo}%
\providecommand \translation [1]{[#1]}%
\providecommand \BibitemOpen [0]{}%
\providecommand \bibitemStop [0]{}%
\providecommand \bibitemNoStop [0]{.\EOS\space}%
\providecommand \EOS [0]{\spacefactor3000\relax}%
\providecommand \BibitemShut  [1]{\csname bibitem#1\endcsname}%
\let\auto@bib@innerbib\@empty
%</preamble>
\bibitem [{\citenamefont {Arndt}\ \emph {et~al.}(1999)\citenamefont {Arndt},
  \citenamefont {Nairz}, \citenamefont {Vos-Andreae}, \citenamefont {Keller},
  \citenamefont {van~der Zouw},\ and\ \citenamefont
  {Zeilinger}}]{Arndt_99_Wave}%
  \BibitemOpen
  \bibfield  {author} {\bibinfo {author} {\bibfnamefont {M.}~\bibnamefont
  {Arndt}}, \bibinfo {author} {\bibfnamefont {O.}~\bibnamefont {Nairz}},
  \bibinfo {author} {\bibfnamefont {J.}~\bibnamefont {Vos-Andreae}}, \bibinfo
  {author} {\bibfnamefont {C.}~\bibnamefont {Keller}}, \bibinfo {author}
  {\bibfnamefont {G.}~\bibnamefont {van~der Zouw}}, \ and\ \bibinfo {author}
  {\bibfnamefont {A.}~\bibnamefont {Zeilinger}},\ }\href {\doibase
  10.1038/44348} {\bibfield  {journal} {\bibinfo  {journal} {Nature}\ }\textbf
  {\bibinfo {volume} {401}},\ \bibinfo {pages} {680} (\bibinfo {year}
  {1999})}\BibitemShut {NoStop}%
\bibitem [{\citenamefont {Leggett}(2002)}]{Leggett_02_Testing}%
  \BibitemOpen
  \bibfield  {author} {\bibinfo {author} {\bibfnamefont {A.~J.}\ \bibnamefont
  {Leggett}},\ }\href {\doibase 10.1088/0953-8984/14/15/201} {\bibfield
  {journal} {\bibinfo  {journal} {Journal of Physics: Condensed Matter}\
  }\textbf {\bibinfo {volume} {14}},\ \bibinfo {pages} {R415} (\bibinfo {year}
  {2002})}\BibitemShut {NoStop}%
\bibitem [{\citenamefont {Blencowe}(2004)}]{Blencowe_04_Nanomechanical}%
  \BibitemOpen
  \bibfield  {author} {\bibinfo {author} {\bibfnamefont {M.}~\bibnamefont
  {Blencowe}},\ }\href {\doibase 10.1126/science.1095768} {\bibfield  {journal}
  {\bibinfo  {journal} {Science}\ }\textbf {\bibinfo {volume} {304}},\ \bibinfo
  {pages} {56} (\bibinfo {year} {2004})},\ \Eprint
  {http://arxiv.org/abs/https://science.sciencemag.org/content/304/5667/56.full.pdf}
  {https://science.sciencemag.org/content/304/5667/56.full.pdf} \BibitemShut
  {NoStop}%
\bibitem [{\citenamefont {Gerlich}\ \emph {et~al.}(2011)\citenamefont
  {Gerlich}, \citenamefont {Eibenberger}, \citenamefont {Tomandl},
  \citenamefont {Nimmrichter}, \citenamefont {Hornberger}, \citenamefont
  {Fagan}, \citenamefont {T{\"u}xen}, \citenamefont {Mayor},\ and\
  \citenamefont {Arndt}}]{Gerlich_11_Quantum}%
  \BibitemOpen
  \bibfield  {author} {\bibinfo {author} {\bibfnamefont {S.}~\bibnamefont
  {Gerlich}}, \bibinfo {author} {\bibfnamefont {S.}~\bibnamefont
  {Eibenberger}}, \bibinfo {author} {\bibfnamefont {M.}~\bibnamefont
  {Tomandl}}, \bibinfo {author} {\bibfnamefont {S.}~\bibnamefont
  {Nimmrichter}}, \bibinfo {author} {\bibfnamefont {K.}~\bibnamefont
  {Hornberger}}, \bibinfo {author} {\bibfnamefont {P.~J.}\ \bibnamefont
  {Fagan}}, \bibinfo {author} {\bibfnamefont {J.}~\bibnamefont {T{\"u}xen}},
  \bibinfo {author} {\bibfnamefont {M.}~\bibnamefont {Mayor}}, \ and\ \bibinfo
  {author} {\bibfnamefont {M.}~\bibnamefont {Arndt}},\ }\href {\doibase
  10.1038/ncomms1263} {\bibfield  {journal} {\bibinfo  {journal} {Nature
  Communications}\ }\textbf {\bibinfo {volume} {2}},\ \bibinfo {pages} {263}
  (\bibinfo {year} {2011})}\BibitemShut {NoStop}%
\bibitem [{\citenamefont {Wollman}\ \emph {et~al.}(2015)\citenamefont
  {Wollman}, \citenamefont {Lei}, \citenamefont {Weinstein}, \citenamefont
  {Suh}, \citenamefont {Kronwald}, \citenamefont {Marquardt}, \citenamefont
  {Clerk},\ and\ \citenamefont {Schwab}}]{Wollman_15_Quantum}%
  \BibitemOpen
  \bibfield  {author} {\bibinfo {author} {\bibfnamefont {E.~E.}\ \bibnamefont
  {Wollman}}, \bibinfo {author} {\bibfnamefont {C.~U.}\ \bibnamefont {Lei}},
  \bibinfo {author} {\bibfnamefont {A.~J.}\ \bibnamefont {Weinstein}}, \bibinfo
  {author} {\bibfnamefont {J.}~\bibnamefont {Suh}}, \bibinfo {author}
  {\bibfnamefont {A.}~\bibnamefont {Kronwald}}, \bibinfo {author}
  {\bibfnamefont {F.}~\bibnamefont {Marquardt}}, \bibinfo {author}
  {\bibfnamefont {A.~A.}\ \bibnamefont {Clerk}}, \ and\ \bibinfo {author}
  {\bibfnamefont {K.~C.}\ \bibnamefont {Schwab}},\ }\href {\doibase
  10.1126/science.aac5138} {\bibfield  {journal} {\bibinfo  {journal}
  {Science}\ }\textbf {\bibinfo {volume} {349}},\ \bibinfo {pages} {952}
  (\bibinfo {year} {2015})},\ \Eprint
  {http://arxiv.org/abs/https://science.sciencemag.org/content/349/6251/952.full.pdf}
  {https://science.sciencemag.org/content/349/6251/952.full.pdf} \BibitemShut
  {NoStop}%
\bibitem [{\citenamefont {Kaltenbaek}\ \emph {et~al.}(2016)\citenamefont
  {Kaltenbaek}, \citenamefont {Aspelmeyer}, \citenamefont {Barker},
  \citenamefont {Bassi}, \citenamefont {Bateman}, \citenamefont {Bongs},
  \citenamefont {Bose}, \citenamefont {Braxmaier}, \citenamefont {Brukner},
  \citenamefont {Christophe}, \citenamefont {Chwalla}, \citenamefont {Cohadon},
  \citenamefont {Cruise}, \citenamefont {Curceanu}, \citenamefont {Dholakia},
  \citenamefont {Di{\'o}si}, \citenamefont {D{\"o}ringshoff}, \citenamefont
  {Ertmer}, \citenamefont {Gieseler}, \citenamefont {G{\"u}rlebeck},
  \citenamefont {Hechenblaikner}, \citenamefont {Heidmann}, \citenamefont
  {Herrmann}, \citenamefont {Hossenfelder}, \citenamefont {Johann},
  \citenamefont {Kiesel}, \citenamefont {Kim}, \citenamefont {L{\"a}mmerzahl},
  \citenamefont {Lambrecht}, \citenamefont {Mazilu}, \citenamefont {Milburn},
  \citenamefont {M{\"u}ller}, \citenamefont {Novotny}, \citenamefont
  {Paternostro}, \citenamefont {Peters}, \citenamefont {Pikovski},
  \citenamefont {Pilan~Zanoni}, \citenamefont {Rasel}, \citenamefont {Reynaud},
  \citenamefont {Riedel}, \citenamefont {Rodrigues}, \citenamefont {Rondin},
  \citenamefont {Roura}, \citenamefont {Schleich}, \citenamefont
  {Schmiedmayer}, \citenamefont {Schuldt}, \citenamefont {Schwab},
  \citenamefont {Tajmar}, \citenamefont {Tino}, \citenamefont {Ulbricht},
  \citenamefont {Ursin},\ and\ \citenamefont
  {Vedral}}]{Kaltenbaek_16_Macroscopic}%
  \BibitemOpen
  \bibfield  {author} {\bibinfo {author} {\bibfnamefont {R.}~\bibnamefont
  {Kaltenbaek}}, \bibinfo {author} {\bibfnamefont {M.}~\bibnamefont
  {Aspelmeyer}}, \bibinfo {author} {\bibfnamefont {P.~F.}\ \bibnamefont
  {Barker}}, \bibinfo {author} {\bibfnamefont {A.}~\bibnamefont {Bassi}},
  \bibinfo {author} {\bibfnamefont {J.}~\bibnamefont {Bateman}}, \bibinfo
  {author} {\bibfnamefont {K.}~\bibnamefont {Bongs}}, \bibinfo {author}
  {\bibfnamefont {S.}~\bibnamefont {Bose}}, \bibinfo {author} {\bibfnamefont
  {C.}~\bibnamefont {Braxmaier}}, \bibinfo {author} {\bibfnamefont
  {C.}~\bibnamefont {Brukner}}, \bibinfo {author} {\bibfnamefont
  {B.}~\bibnamefont {Christophe}}, \bibinfo {author} {\bibfnamefont
  {M.}~\bibnamefont {Chwalla}}, \bibinfo {author} {\bibfnamefont {P.-F.}\
  \bibnamefont {Cohadon}}, \bibinfo {author} {\bibfnamefont {A.~M.}\
  \bibnamefont {Cruise}}, \bibinfo {author} {\bibfnamefont {C.}~\bibnamefont
  {Curceanu}}, \bibinfo {author} {\bibfnamefont {K.}~\bibnamefont {Dholakia}},
  \bibinfo {author} {\bibfnamefont {L.}~\bibnamefont {Di{\'o}si}}, \bibinfo
  {author} {\bibfnamefont {K.}~\bibnamefont {D{\"o}ringshoff}}, \bibinfo
  {author} {\bibfnamefont {W.}~\bibnamefont {Ertmer}}, \bibinfo {author}
  {\bibfnamefont {J.}~\bibnamefont {Gieseler}}, \bibinfo {author}
  {\bibfnamefont {N.}~\bibnamefont {G{\"u}rlebeck}}, \bibinfo {author}
  {\bibfnamefont {G.}~\bibnamefont {Hechenblaikner}}, \bibinfo {author}
  {\bibfnamefont {A.}~\bibnamefont {Heidmann}}, \bibinfo {author}
  {\bibfnamefont {S.}~\bibnamefont {Herrmann}}, \bibinfo {author}
  {\bibfnamefont {S.}~\bibnamefont {Hossenfelder}}, \bibinfo {author}
  {\bibfnamefont {U.}~\bibnamefont {Johann}}, \bibinfo {author} {\bibfnamefont
  {N.}~\bibnamefont {Kiesel}}, \bibinfo {author} {\bibfnamefont
  {M.}~\bibnamefont {Kim}}, \bibinfo {author} {\bibfnamefont {C.}~\bibnamefont
  {L{\"a}mmerzahl}}, \bibinfo {author} {\bibfnamefont {A.}~\bibnamefont
  {Lambrecht}}, \bibinfo {author} {\bibfnamefont {M.}~\bibnamefont {Mazilu}},
  \bibinfo {author} {\bibfnamefont {G.~J.}\ \bibnamefont {Milburn}}, \bibinfo
  {author} {\bibfnamefont {H.}~\bibnamefont {M{\"u}ller}}, \bibinfo {author}
  {\bibfnamefont {L.}~\bibnamefont {Novotny}}, \bibinfo {author} {\bibfnamefont
  {M.}~\bibnamefont {Paternostro}}, \bibinfo {author} {\bibfnamefont
  {A.}~\bibnamefont {Peters}}, \bibinfo {author} {\bibfnamefont
  {I.}~\bibnamefont {Pikovski}}, \bibinfo {author} {\bibfnamefont
  {A.}~\bibnamefont {Pilan~Zanoni}}, \bibinfo {author} {\bibfnamefont {E.~M.}\
  \bibnamefont {Rasel}}, \bibinfo {author} {\bibfnamefont {S.}~\bibnamefont
  {Reynaud}}, \bibinfo {author} {\bibfnamefont {C.~J.}\ \bibnamefont {Riedel}},
  \bibinfo {author} {\bibfnamefont {M.}~\bibnamefont {Rodrigues}}, \bibinfo
  {author} {\bibfnamefont {L.}~\bibnamefont {Rondin}}, \bibinfo {author}
  {\bibfnamefont {A.}~\bibnamefont {Roura}}, \bibinfo {author} {\bibfnamefont
  {W.~P.}\ \bibnamefont {Schleich}}, \bibinfo {author} {\bibfnamefont
  {J.}~\bibnamefont {Schmiedmayer}}, \bibinfo {author} {\bibfnamefont
  {T.}~\bibnamefont {Schuldt}}, \bibinfo {author} {\bibfnamefont {K.~C.}\
  \bibnamefont {Schwab}}, \bibinfo {author} {\bibfnamefont {M.}~\bibnamefont
  {Tajmar}}, \bibinfo {author} {\bibfnamefont {G.~M.}\ \bibnamefont {Tino}},
  \bibinfo {author} {\bibfnamefont {H.}~\bibnamefont {Ulbricht}}, \bibinfo
  {author} {\bibfnamefont {R.}~\bibnamefont {Ursin}}, \ and\ \bibinfo {author}
  {\bibfnamefont {V.}~\bibnamefont {Vedral}},\ }\href {\doibase
  10.1140/epjqt/s40507-016-0043-7} {\bibfield  {journal} {\bibinfo  {journal}
  {EPJ Quantum Technology}\ }\textbf {\bibinfo {volume} {3}},\ \bibinfo {pages}
  {5} (\bibinfo {year} {2016})}\BibitemShut {NoStop}%
\bibitem [{\citenamefont {Bose}\ \emph {et~al.}(2017)\citenamefont {Bose},
  \citenamefont {Mazumdar}, \citenamefont {Morley}, \citenamefont {Ulbricht},
  \citenamefont {Toro{\v{s}}}, \citenamefont {Paternostro}, \citenamefont
  {Geraci}, \citenamefont {Barker}, \citenamefont {Kim},\ and\ \citenamefont
  {Milburn}}]{bose2017spin}%
  \BibitemOpen
  \bibfield  {author} {\bibinfo {author} {\bibfnamefont {S.}~\bibnamefont
  {Bose}}, \bibinfo {author} {\bibfnamefont {A.}~\bibnamefont {Mazumdar}},
  \bibinfo {author} {\bibfnamefont {G.~W.}\ \bibnamefont {Morley}}, \bibinfo
  {author} {\bibfnamefont {H.}~\bibnamefont {Ulbricht}}, \bibinfo {author}
  {\bibfnamefont {M.}~\bibnamefont {Toro{\v{s}}}}, \bibinfo {author}
  {\bibfnamefont {M.}~\bibnamefont {Paternostro}}, \bibinfo {author}
  {\bibfnamefont {A.~A.}\ \bibnamefont {Geraci}}, \bibinfo {author}
  {\bibfnamefont {P.~F.}\ \bibnamefont {Barker}}, \bibinfo {author}
  {\bibfnamefont {M.}~\bibnamefont {Kim}}, \ and\ \bibinfo {author}
  {\bibfnamefont {G.}~\bibnamefont {Milburn}},\ }\href@noop {} {\bibfield
  {journal} {\bibinfo  {journal} {Physical review letters}\ }\textbf {\bibinfo
  {volume} {119}},\ \bibinfo {pages} {240401} (\bibinfo {year}
  {2017})}\BibitemShut {NoStop}%
\bibitem [{\citenamefont {Yao}\ and\ \citenamefont
  {Hughes}(2009)}]{Yao_09_Quantum}%
  \BibitemOpen
  \bibfield  {author} {\bibinfo {author} {\bibfnamefont {P.}~\bibnamefont
  {Yao}}\ and\ \bibinfo {author} {\bibfnamefont {S.}~\bibnamefont {Hughes}},\
  }\href {\doibase 10.1364/OE.17.011505} {\bibfield  {journal} {\bibinfo
  {journal} {Opt. Express}\ }\textbf {\bibinfo {volume} {17}},\ \bibinfo
  {pages} {11505} (\bibinfo {year} {2009})}\BibitemShut {NoStop}%
\bibitem [{\citenamefont {Marinkovi\ifmmode~\acute{c}\else \'{c}\fi{}}\ \emph
  {et~al.}(2018)\citenamefont {Marinkovi\ifmmode~\acute{c}\else \'{c}\fi{}},
  \citenamefont {Wallucks}, \citenamefont {Riedinger}, \citenamefont {Hong},
  \citenamefont {Aspelmeyer},\ and\ \citenamefont
  {Gr\"oblacher}}]{Marinkovic_18_Optomechanical}%
  \BibitemOpen
  \bibfield  {author} {\bibinfo {author} {\bibfnamefont {I.}~\bibnamefont
  {Marinkovi\ifmmode~\acute{c}\else \'{c}\fi{}}}, \bibinfo {author}
  {\bibfnamefont {A.}~\bibnamefont {Wallucks}}, \bibinfo {author}
  {\bibfnamefont {R.}~\bibnamefont {Riedinger}}, \bibinfo {author}
  {\bibfnamefont {S.}~\bibnamefont {Hong}}, \bibinfo {author} {\bibfnamefont
  {M.}~\bibnamefont {Aspelmeyer}}, \ and\ \bibinfo {author} {\bibfnamefont
  {S.}~\bibnamefont {Gr\"oblacher}},\ }\href {\doibase
  10.1103/PhysRevLett.121.220404} {\bibfield  {journal} {\bibinfo  {journal}
  {Phys. Rev. Lett.}\ }\textbf {\bibinfo {volume} {121}},\ \bibinfo {pages}
  {220404} (\bibinfo {year} {2018})}\BibitemShut {NoStop}%
\bibitem [{\citenamefont {Schmied}\ \emph {et~al.}(2016)\citenamefont
  {Schmied}, \citenamefont {Bancal}, \citenamefont {Allard}, \citenamefont
  {Fadel}, \citenamefont {Scarani}, \citenamefont {Treutlein},\ and\
  \citenamefont {Sangouard}}]{Schmid_16_Bell}%
  \BibitemOpen
  \bibfield  {author} {\bibinfo {author} {\bibfnamefont {R.}~\bibnamefont
  {Schmied}}, \bibinfo {author} {\bibfnamefont {J.-D.}\ \bibnamefont {Bancal}},
  \bibinfo {author} {\bibfnamefont {B.}~\bibnamefont {Allard}}, \bibinfo
  {author} {\bibfnamefont {M.}~\bibnamefont {Fadel}}, \bibinfo {author}
  {\bibfnamefont {V.}~\bibnamefont {Scarani}}, \bibinfo {author} {\bibfnamefont
  {P.}~\bibnamefont {Treutlein}}, \ and\ \bibinfo {author} {\bibfnamefont
  {N.}~\bibnamefont {Sangouard}},\ }\href {\doibase 10.1126/science.aad8665}
  {\bibfield  {journal} {\bibinfo  {journal} {Science}\ }\textbf {\bibinfo
  {volume} {352}},\ \bibinfo {pages} {441} (\bibinfo {year} {2016})},\ \Eprint
  {http://arxiv.org/abs/https://science.sciencemag.org/content/352/6284/441.full.pdf}
  {https://science.sciencemag.org/content/352/6284/441.full.pdf} \BibitemShut
  {NoStop}%
\bibitem [{\citenamefont {Reid}\ \emph {et~al.}(2002)\citenamefont {Reid},
  \citenamefont {Munro},\ and\ \citenamefont {De~Martini}}]{Reid_02_Violation}%
  \BibitemOpen
  \bibfield  {author} {\bibinfo {author} {\bibfnamefont {M.~D.}\ \bibnamefont
  {Reid}}, \bibinfo {author} {\bibfnamefont {W.~J.}\ \bibnamefont {Munro}}, \
  and\ \bibinfo {author} {\bibfnamefont {F.}~\bibnamefont {De~Martini}},\
  }\href {\doibase 10.1103/PhysRevA.66.033801} {\bibfield  {journal} {\bibinfo
  {journal} {Phys. Rev. A}\ }\textbf {\bibinfo {volume} {66}},\ \bibinfo
  {pages} {033801} (\bibinfo {year} {2002})}\BibitemShut {NoStop}%
\bibitem [{\citenamefont {Mermin}(1980)}]{Mermin_80_Quantum}%
  \BibitemOpen
  \bibfield  {author} {\bibinfo {author} {\bibfnamefont {N.~D.}\ \bibnamefont
  {Mermin}},\ }\href {\doibase 10.1103/PhysRevD.22.356} {\bibfield  {journal}
  {\bibinfo  {journal} {Phys. Rev. D}\ }\textbf {\bibinfo {volume} {22}},\
  \bibinfo {pages} {356} (\bibinfo {year} {1980})}\BibitemShut {NoStop}%
\bibitem [{\citenamefont {Drummond}(1983)}]{Drummond_83_Violations}%
  \BibitemOpen
  \bibfield  {author} {\bibinfo {author} {\bibfnamefont {P.~D.}\ \bibnamefont
  {Drummond}},\ }\href {\doibase 10.1103/PhysRevLett.50.1407} {\bibfield
  {journal} {\bibinfo  {journal} {Phys. Rev. Lett.}\ }\textbf {\bibinfo
  {volume} {50}},\ \bibinfo {pages} {1407} (\bibinfo {year}
  {1983})}\BibitemShut {NoStop}%
\bibitem [{\citenamefont {Mermin}(1990)}]{Mermin_90_Extreme}%
  \BibitemOpen
  \bibfield  {author} {\bibinfo {author} {\bibfnamefont {N.~D.}\ \bibnamefont
  {Mermin}},\ }\href {\doibase 10.1103/PhysRevLett.65.1838} {\bibfield
  {journal} {\bibinfo  {journal} {Phys. Rev. Lett.}\ }\textbf {\bibinfo
  {volume} {65}},\ \bibinfo {pages} {1838} (\bibinfo {year}
  {1990})}\BibitemShut {NoStop}%
\bibitem [{\citenamefont {Ardehali}(1992)}]{Ardehali_92_Bell}%
  \BibitemOpen
  \bibfield  {author} {\bibinfo {author} {\bibfnamefont {M.}~\bibnamefont
  {Ardehali}},\ }\href {\doibase 10.1103/PhysRevA.46.5375} {\bibfield
  {journal} {\bibinfo  {journal} {Phys. Rev. A}\ }\textbf {\bibinfo {volume}
  {46}},\ \bibinfo {pages} {5375} (\bibinfo {year} {1992})}\BibitemShut
  {NoStop}%
\bibitem [{\citenamefont {Belinski{\u{\i}}}\ and\ \citenamefont
  {Klyshko}(1993)}]{Belinskii_93_Interference}%
  \BibitemOpen
  \bibfield  {author} {\bibinfo {author} {\bibfnamefont {A.~V.}\ \bibnamefont
  {Belinski{\u{\i}}}}\ and\ \bibinfo {author} {\bibfnamefont {D.~N.}\
  \bibnamefont {Klyshko}},\ }\href {\doibase 10.1070/pu1993v036n08abeh002299}
  {\bibfield  {journal} {\bibinfo  {journal} {Physics-Uspekhi}\ }\textbf
  {\bibinfo {volume} {36}},\ \bibinfo {pages} {653} (\bibinfo {year}
  {1993})}\BibitemShut {NoStop}%
\bibitem [{\citenamefont {\ifmmode~\dot{Z}\else \.{Z}\fi{}ukowski}\ and\
  \citenamefont {Brukner}(2002)}]{Zukowski_02_Bell's}%
  \BibitemOpen
  \bibfield  {author} {\bibinfo {author} {\bibfnamefont {M.}~\bibnamefont
  {\ifmmode~\dot{Z}\else \.{Z}\fi{}ukowski}}\ and\ \bibinfo {author}
  {\bibfnamefont {i.~c.~v.}\ \bibnamefont {Brukner}},\ }\href {\doibase
  10.1103/PhysRevLett.88.210401} {\bibfield  {journal} {\bibinfo  {journal}
  {Phys. Rev. Lett.}\ }\textbf {\bibinfo {volume} {88}},\ \bibinfo {pages}
  {210401} (\bibinfo {year} {2002})}\BibitemShut {NoStop}%
\bibitem [{\citenamefont {Collins}\ \emph {et~al.}(2002)\citenamefont
  {Collins}, \citenamefont {Gisin}, \citenamefont {Linden}, \citenamefont
  {Massar},\ and\ \citenamefont {Popescu}}]{Collins_02_Bell}%
  \BibitemOpen
  \bibfield  {author} {\bibinfo {author} {\bibfnamefont {D.}~\bibnamefont
  {Collins}}, \bibinfo {author} {\bibfnamefont {N.}~\bibnamefont {Gisin}},
  \bibinfo {author} {\bibfnamefont {N.}~\bibnamefont {Linden}}, \bibinfo
  {author} {\bibfnamefont {S.}~\bibnamefont {Massar}}, \ and\ \bibinfo {author}
  {\bibfnamefont {S.}~\bibnamefont {Popescu}},\ }\href {\doibase
  10.1103/PhysRevLett.88.040404} {\bibfield  {journal} {\bibinfo  {journal}
  {Phys. Rev. Lett.}\ }\textbf {\bibinfo {volume} {88}},\ \bibinfo {pages}
  {040404} (\bibinfo {year} {2002})}\BibitemShut {NoStop}%
\bibitem [{\citenamefont {Cavalcanti}\ \emph {et~al.}(2007)\citenamefont
  {Cavalcanti}, \citenamefont {Foster}, \citenamefont {Reid},\ and\
  \citenamefont {Drummond}}]{Cavalcanti_07_Bell}%
  \BibitemOpen
  \bibfield  {author} {\bibinfo {author} {\bibfnamefont {E.~G.}\ \bibnamefont
  {Cavalcanti}}, \bibinfo {author} {\bibfnamefont {C.~J.}\ \bibnamefont
  {Foster}}, \bibinfo {author} {\bibfnamefont {M.~D.}\ \bibnamefont {Reid}}, \
  and\ \bibinfo {author} {\bibfnamefont {P.~D.}\ \bibnamefont {Drummond}},\
  }\href {\doibase 10.1103/PhysRevLett.99.210405} {\bibfield  {journal}
  {\bibinfo  {journal} {Phys. Rev. Lett.}\ }\textbf {\bibinfo {volume} {99}},\
  \bibinfo {pages} {210405} (\bibinfo {year} {2007})}\BibitemShut {NoStop}%
\bibitem [{\citenamefont {Jeong}\ \emph {et~al.}(2009)\citenamefont {Jeong},
  \citenamefont {Paternostro},\ and\ \citenamefont {Ralph}}]{jeong2009failure}%
  \BibitemOpen
  \bibfield  {author} {\bibinfo {author} {\bibfnamefont {H.}~\bibnamefont
  {Jeong}}, \bibinfo {author} {\bibfnamefont {M.}~\bibnamefont {Paternostro}},
  \ and\ \bibinfo {author} {\bibfnamefont {T.~C.}\ \bibnamefont {Ralph}},\
  }\href@noop {} {\bibfield  {journal} {\bibinfo  {journal} {Physical Review
  Letters}\ }\textbf {\bibinfo {volume} {102}},\ \bibinfo {pages} {060403}
  (\bibinfo {year} {2009})}\BibitemShut {NoStop}%
\bibitem [{\citenamefont {He}\ \emph {et~al.}(2010)\citenamefont {He},
  \citenamefont {Cavalcanti}, \citenamefont {Reid},\ and\ \citenamefont
  {Drummond}}]{He_10_Bell}%
  \BibitemOpen
  \bibfield  {author} {\bibinfo {author} {\bibfnamefont {Q.~Y.}\ \bibnamefont
  {He}}, \bibinfo {author} {\bibfnamefont {E.~G.}\ \bibnamefont {Cavalcanti}},
  \bibinfo {author} {\bibfnamefont {M.~D.}\ \bibnamefont {Reid}}, \ and\
  \bibinfo {author} {\bibfnamefont {P.~D.}\ \bibnamefont {Drummond}},\ }\href
  {\doibase 10.1103/PhysRevA.81.062106} {\bibfield  {journal} {\bibinfo
  {journal} {Phys. Rev. A}\ }\textbf {\bibinfo {volume} {81}},\ \bibinfo
  {pages} {062106} (\bibinfo {year} {2010})}\BibitemShut {NoStop}%
\bibitem [{\citenamefont {Stobi\ifmmode~\acute{n}\else \'{n}\fi{}ska}\ \emph
  {et~al.}(2011)\citenamefont {Stobi\ifmmode~\acute{n}\else \'{n}\fi{}ska},
  \citenamefont {Sekatski}, \citenamefont {Buraczewski}, \citenamefont
  {Gisin},\ and\ \citenamefont {Leuchs}}]{Stobinska_11_Bell}%
  \BibitemOpen
  \bibfield  {author} {\bibinfo {author} {\bibfnamefont {M.}~\bibnamefont
  {Stobi\ifmmode~\acute{n}\else \'{n}\fi{}ska}}, \bibinfo {author}
  {\bibfnamefont {P.}~\bibnamefont {Sekatski}}, \bibinfo {author}
  {\bibfnamefont {A.}~\bibnamefont {Buraczewski}}, \bibinfo {author}
  {\bibfnamefont {N.}~\bibnamefont {Gisin}}, \ and\ \bibinfo {author}
  {\bibfnamefont {G.}~\bibnamefont {Leuchs}},\ }\href {\doibase
  10.1103/PhysRevA.84.034104} {\bibfield  {journal} {\bibinfo  {journal} {Phys.
  Rev. A}\ }\textbf {\bibinfo {volume} {84}},\ \bibinfo {pages} {034104}
  (\bibinfo {year} {2011})}\BibitemShut {NoStop}%
\bibitem [{\citenamefont {He}\ \emph {et~al.}(2011)\citenamefont {He},
  \citenamefont {Drummond},\ and\ \citenamefont {Reid}}]{He_11_Entanglement}%
  \BibitemOpen
  \bibfield  {author} {\bibinfo {author} {\bibfnamefont {Q.~Y.}\ \bibnamefont
  {He}}, \bibinfo {author} {\bibfnamefont {P.~D.}\ \bibnamefont {Drummond}}, \
  and\ \bibinfo {author} {\bibfnamefont {M.~D.}\ \bibnamefont {Reid}},\ }\href
  {\doibase 10.1103/PhysRevA.83.032120} {\bibfield  {journal} {\bibinfo
  {journal} {Phys. Rev. A}\ }\textbf {\bibinfo {volume} {83}},\ \bibinfo
  {pages} {032120} (\bibinfo {year} {2011})}\BibitemShut {NoStop}%
\bibitem [{\citenamefont {Tura}\ \emph {et~al.}(2014)\citenamefont {Tura},
  \citenamefont {Augusiak}, \citenamefont {Sainz}, \citenamefont {V{\'e}rtesi},
  \citenamefont {Lewenstein},\ and\ \citenamefont
  {Ac{\'\i}n}}]{Tura_14_Detecting}%
  \BibitemOpen
  \bibfield  {author} {\bibinfo {author} {\bibfnamefont {J.}~\bibnamefont
  {Tura}}, \bibinfo {author} {\bibfnamefont {R.}~\bibnamefont {Augusiak}},
  \bibinfo {author} {\bibfnamefont {A.~B.}\ \bibnamefont {Sainz}}, \bibinfo
  {author} {\bibfnamefont {T.}~\bibnamefont {V{\'e}rtesi}}, \bibinfo {author}
  {\bibfnamefont {M.}~\bibnamefont {Lewenstein}}, \ and\ \bibinfo {author}
  {\bibfnamefont {A.}~\bibnamefont {Ac{\'\i}n}},\ }\href {\doibase
  10.1126/science.1247715} {\bibfield  {journal} {\bibinfo  {journal}
  {Science}\ }\textbf {\bibinfo {volume} {344}},\ \bibinfo {pages} {1256}
  (\bibinfo {year} {2014})},\ \Eprint
  {http://arxiv.org/abs/https://science.sciencemag.org/content/344/6189/1256.full.pdf}
  {https://science.sciencemag.org/content/344/6189/1256.full.pdf} \BibitemShut
  {NoStop}%
\bibitem [{\citenamefont {Tura}\ \emph {et~al.}(2015)\citenamefont {Tura},
  \citenamefont {Augusiak}, \citenamefont {Sainz}, \citenamefont {L{\"u}cke},
  \citenamefont {Klempt}, \citenamefont {Lewenstein},\ and\ \citenamefont
  {Ac\'in}}]{Tura_15_Nonlocality}%
  \BibitemOpen
  \bibfield  {author} {\bibinfo {author} {\bibfnamefont {J.}~\bibnamefont
  {Tura}}, \bibinfo {author} {\bibfnamefont {R.}~\bibnamefont {Augusiak}},
  \bibinfo {author} {\bibfnamefont {A.}~\bibnamefont {Sainz}}, \bibinfo
  {author} {\bibfnamefont {B.}~\bibnamefont {L{\"u}cke}}, \bibinfo {author}
  {\bibfnamefont {C.}~\bibnamefont {Klempt}}, \bibinfo {author} {\bibfnamefont
  {M.}~\bibnamefont {Lewenstein}}, \ and\ \bibinfo {author} {\bibfnamefont
  {A.}~\bibnamefont {Ac\'in}},\ }\href {\doibase
  https://doi.org/10.1016/j.aop.2015.07.021} {\bibfield  {journal} {\bibinfo
  {journal} {Annals of Physics}\ }\textbf {\bibinfo {volume} {362}},\ \bibinfo
  {pages} {370 } (\bibinfo {year} {2015})}\BibitemShut {NoStop}%
\bibitem [{\citenamefont {Engelsen}\ \emph {et~al.}(2017)\citenamefont
  {Engelsen}, \citenamefont {Krishnakumar}, \citenamefont {Hosten},\ and\
  \citenamefont {Kasevich}}]{Engelsen_17_Bell}%
  \BibitemOpen
  \bibfield  {author} {\bibinfo {author} {\bibfnamefont {N.~J.}\ \bibnamefont
  {Engelsen}}, \bibinfo {author} {\bibfnamefont {R.}~\bibnamefont
  {Krishnakumar}}, \bibinfo {author} {\bibfnamefont {O.}~\bibnamefont
  {Hosten}}, \ and\ \bibinfo {author} {\bibfnamefont {M.~A.}\ \bibnamefont
  {Kasevich}},\ }\href {\doibase 10.1103/PhysRevLett.118.140401} {\bibfield
  {journal} {\bibinfo  {journal} {Phys. Rev. Lett.}\ }\textbf {\bibinfo
  {volume} {118}},\ \bibinfo {pages} {140401} (\bibinfo {year}
  {2017})}\BibitemShut {NoStop}%
\bibitem [{\citenamefont {Dalton}(2018)}]{Dalton_19_CGLMP}%
  \BibitemOpen
  \bibfield  {author} {\bibinfo {author} {\bibfnamefont {B.~J.}\ \bibnamefont
  {Dalton}},\ }\href@noop {} {\bibfield  {journal} {\bibinfo  {journal}
  {arXiv}\ } (\bibinfo {year} {2018})},\ \Eprint
  {http://arxiv.org/abs/1812.09651} {arXiv:1812.09651 [quant-ph]} \BibitemShut
  {NoStop}%
\bibitem [{\citenamefont {Thenabadu}\ \emph {et~al.}(2019)\citenamefont
  {Thenabadu}, \citenamefont {Cheng}, \citenamefont {Pham}, \citenamefont
  {Drummond}, \citenamefont {Rosales-Z{\'a}rate},\ and\ \citenamefont
  {Reid}}]{thenabadu2019testing}%
  \BibitemOpen
  \bibfield  {author} {\bibinfo {author} {\bibfnamefont {M.}~\bibnamefont
  {Thenabadu}}, \bibinfo {author} {\bibfnamefont {G.}~\bibnamefont {Cheng}},
  \bibinfo {author} {\bibfnamefont {T.}~\bibnamefont {Pham}}, \bibinfo {author}
  {\bibfnamefont {L.}~\bibnamefont {Drummond}}, \bibinfo {author}
  {\bibfnamefont {L.}~\bibnamefont {Rosales-Z{\'a}rate}}, \ and\ \bibinfo
  {author} {\bibfnamefont {M.}~\bibnamefont {Reid}},\ }\href@noop {} {\bibfield
   {journal} {\bibinfo  {journal} {arXiv preprint arXiv:1906.04900}\ }
  (\bibinfo {year} {2019})}\BibitemShut {NoStop}%
\bibitem [{\citenamefont {Navascu{\'e}s}\ \emph {et~al.}(2013)\citenamefont
  {Navascu{\'e}s}, \citenamefont {P{\'e}rez-Garc{\'\i}a},\ and\ \citenamefont
  {Villanueva}}]{navascues2013testing}%
  \BibitemOpen
  \bibfield  {author} {\bibinfo {author} {\bibfnamefont {M.}~\bibnamefont
  {Navascu{\'e}s}}, \bibinfo {author} {\bibfnamefont {D.}~\bibnamefont
  {P{\'e}rez-Garc{\'\i}a}}, \ and\ \bibinfo {author} {\bibfnamefont
  {I.}~\bibnamefont {Villanueva}},\ }\href@noop {} {\bibfield  {journal}
  {\bibinfo  {journal} {Journal of Physics A: Mathematical and Theoretical}\
  }\textbf {\bibinfo {volume} {46}},\ \bibinfo {pages} {085304} (\bibinfo
  {year} {2013})}\BibitemShut {NoStop}%
\bibitem [{\citenamefont {Dalton}(2019)}]{Dalton_19_Bell}%
  \BibitemOpen
  \bibfield  {author} {\bibinfo {author} {\bibfnamefont {B.~J.}\ \bibnamefont
  {Dalton}},\ }\href@noop {} {\bibfield  {journal} {\bibinfo  {journal} {Euro.
  Phys. J. Special Topics}\ }\textbf {\bibinfo {volume} {227}},\ \bibinfo
  {pages} {2069–2083} (\bibinfo {year} {2019})}\BibitemShut {NoStop}%
\bibitem [{\citenamefont {Reid}\ \emph {et~al.}(2012)\citenamefont {Reid},
  \citenamefont {He},\ and\ \citenamefont {Drummond}}]{Reid_12_Entanglement}%
  \BibitemOpen
  \bibfield  {author} {\bibinfo {author} {\bibfnamefont {M.~D.}\ \bibnamefont
  {Reid}}, \bibinfo {author} {\bibfnamefont {Q.-Y.}\ \bibnamefont {He}}, \ and\
  \bibinfo {author} {\bibfnamefont {P.~D.}\ \bibnamefont {Drummond}},\ }\href
  {\doibase 10.1007/s11467-011-0233-9} {\bibfield  {journal} {\bibinfo
  {journal} {Frontiers of Physics}\ }\textbf {\bibinfo {volume} {7}},\ \bibinfo
  {pages} {72} (\bibinfo {year} {2012})}\BibitemShut {NoStop}%
\bibitem [{\citenamefont {Bell}(1964)}]{Bell_64_On}%
  \BibitemOpen
  \bibfield  {author} {\bibinfo {author} {\bibfnamefont {J.~S.}\ \bibnamefont
  {Bell}},\ }\href {\doibase 10.1103/PhysicsPhysiqueFizika.1.195} {\bibfield
  {journal} {\bibinfo  {journal} {Physics Physique Fizika}\ }\textbf {\bibinfo
  {volume} {1}},\ \bibinfo {pages} {195} (\bibinfo {year} {1964})}\BibitemShut
  {NoStop}%
\bibitem [{\citenamefont {Navascués}\ and\ \citenamefont
  {Wunderlich}(2010)}]{Navascues_10_Glance}%
  \BibitemOpen
  \bibfield  {author} {\bibinfo {author} {\bibfnamefont {M.}~\bibnamefont
  {Navascués}}\ and\ \bibinfo {author} {\bibfnamefont {H.}~\bibnamefont
  {Wunderlich}},\ }\href {\doibase 10.1098/rspa.2009.0453} {\bibfield
  {journal} {\bibinfo  {journal} {Proceedings of the Royal Society A:
  Mathematical, Physical and Engineering Sciences}\ }\textbf {\bibinfo {volume}
  {466}},\ \bibinfo {pages} {881} (\bibinfo {year} {2010})}\BibitemShut
  {NoStop}%
\bibitem [{\citenamefont {Barz}\ \emph {et~al.}(2010)\citenamefont {Barz},
  \citenamefont {Cronenberg}, \citenamefont {Zeilinger},\ and\ \citenamefont
  {Walther}}]{Barz_10_Heralded}%
  \BibitemOpen
  \bibfield  {author} {\bibinfo {author} {\bibfnamefont {S.}~\bibnamefont
  {Barz}}, \bibinfo {author} {\bibfnamefont {G.}~\bibnamefont {Cronenberg}},
  \bibinfo {author} {\bibfnamefont {A.}~\bibnamefont {Zeilinger}}, \ and\
  \bibinfo {author} {\bibfnamefont {P.}~\bibnamefont {Walther}},\ }\href
  {https://doi.org/10.1038/nphoton.2010.156} {\bibfield  {journal} {\bibinfo
  {journal} {Nature Photonics}\ }\textbf {\bibinfo {volume} {4}},\ \bibinfo
  {pages} {553} (\bibinfo {year} {2010})}\BibitemShut {NoStop}%
\bibitem [{\citenamefont {Bernien}\ \emph {et~al.}(2013)\citenamefont
  {Bernien}, \citenamefont {Hensen}, \citenamefont {Pfaff}, \citenamefont
  {Koolstra}, \citenamefont {Blok}, \citenamefont {Robledo}, \citenamefont
  {Taminiau}, \citenamefont {Markham}, \citenamefont {Twitchen}, \citenamefont
  {Childress},\ and\ \citenamefont {Hanson}}]{Bernien_13_Heralded}%
  \BibitemOpen
  \bibfield  {author} {\bibinfo {author} {\bibfnamefont {H.}~\bibnamefont
  {Bernien}}, \bibinfo {author} {\bibfnamefont {B.}~\bibnamefont {Hensen}},
  \bibinfo {author} {\bibfnamefont {W.}~\bibnamefont {Pfaff}}, \bibinfo
  {author} {\bibfnamefont {G.}~\bibnamefont {Koolstra}}, \bibinfo {author}
  {\bibfnamefont {M.~S.}\ \bibnamefont {Blok}}, \bibinfo {author}
  {\bibfnamefont {L.}~\bibnamefont {Robledo}}, \bibinfo {author} {\bibfnamefont
  {T.~H.}\ \bibnamefont {Taminiau}}, \bibinfo {author} {\bibfnamefont
  {M.}~\bibnamefont {Markham}}, \bibinfo {author} {\bibfnamefont {D.~J.}\
  \bibnamefont {Twitchen}}, \bibinfo {author} {\bibfnamefont {L.}~\bibnamefont
  {Childress}}, \ and\ \bibinfo {author} {\bibfnamefont {R.}~\bibnamefont
  {Hanson}},\ }\href {https://doi.org/10.1038/nature12016} {\bibfield
  {journal} {\bibinfo  {journal} {Nature}\ }\textbf {\bibinfo {volume} {497}},\
  \bibinfo {pages} {86} (\bibinfo {year} {2013})}\BibitemShut {NoStop}%
\bibitem [{\citenamefont {Laurat}\ \emph {et~al.}(2007)\citenamefont {Laurat},
  \citenamefont {Choi}, \citenamefont {Deng}, \citenamefont {Chou},\ and\
  \citenamefont {Kimble}}]{Laurat_07_Heralded}%
  \BibitemOpen
  \bibfield  {author} {\bibinfo {author} {\bibfnamefont {J.}~\bibnamefont
  {Laurat}}, \bibinfo {author} {\bibfnamefont {K.~S.}\ \bibnamefont {Choi}},
  \bibinfo {author} {\bibfnamefont {H.}~\bibnamefont {Deng}}, \bibinfo {author}
  {\bibfnamefont {C.~W.}\ \bibnamefont {Chou}}, \ and\ \bibinfo {author}
  {\bibfnamefont {H.~J.}\ \bibnamefont {Kimble}},\ }\href {\doibase
  10.1103/PhysRevLett.99.180504} {\bibfield  {journal} {\bibinfo  {journal}
  {Phys. Rev. Lett.}\ }\textbf {\bibinfo {volume} {99}},\ \bibinfo {pages}
  {180504} (\bibinfo {year} {2007})}\BibitemShut {NoStop}%
\bibitem [{\citenamefont {Hofmann}\ \emph {et~al.}(2012)\citenamefont
  {Hofmann}, \citenamefont {Krug}, \citenamefont {Ortegel}, \citenamefont
  {G{\'e}rard}, \citenamefont {Weber}, \citenamefont {Rosenfeld},\ and\
  \citenamefont {Weinfurter}}]{Hofmann_12_Heralded}%
  \BibitemOpen
  \bibfield  {author} {\bibinfo {author} {\bibfnamefont {J.}~\bibnamefont
  {Hofmann}}, \bibinfo {author} {\bibfnamefont {M.}~\bibnamefont {Krug}},
  \bibinfo {author} {\bibfnamefont {N.}~\bibnamefont {Ortegel}}, \bibinfo
  {author} {\bibfnamefont {L.}~\bibnamefont {G{\'e}rard}}, \bibinfo {author}
  {\bibfnamefont {M.}~\bibnamefont {Weber}}, \bibinfo {author} {\bibfnamefont
  {W.}~\bibnamefont {Rosenfeld}}, \ and\ \bibinfo {author} {\bibfnamefont
  {H.}~\bibnamefont {Weinfurter}},\ }\href {\doibase 10.1126/science.1221856}
  {\bibfield  {journal} {\bibinfo  {journal} {Science}\ }\textbf {\bibinfo
  {volume} {337}},\ \bibinfo {pages} {72} (\bibinfo {year} {2012})},\ \Eprint
  {http://arxiv.org/abs/https://science.sciencemag.org/content/337/6090/72.full.pdf}
  {https://science.sciencemag.org/content/337/6090/72.full.pdf} \BibitemShut
  {NoStop}%
\bibitem [{\citenamefont {Casabone}\ \emph {et~al.}(2013)\citenamefont
  {Casabone}, \citenamefont {Stute}, \citenamefont {Friebe}, \citenamefont
  {Brandst\"atter}, \citenamefont {Sch\"uppert}, \citenamefont {Blatt},\ and\
  \citenamefont {Northup}}]{Casabone_13_Heralded}%
  \BibitemOpen
  \bibfield  {author} {\bibinfo {author} {\bibfnamefont {B.}~\bibnamefont
  {Casabone}}, \bibinfo {author} {\bibfnamefont {A.}~\bibnamefont {Stute}},
  \bibinfo {author} {\bibfnamefont {K.}~\bibnamefont {Friebe}}, \bibinfo
  {author} {\bibfnamefont {B.}~\bibnamefont {Brandst\"atter}}, \bibinfo
  {author} {\bibfnamefont {K.}~\bibnamefont {Sch\"uppert}}, \bibinfo {author}
  {\bibfnamefont {R.}~\bibnamefont {Blatt}}, \ and\ \bibinfo {author}
  {\bibfnamefont {T.~E.}\ \bibnamefont {Northup}},\ }\href {\doibase
  10.1103/PhysRevLett.111.100505} {\bibfield  {journal} {\bibinfo  {journal}
  {Phys. Rev. Lett.}\ }\textbf {\bibinfo {volume} {111}},\ \bibinfo {pages}
  {100505} (\bibinfo {year} {2013})}\BibitemShut {NoStop}%
\bibitem [{\citenamefont {Yang}\ \emph {et~al.}(2011)\citenamefont {Yang},
  \citenamefont {Navascu\'es}, \citenamefont {Sheridan},\ and\ \citenamefont
  {Scarani}}]{Yang_11_Quantum}%
  \BibitemOpen
  \bibfield  {author} {\bibinfo {author} {\bibfnamefont {T.~H.}\ \bibnamefont
  {Yang}}, \bibinfo {author} {\bibfnamefont {M.}~\bibnamefont {Navascu\'es}},
  \bibinfo {author} {\bibfnamefont {L.}~\bibnamefont {Sheridan}}, \ and\
  \bibinfo {author} {\bibfnamefont {V.}~\bibnamefont {Scarani}},\ }\href
  {\doibase 10.1103/PhysRevA.83.022105} {\bibfield  {journal} {\bibinfo
  {journal} {Phys. Rev. A}\ }\textbf {\bibinfo {volume} {83}},\ \bibinfo
  {pages} {022105} (\bibinfo {year} {2011})}\BibitemShut {NoStop}%
\bibitem [{\citenamefont {Navascu\'es}(2016)}]{Navascues_16_Macroscopic}%
  \BibitemOpen
  \bibfield  {author} {\bibinfo {author} {\bibfnamefont {M.}~\bibnamefont
  {Navascu\'es}},\ }\enquote {\bibinfo {title} {Macroscopic locality},}\ in\
  \href {\doibase 10.1007/978-94-017-7303-4} {\emph {\bibinfo {booktitle}
  {Quantum Theory: Informational Foundations and Foils}}}\ (\bibinfo
  {publisher} {Springer},\ \bibinfo {address} {Netherlands},\ \bibinfo {year}
  {2016})\ pp.\ \bibinfo {pages} {439--463}\BibitemShut {NoStop}%
\bibitem [{\citenamefont {{Fisher}}(2015)}]{Fisher_15_Quantum}%
  \BibitemOpen
  \bibfield  {author} {\bibinfo {author} {\bibfnamefont {M.~P.~A.}\
  \bibnamefont {{Fisher}}},\ }\href {\doibase 10.1016/j.aop.2015.08.020}
  {\bibfield  {journal} {\bibinfo  {journal} {Annals of Physics}\ }\textbf
  {\bibinfo {volume} {362}},\ \bibinfo {pages} {593} (\bibinfo {year}
  {2015})}\BibitemShut {NoStop}%
\bibitem [{\citenamefont {Fisher}\ and\ \citenamefont
  {Radzihovsky}(2018)}]{Fisher_18_Quantum}%
  \BibitemOpen
  \bibfield  {author} {\bibinfo {author} {\bibfnamefont {M.~P.~A.}\
  \bibnamefont {Fisher}}\ and\ \bibinfo {author} {\bibfnamefont
  {L.}~\bibnamefont {Radzihovsky}},\ }\href {\doibase 10.1073/pnas.1718402115}
  {\bibfield  {journal} {\bibinfo  {journal} {Proceedings of the National
  Academy of Sciences}\ }\textbf {\bibinfo {volume} {115}},\ \bibinfo {pages}
  {E4551} (\bibinfo {year} {2018})},\ \Eprint
  {http://arxiv.org/abs/https://www.pnas.org/content/115/20/E4551.full.pdf}
  {https://www.pnas.org/content/115/20/E4551.full.pdf} \BibitemShut {NoStop}%
\bibitem [{\citenamefont {{Yunger Halpern}}\ and\ \citenamefont
  {Crosson}(2019)}]{NYH_19_Quantum}%
  \BibitemOpen
  \bibfield  {author} {\bibinfo {author} {\bibfnamefont {N.}~\bibnamefont
  {{Yunger Halpern}}}\ and\ \bibinfo {author} {\bibfnamefont {E.}~\bibnamefont
  {Crosson}},\ }\href {\doibase https://doi.org/10.1016/j.aop.2018.11.016}
  {\bibfield  {journal} {\bibinfo  {journal} {Annals of Physics}\ }\textbf
  {\bibinfo {volume} {407}},\ \bibinfo {pages} {92 } (\bibinfo {year}
  {2019})}\BibitemShut {NoStop}%
\bibitem [{\citenamefont {Fisher}(2017)}]{Fisher_17_Are}%
  \BibitemOpen
  \bibfield  {author} {\bibinfo {author} {\bibfnamefont {M.~P.~A.}\
  \bibnamefont {Fisher}},\ }\href {\doibase 10.1142/S0217979217430019}
  {\bibfield  {journal} {\bibinfo  {journal} {International Journal of Modern
  Physics B}\ }\textbf {\bibinfo {volume} {31}},\ \bibinfo {pages} {1743001}
  (\bibinfo {year} {2017})}\BibitemShut {NoStop}%
\bibitem [{\citenamefont {Fifer}\ \emph {et~al.}(2019)\citenamefont {Fifer},
  \citenamefont {Torres}, \citenamefont {Erne}, \citenamefont {Avgoustidis},
  \citenamefont {Hill},\ and\ \citenamefont {Weinfurtner}}]{Fifer_19_Analog}%
  \BibitemOpen
  \bibfield  {author} {\bibinfo {author} {\bibfnamefont {Z.}~\bibnamefont
  {Fifer}}, \bibinfo {author} {\bibfnamefont {T.}~\bibnamefont {Torres}},
  \bibinfo {author} {\bibfnamefont {S.}~\bibnamefont {Erne}}, \bibinfo {author}
  {\bibfnamefont {A.}~\bibnamefont {Avgoustidis}}, \bibinfo {author}
  {\bibfnamefont {R.~J.~A.}\ \bibnamefont {Hill}}, \ and\ \bibinfo {author}
  {\bibfnamefont {S.}~\bibnamefont {Weinfurtner}},\ }\href {\doibase
  10.1103/PhysRevE.99.031101} {\bibfield  {journal} {\bibinfo  {journal} {Phys.
  Rev. E}\ }\textbf {\bibinfo {volume} {99}},\ \bibinfo {pages} {031101(R)}
  (\bibinfo {year} {2019})}\BibitemShut {NoStop}%
\bibitem [{\citenamefont {Pozsgay}\ \emph {et~al.}(2017)\citenamefont
  {Pozsgay}, \citenamefont {Hirsch}, \citenamefont {Branciard},\ and\
  \citenamefont {Brunner}}]{Pozsgay_17_Covariance}%
  \BibitemOpen
  \bibfield  {author} {\bibinfo {author} {\bibfnamefont {V.}~\bibnamefont
  {Pozsgay}}, \bibinfo {author} {\bibfnamefont {F.}~\bibnamefont {Hirsch}},
  \bibinfo {author} {\bibfnamefont {C.}~\bibnamefont {Branciard}}, \ and\
  \bibinfo {author} {\bibfnamefont {N.}~\bibnamefont {Brunner}},\ }\href
  {\doibase 10.1103/PhysRevA.96.062128} {\bibfield  {journal} {\bibinfo
  {journal} {Phys. Rev. A}\ }\textbf {\bibinfo {volume} {96}},\ \bibinfo
  {pages} {062128} (\bibinfo {year} {2017})}\BibitemShut {NoStop}%
\bibitem [{\citenamefont {Kwiat}\ \emph {et~al.}(1995)\citenamefont {Kwiat},
  \citenamefont {Mattle}, \citenamefont {Weinfurter}, \citenamefont
  {Zeilinger}, \citenamefont {Sergienko},\ and\ \citenamefont
  {Shih}}]{kwiat1995new}%
  \BibitemOpen
  \bibfield  {author} {\bibinfo {author} {\bibfnamefont {P.~G.}\ \bibnamefont
  {Kwiat}}, \bibinfo {author} {\bibfnamefont {K.}~\bibnamefont {Mattle}},
  \bibinfo {author} {\bibfnamefont {H.}~\bibnamefont {Weinfurter}}, \bibinfo
  {author} {\bibfnamefont {A.}~\bibnamefont {Zeilinger}}, \bibinfo {author}
  {\bibfnamefont {A.~V.}\ \bibnamefont {Sergienko}}, \ and\ \bibinfo {author}
  {\bibfnamefont {Y.}~\bibnamefont {Shih}},\ }\href@noop {} {\bibfield
  {journal} {\bibinfo  {journal} {Physical Review Letters}\ }\textbf {\bibinfo
  {volume} {75}},\ \bibinfo {pages} {4337} (\bibinfo {year}
  {1995})}\BibitemShut {NoStop}%
\bibitem [{\citenamefont {Clauser}\ \emph {et~al.}(1969)\citenamefont
  {Clauser}, \citenamefont {Horne}, \citenamefont {Shimony},\ and\
  \citenamefont {Holt}}]{Clauser_69_Proposed}%
  \BibitemOpen
  \bibfield  {author} {\bibinfo {author} {\bibfnamefont {J.~F.}\ \bibnamefont
  {Clauser}}, \bibinfo {author} {\bibfnamefont {M.~A.}\ \bibnamefont {Horne}},
  \bibinfo {author} {\bibfnamefont {A.}~\bibnamefont {Shimony}}, \ and\
  \bibinfo {author} {\bibfnamefont {R.~A.}\ \bibnamefont {Holt}},\ }\href
  {\doibase 10.1103/PhysRevLett.23.880} {\bibfield  {journal} {\bibinfo
  {journal} {Phys. Rev. Lett.}\ }\textbf {\bibinfo {volume} {23}},\ \bibinfo
  {pages} {880} (\bibinfo {year} {1969})}\BibitemShut {NoStop}%
\bibitem [{\citenamefont {Buhrman}\ \emph {et~al.}(2001)\citenamefont
  {Buhrman}, \citenamefont {Cleve}, \citenamefont {Watrous},\ and\
  \citenamefont {de~Wolf}}]{Buhrman_01_Quantum}%
  \BibitemOpen
  \bibfield  {author} {\bibinfo {author} {\bibfnamefont {H.}~\bibnamefont
  {Buhrman}}, \bibinfo {author} {\bibfnamefont {R.}~\bibnamefont {Cleve}},
  \bibinfo {author} {\bibfnamefont {J.}~\bibnamefont {Watrous}}, \ and\
  \bibinfo {author} {\bibfnamefont {R.}~\bibnamefont {de~Wolf}},\ }\href@noop
  {} {\bibfield  {journal} {\bibinfo  {journal} {Phys. Rev. Lett.}\ }\textbf
  {\bibinfo {volume} {87}},\ \bibinfo {pages} {167902} (\bibinfo {year}
  {2001})},\ \Eprint {http://arxiv.org/abs/quant-ph/0102001} {quant-ph/0102001}
  \BibitemShut {NoStop}%
\bibitem [{\citenamefont {Reid}(2001{\natexlab{a}})}]{reid2001proposal}%
  \BibitemOpen
  \bibfield  {author} {\bibinfo {author} {\bibfnamefont {M.}~\bibnamefont
  {Reid}},\ }\href@noop {} {\bibfield  {journal} {\bibinfo  {journal} {arXiv
  preprint quant-ph/0101052}\ } (\bibinfo {year}
  {2001}{\natexlab{a}})}\BibitemShut {NoStop}%
\bibitem [{\citenamefont {Reid}(2001{\natexlab{b}})}]{reid2001new}%
  \BibitemOpen
  \bibfield  {author} {\bibinfo {author} {\bibfnamefont {M.}~\bibnamefont
  {Reid}},\ }in\ \href@noop {} {\emph {\bibinfo {booktitle} {Directions in
  Quantum Optics}}}\ (\bibinfo  {publisher} {Springer},\ \bibinfo {year}
  {2001})\ pp.\ \bibinfo {pages} {176--186}\BibitemShut {NoStop}%
\bibitem [{\citenamefont {Cleve}\ \emph {et~al.}(2004)\citenamefont {Cleve},
  \citenamefont {H{\o }yer}, \citenamefont {Toner},\ and\ \citenamefont
  {Watrous}}]{CHTW04}%
  \BibitemOpen
  \bibfield  {author} {\bibinfo {author} {\bibfnamefont {R.}~\bibnamefont
  {Cleve}}, \bibinfo {author} {\bibfnamefont {P.}~\bibnamefont {H{\o }yer}},
  \bibinfo {author} {\bibfnamefont {B.}~\bibnamefont {Toner}}, \ and\ \bibinfo
  {author} {\bibfnamefont {J.}~\bibnamefont {Watrous}},\ }in\ \href@noop {}
  {\emph {\bibinfo {booktitle} {Proc. 19\textsuperscript{th} Annual IEEE Conf.
  Computational Complexity}}}\ (\bibinfo {year} {2004})\ pp.\ \bibinfo {pages}
  {236--249},\ \Eprint {http://arxiv.org/abs/quant-ph/0404076}
  {quant-ph/0404076} \BibitemShut {NoStop}%
\bibitem [{\citenamefont {Preskill}(2001)}]{Preskill_01_Lecture}%
  \BibitemOpen
  \bibfield  {author} {\bibinfo {author} {\bibfnamefont {J.}~\bibnamefont
  {Preskill}},\ }\href
  {http://www.theory.caltech.edu/~preskill/ph229/notes/chap4_01.pdf} {\enquote
  {\bibinfo {title} {Lecture notes for ph219/cs219: Quantum information and
  computation, chapter 4: Entanglement},}\ } (\bibinfo {year}
  {2001})\BibitemShut {NoStop}%
\bibitem [{\citenamefont {Onuma}\ and\ \citenamefont
  {Ito}(1998)}]{Onuma_98_Posner}%
  \BibitemOpen
  \bibfield  {author} {\bibinfo {author} {\bibfnamefont {K.}~\bibnamefont
  {Onuma}}\ and\ \bibinfo {author} {\bibfnamefont {A.}~\bibnamefont {Ito}},\
  }\href {\doibase 10.1021/cm980062c} {\bibfield  {journal} {\bibinfo
  {journal} {Chemistry of Materials}\ }\textbf {\bibinfo {volume} {10}},\
  \bibinfo {pages} {3346} (\bibinfo {year} {1998})},\ \Eprint
  {http://arxiv.org/abs/http://dx.doi.org/10.1021/cm980062c}
  {http://dx.doi.org/10.1021/cm980062c} \BibitemShut {NoStop}%
\bibitem [{\citenamefont {Oyane}\ \emph {et~al.}(1999)\citenamefont {Oyane},
  \citenamefont {Onuma}, \citenamefont {Kokubo},\ and\ \citenamefont
  {Atsuo}}]{Ayako_99_Posner}%
  \BibitemOpen
  \bibfield  {author} {\bibinfo {author} {\bibfnamefont {A.}~\bibnamefont
  {Oyane}}, \bibinfo {author} {\bibfnamefont {K.}~\bibnamefont {Onuma}},
  \bibinfo {author} {\bibfnamefont {T.}~\bibnamefont {Kokubo}}, \ and\ \bibinfo
  {author} {\bibfnamefont {I.}~\bibnamefont {Atsuo}},\ }\href {\doibase
  10.1021/jp9910340} {\bibfield  {journal} {\bibinfo  {journal} {The Journal of
  Physical Chemistry B}\ }\textbf {\bibinfo {volume} {103}},\ \bibinfo {pages}
  {8230} (\bibinfo {year} {1999})},\ \Eprint
  {http://arxiv.org/abs/http://dx.doi.org/10.1021/jp9910340}
  {http://dx.doi.org/10.1021/jp9910340} \BibitemShut {NoStop}%
\bibitem [{\citenamefont {Dey}\ \emph {et~al.}(2010)\citenamefont {Dey},
  \citenamefont {Bomans}, \citenamefont {Muller}, \citenamefont {Will},
  \citenamefont {Frederik}, \citenamefont {{de With}},\ and\ \citenamefont
  {Sommerdijk}}]{Dey_10_Posner}%
  \BibitemOpen
  \bibfield  {author} {\bibinfo {author} {\bibfnamefont {A.}~\bibnamefont
  {Dey}}, \bibinfo {author} {\bibfnamefont {P.}~\bibnamefont {Bomans}},
  \bibinfo {author} {\bibfnamefont {F.}~\bibnamefont {Muller}}, \bibinfo
  {author} {\bibfnamefont {J.}~\bibnamefont {Will}}, \bibinfo {author}
  {\bibfnamefont {P.}~\bibnamefont {Frederik}}, \bibinfo {author}
  {\bibfnamefont {G.}~\bibnamefont {{de With}}}, \ and\ \bibinfo {author}
  {\bibfnamefont {N.~A.}\ \bibnamefont {Sommerdijk}},\ }\href@noop {}
  {\bibfield  {journal} {\bibinfo  {journal} {Nature Materials}\ }\textbf
  {\bibinfo {volume} {9}},\ \bibinfo {pages} {1010} (\bibinfo {year}
  {2010})}\BibitemShut {NoStop}%
\bibitem [{\citenamefont {Treboux}\ \emph {et~al.}(2000)\citenamefont
  {Treboux}, \citenamefont {Layrolle}, \citenamefont {Kanzaki}, \citenamefont
  {Onuma},\ and\ \citenamefont {Ito}}]{Treboux_00_Posner}%
  \BibitemOpen
  \bibfield  {author} {\bibinfo {author} {\bibfnamefont {G.}~\bibnamefont
  {Treboux}}, \bibinfo {author} {\bibfnamefont {P.}~\bibnamefont {Layrolle}},
  \bibinfo {author} {\bibfnamefont {N.}~\bibnamefont {Kanzaki}}, \bibinfo
  {author} {\bibfnamefont {K.}~\bibnamefont {Onuma}}, \ and\ \bibinfo {author}
  {\bibfnamefont {A.}~\bibnamefont {Ito}},\ }\href {\doibase 10.1021/jp994399t}
  {\bibfield  {journal} {\bibinfo  {journal} {The Journal of Physical Chemistry
  A}\ }\textbf {\bibinfo {volume} {104}},\ \bibinfo {pages} {5111} (\bibinfo
  {year} {2000})},\ \Eprint
  {http://arxiv.org/abs/http://dx.doi.org/10.1021/jp994399t}
  {http://dx.doi.org/10.1021/jp994399t} \BibitemShut {NoStop}%
\bibitem [{\citenamefont {Kanzaki}\ \emph {et~al.}(2001)\citenamefont
  {Kanzaki}, \citenamefont {Treboux}, \citenamefont {Onuma}, \citenamefont
  {Tsutsumi},\ and\ \citenamefont {Ito}}]{Kanzaki_01_Posner}%
  \BibitemOpen
  \bibfield  {author} {\bibinfo {author} {\bibfnamefont {N.}~\bibnamefont
  {Kanzaki}}, \bibinfo {author} {\bibfnamefont {G.}~\bibnamefont {Treboux}},
  \bibinfo {author} {\bibfnamefont {K.}~\bibnamefont {Onuma}}, \bibinfo
  {author} {\bibfnamefont {S.}~\bibnamefont {Tsutsumi}}, \ and\ \bibinfo
  {author} {\bibfnamefont {A.}~\bibnamefont {Ito}},\ }\href {\doibase
  http://dx.doi.org/10.1016/S0142-9612(01)00039-4} {\bibfield  {journal}
  {\bibinfo  {journal} {Biomaterials}\ }\textbf {\bibinfo {volume} {22}},\
  \bibinfo {pages} {2921 } (\bibinfo {year} {2001})}\BibitemShut {NoStop}%
\bibitem [{\citenamefont {Yin}\ and\ \citenamefont
  {Stott}(2003)}]{Yin_03_Posner}%
  \BibitemOpen
  \bibfield  {author} {\bibinfo {author} {\bibfnamefont {X.}~\bibnamefont
  {Yin}}\ and\ \bibinfo {author} {\bibfnamefont {M.}~\bibnamefont {Stott}},\
  }\href@noop {} {\bibfield  {journal} {\bibinfo  {journal} {J. Chem. Phys.}\
  }\textbf {\bibinfo {volume} {118}},\ \bibinfo {pages} {3717} (\bibinfo {year}
  {2003})}\BibitemShut {NoStop}%
\bibitem [{\citenamefont {Swift}\ \emph {et~al.}(2018)\citenamefont {Swift},
  \citenamefont {Van~de Walle},\ and\ \citenamefont
  {Fisher}}]{Swift_17_Posner}%
  \BibitemOpen
  \bibfield  {author} {\bibinfo {author} {\bibfnamefont {M.~W.}\ \bibnamefont
  {Swift}}, \bibinfo {author} {\bibfnamefont {C.~G.}\ \bibnamefont {Van~de
  Walle}}, \ and\ \bibinfo {author} {\bibfnamefont {M.~P.~A.}\ \bibnamefont
  {Fisher}},\ }\href {\doibase 10.1039/C7CP07720C} {\bibfield  {journal}
  {\bibinfo  {journal} {Phys. Chem. Chem. Phys.}\ }\textbf {\bibinfo {volume}
  {20}},\ \bibinfo {pages} {12373} (\bibinfo {year} {2018})}\BibitemShut
  {NoStop}%
\bibitem [{\citenamefont {Freedman}\ \emph {et~al.}(2018)\citenamefont
  {Freedman}, \citenamefont {Shokrian-Zini},\ and\ \citenamefont
  {Wang}}]{Freedman_18_Quantum}%
  \BibitemOpen
  \bibfield  {author} {\bibinfo {author} {\bibfnamefont {M.}~\bibnamefont
  {Freedman}}, \bibinfo {author} {\bibfnamefont {M.}~\bibnamefont
  {Shokrian-Zini}}, \ and\ \bibinfo {author} {\bibfnamefont {Z.}~\bibnamefont
  {Wang}},\ }\href@noop {} {\bibfield  {journal} {\bibinfo  {journal} {arXiv
  e-prints}\ } (\bibinfo {year} {2018})},\ \Eprint
  {http://arxiv.org/abs/1811.08580} {arXiv:1811.08580 [quant-ph]} \BibitemShut
  {NoStop}%
\bibitem [{\citenamefont {Li}\ \emph {et~al.}(2019)\citenamefont {Li},
  \citenamefont {Chen},\ and\ \citenamefont {Fisher}}]{Li_19_Measurement}%
  \BibitemOpen
  \bibfield  {author} {\bibinfo {author} {\bibfnamefont {Y.}~\bibnamefont
  {Li}}, \bibinfo {author} {\bibfnamefont {X.}~\bibnamefont {Chen}}, \ and\
  \bibinfo {author} {\bibfnamefont {M.~P.~A.}\ \bibnamefont {Fisher}},\ }\href
  {\doibase 10.1103/PhysRevB.100.134306} {\bibfield  {journal} {\bibinfo
  {journal} {Phys. Rev. B}\ }\textbf {\bibinfo {volume} {100}},\ \bibinfo
  {pages} {134306} (\bibinfo {year} {2019})}\BibitemShut {NoStop}%
\bibitem [{\citenamefont {Chan}\ \emph {et~al.}(2019)\citenamefont {Chan},
  \citenamefont {Nandkishore}, \citenamefont {Pretko},\ and\ \citenamefont
  {Smith}}]{Chan_19_Unitary}%
  \BibitemOpen
  \bibfield  {author} {\bibinfo {author} {\bibfnamefont {A.}~\bibnamefont
  {Chan}}, \bibinfo {author} {\bibfnamefont {R.~M.}\ \bibnamefont
  {Nandkishore}}, \bibinfo {author} {\bibfnamefont {M.}~\bibnamefont {Pretko}},
  \ and\ \bibinfo {author} {\bibfnamefont {G.}~\bibnamefont {Smith}},\ }\href
  {\doibase 10.1103/PhysRevB.99.224307} {\bibfield  {journal} {\bibinfo
  {journal} {Phys. Rev. B}\ }\textbf {\bibinfo {volume} {99}},\ \bibinfo
  {pages} {224307} (\bibinfo {year} {2019})}\BibitemShut {NoStop}%
\bibitem [{\citenamefont {Skinner}\ \emph {et~al.}(2019)\citenamefont
  {Skinner}, \citenamefont {Ruhman},\ and\ \citenamefont
  {Nahum}}]{Skinner_19_Measurement}%
  \BibitemOpen
  \bibfield  {author} {\bibinfo {author} {\bibfnamefont {B.}~\bibnamefont
  {Skinner}}, \bibinfo {author} {\bibfnamefont {J.}~\bibnamefont {Ruhman}}, \
  and\ \bibinfo {author} {\bibfnamefont {A.}~\bibnamefont {Nahum}},\ }\href
  {\doibase 10.1103/PhysRevX.9.031009} {\bibfield  {journal} {\bibinfo
  {journal} {Phys. Rev. X}\ }\textbf {\bibinfo {volume} {9}},\ \bibinfo {pages}
  {031009} (\bibinfo {year} {2019})}\BibitemShut {NoStop}%
\bibitem [{\citenamefont {Martin}\ and\ \citenamefont
  {Vennin}(2017)}]{Martin_17_Obstructions}%
  \BibitemOpen
  \bibfield  {author} {\bibinfo {author} {\bibfnamefont {J.}~\bibnamefont
  {Martin}}\ and\ \bibinfo {author} {\bibfnamefont {V.}~\bibnamefont
  {Vennin}},\ }\href {\doibase 10.1103/PhysRevD.96.063501} {\bibfield
  {journal} {\bibinfo  {journal} {Phys. Rev. D}\ }\textbf {\bibinfo {volume}
  {96}},\ \bibinfo {pages} {063501} (\bibinfo {year} {2017})}\BibitemShut
  {NoStop}%
\bibitem [{\citenamefont {Jing}\ \emph {et~al.}(2019)\citenamefont {Jing},
  \citenamefont {Fadel}, \citenamefont {Ivannikov},\ and\ \citenamefont
  {Byrnes}}]{jing2019split}%
  \BibitemOpen
  \bibfield  {author} {\bibinfo {author} {\bibfnamefont {Y.}~\bibnamefont
  {Jing}}, \bibinfo {author} {\bibfnamefont {M.}~\bibnamefont {Fadel}},
  \bibinfo {author} {\bibfnamefont {V.}~\bibnamefont {Ivannikov}}, \ and\
  \bibinfo {author} {\bibfnamefont {T.}~\bibnamefont {Byrnes}},\ }\href@noop {}
  {\bibfield  {journal} {\bibinfo  {journal} {New Journal of Physics}\ }\textbf
  {\bibinfo {volume} {21}},\ \bibinfo {pages} {093038} (\bibinfo {year}
  {2019})}\BibitemShut {NoStop}%
\bibitem [{\citenamefont {Fadel}\ \emph {et~al.}(2018)\citenamefont {Fadel},
  \citenamefont {Zibold}, \citenamefont {D{\'e}camps},\ and\ \citenamefont
  {Treutlein}}]{fadel2018spatial}%
  \BibitemOpen
  \bibfield  {author} {\bibinfo {author} {\bibfnamefont {M.}~\bibnamefont
  {Fadel}}, \bibinfo {author} {\bibfnamefont {T.}~\bibnamefont {Zibold}},
  \bibinfo {author} {\bibfnamefont {B.}~\bibnamefont {D{\'e}camps}}, \ and\
  \bibinfo {author} {\bibfnamefont {P.}~\bibnamefont {Treutlein}},\ }\href@noop
  {} {\bibfield  {journal} {\bibinfo  {journal} {Science}\ }\textbf {\bibinfo
  {volume} {360}},\ \bibinfo {pages} {409} (\bibinfo {year}
  {2018})}\BibitemShut {NoStop}%
\bibitem [{\citenamefont {Larsson}(2014)}]{Larsson_14_Loopholes}%
  \BibitemOpen
  \bibfield  {author} {\bibinfo {author} {\bibfnamefont {J.-{\AA}.}\
  \bibnamefont {Larsson}},\ }\href {\doibase 10.1088/1751-8113/47/42/424003}
  {\bibfield  {journal} {\bibinfo  {journal} {Journal of Physics A:
  Mathematical and Theoretical}\ }\textbf {\bibinfo {volume} {47}},\ \bibinfo
  {pages} {424003} (\bibinfo {year} {2014})}\BibitemShut {NoStop}%
\bibitem [{\citenamefont {Popescu}\ and\ \citenamefont
  {Rohrlich}(1994)}]{Popescu_94_Quantum}%
  \BibitemOpen
  \bibfield  {author} {\bibinfo {author} {\bibfnamefont {S.}~\bibnamefont
  {Popescu}}\ and\ \bibinfo {author} {\bibfnamefont {D.}~\bibnamefont
  {Rohrlich}},\ }\href {\doibase 10.1007/BF02058098} {\bibfield  {journal}
  {\bibinfo  {journal} {Foundations of Physics}\ }\textbf {\bibinfo {volume}
  {24}},\ \bibinfo {pages} {379} (\bibinfo {year} {1994})}\BibitemShut
  {NoStop}%
\bibitem [{\citenamefont {Janotta}\ and\ \citenamefont
  {Hinrichsen}(2014)}]{Janotta_14_Generalized}%
  \BibitemOpen
  \bibfield  {author} {\bibinfo {author} {\bibfnamefont {P.}~\bibnamefont
  {Janotta}}\ and\ \bibinfo {author} {\bibfnamefont {H.}~\bibnamefont
  {Hinrichsen}},\ }\href {\doibase 10.1088/1751-8113/47/32/323001} {\bibfield
  {journal} {\bibinfo  {journal} {Journal of Physics A: Mathematical and
  Theoretical}\ }\textbf {\bibinfo {volume} {47}},\ \bibinfo {pages} {323001}
  (\bibinfo {year} {2014})}\BibitemShut {NoStop}%
\bibitem [{\citenamefont {Weedbrook}\ \emph {et~al.}(2012)\citenamefont
  {Weedbrook}, \citenamefont {Pirandola}, \citenamefont {Garc\'{\i}a-Patr\'on},
  \citenamefont {Cerf}, \citenamefont {Ralph}, \citenamefont {Shapiro},\ and\
  \citenamefont {Lloyd}}]{Weedbrook_12_Gaussian}%
  \BibitemOpen
  \bibfield  {author} {\bibinfo {author} {\bibfnamefont {C.}~\bibnamefont
  {Weedbrook}}, \bibinfo {author} {\bibfnamefont {S.}~\bibnamefont
  {Pirandola}}, \bibinfo {author} {\bibfnamefont {R.}~\bibnamefont
  {Garc\'{\i}a-Patr\'on}}, \bibinfo {author} {\bibfnamefont {N.~J.}\
  \bibnamefont {Cerf}}, \bibinfo {author} {\bibfnamefont {T.~C.}\ \bibnamefont
  {Ralph}}, \bibinfo {author} {\bibfnamefont {J.~H.}\ \bibnamefont {Shapiro}},
  \ and\ \bibinfo {author} {\bibfnamefont {S.}~\bibnamefont {Lloyd}},\ }\href
  {\doibase 10.1103/RevModPhys.84.621} {\bibfield  {journal} {\bibinfo
  {journal} {Rev. Mod. Phys.}\ }\textbf {\bibinfo {volume} {84}},\ \bibinfo
  {pages} {621} (\bibinfo {year} {2012})}\BibitemShut {NoStop}%
\bibitem [{\citenamefont {{{\v{S}}upi{\'c}}}\ and\ \citenamefont
  {{Bowles}}(2019)}]{Supic_19_Self}%
  \BibitemOpen
  \bibfield  {author} {\bibinfo {author} {\bibfnamefont {I.}~\bibnamefont
  {{{\v{S}}upi{\'c}}}}\ and\ \bibinfo {author} {\bibfnamefont {J.}~\bibnamefont
  {{Bowles}}},\ }\href@noop {} {\bibfield  {journal} {\bibinfo  {journal}
  {arXiv e-prints}\ ,\ \bibinfo {eid} {arXiv:1904.10042}} (\bibinfo {year}
  {2019})},\ \Eprint {http://arxiv.org/abs/1904.10042} {arXiv:1904.10042}
  \BibitemShut {NoStop}%
\bibitem [{\citenamefont {Cavalcanti}\ \emph {et~al.}(2010)\citenamefont
  {Cavalcanti}, \citenamefont {Salles},\ and\ \citenamefont
  {Scarani}}]{Cavalcanti_10_Macroscopically}%
  \BibitemOpen
  \bibfield  {author} {\bibinfo {author} {\bibfnamefont {D.}~\bibnamefont
  {Cavalcanti}}, \bibinfo {author} {\bibfnamefont {A.}~\bibnamefont {Salles}},
  \ and\ \bibinfo {author} {\bibfnamefont {V.}~\bibnamefont {Scarani}},\ }\href
  {\doibase 10.1038/ncomms1138} {\bibfield  {journal} {\bibinfo  {journal}
  {Nature Communications}\ }\textbf {\bibinfo {volume} {1}},\ \bibinfo {pages}
  {136} (\bibinfo {year} {2010})}\BibitemShut {NoStop}%
\bibitem [{\citenamefont {Chitambar}\ and\ \citenamefont
  {Gour}(2019)}]{Chitambar_19_Quantum}%
  \BibitemOpen
  \bibfield  {author} {\bibinfo {author} {\bibfnamefont {E.}~\bibnamefont
  {Chitambar}}\ and\ \bibinfo {author} {\bibfnamefont {G.}~\bibnamefont
  {Gour}},\ }\href {\doibase 10.1103/RevModPhys.91.025001} {\bibfield
  {journal} {\bibinfo  {journal} {Rev. Mod. Phys.}\ }\textbf {\bibinfo {volume}
  {91}},\ \bibinfo {pages} {025001} (\bibinfo {year} {2019})}\BibitemShut
  {NoStop}%
\bibitem [{\citenamefont {Horodecki}\ \emph {et~al.}(2009)\citenamefont
  {Horodecki}, \citenamefont {Horodecki}, \citenamefont {Horodecki},\ and\
  \citenamefont {Horodecki}}]{Horodecki_09_Quantum}%
  \BibitemOpen
  \bibfield  {author} {\bibinfo {author} {\bibfnamefont {R.}~\bibnamefont
  {Horodecki}}, \bibinfo {author} {\bibfnamefont {P.}~\bibnamefont
  {Horodecki}}, \bibinfo {author} {\bibfnamefont {M.}~\bibnamefont
  {Horodecki}}, \ and\ \bibinfo {author} {\bibfnamefont {K.}~\bibnamefont
  {Horodecki}},\ }\href {\doibase 10.1103/RevModPhys.81.865} {\bibfield
  {journal} {\bibinfo  {journal} {Rev. Mod. Phys.}\ }\textbf {\bibinfo {volume}
  {81}},\ \bibinfo {pages} {865} (\bibinfo {year} {2009})}\BibitemShut
  {NoStop}%
\bibitem [{\citenamefont {Wong}\ \emph {et~al.}(prep)\citenamefont {Wong},
  \citenamefont {Shapiro}, \citenamefont {{Bene Watts}},\ and\ \citenamefont
  {{Yunger Halpern}}}]{Wong_Experiment}%
  \BibitemOpen
  \bibfield  {author} {\bibinfo {author} {\bibfnamefont {F.}~\bibnamefont
  {Wong}}, \bibinfo {author} {\bibfnamefont {J.~H.}\ \bibnamefont {Shapiro}},
  \bibinfo {author} {\bibfnamefont {A.}~\bibnamefont {{Bene Watts}}}, \ and\
  \bibinfo {author} {\bibfnamefont {N.}~\bibnamefont {{Yunger Halpern}}},\
  }\href@noop {} {} (\bibinfo {year} {in prep})\BibitemShut {NoStop}%
\bibitem [{\citenamefont {Menicucci}\ and\ \citenamefont
  {Caves}(2002)}]{Menicucci_02_Local}%
  \BibitemOpen
  \bibfield  {author} {\bibinfo {author} {\bibfnamefont {N.~C.}\ \bibnamefont
  {Menicucci}}\ and\ \bibinfo {author} {\bibfnamefont {C.~M.}\ \bibnamefont
  {Caves}},\ }\href {\doibase 10.1103/PhysRevLett.88.167901} {\bibfield
  {journal} {\bibinfo  {journal} {Phys. Rev. Lett.}\ }\textbf {\bibinfo
  {volume} {88}},\ \bibinfo {pages} {167901} (\bibinfo {year}
  {2002})}\BibitemShut {NoStop}%
\bibitem [{\citenamefont {{Hardy}}(2001)}]{Hardy_01_Quantum}%
  \BibitemOpen
  \bibfield  {author} {\bibinfo {author} {\bibfnamefont {L.}~\bibnamefont
  {{Hardy}}},\ }\href@noop {} {\bibfield  {journal} {\bibinfo  {journal} {arXiv
  e-prints}\ ,\ \bibinfo {eid} {quant-ph/0101012}} (\bibinfo {year} {2001})},\
  \Eprint {http://arxiv.org/abs/quant-ph/0101012} {arXiv:quant-ph/0101012
  [quant-ph]} \BibitemShut {NoStop}%
\bibitem [{\citenamefont {Leggett}\ and\ \citenamefont
  {Garg}(1985)}]{Leggett_85_Quantum}%
  \BibitemOpen
  \bibfield  {author} {\bibinfo {author} {\bibfnamefont {A.~J.}\ \bibnamefont
  {Leggett}}\ and\ \bibinfo {author} {\bibfnamefont {A.}~\bibnamefont {Garg}},\
  }\href {\doibase 10.1103/PhysRevLett.54.857} {\bibfield  {journal} {\bibinfo
  {journal} {Phys. Rev. Lett.}\ }\textbf {\bibinfo {volume} {54}},\ \bibinfo
  {pages} {857} (\bibinfo {year} {1985})}\BibitemShut {NoStop}%
\bibitem [{\citenamefont {Cirel'son}(1980)}]{cirel1980quantum}%
  \BibitemOpen
  \bibfield  {author} {\bibinfo {author} {\bibfnamefont {B.~S.}\ \bibnamefont
  {Cirel'son}},\ }\href@noop {} {\bibfield  {journal} {\bibinfo  {journal}
  {Letters in Mathematical Physics}\ }\textbf {\bibinfo {volume} {4}},\
  \bibinfo {pages} {93} (\bibinfo {year} {1980})}\BibitemShut {NoStop}%
\bibitem [{\citenamefont {Kim}(2008)}]{kim2008applications}%
  \BibitemOpen
  \bibfield  {author} {\bibinfo {author} {\bibfnamefont {T.}~\bibnamefont
  {Kim}},\ }\href@noop {} {\emph {\bibinfo {title} {Applications of
  single-photon two-qubit quantum logic to the quantum information science}}},\
  \bibinfo {type} {Tech. Rep.}\ (\bibinfo  {institution} {Massachusetts
  Institute of Technology},\ \bibinfo {year} {2008})\BibitemShut {NoStop}%
\bibitem [{\citenamefont {Pittman}\ \emph {et~al.}(2003)\citenamefont
  {Pittman}, \citenamefont {Donegan}, \citenamefont {Fitch}, \citenamefont
  {Jacobs}, \citenamefont {Franson}, \citenamefont {Kok}, \citenamefont {Lee},\
  and\ \citenamefont {Dowling}}]{pittman2003heralded}%
  \BibitemOpen
  \bibfield  {author} {\bibinfo {author} {\bibfnamefont {T.~B.}\ \bibnamefont
  {Pittman}}, \bibinfo {author} {\bibfnamefont {M.~M.}\ \bibnamefont
  {Donegan}}, \bibinfo {author} {\bibfnamefont {M.~J.}\ \bibnamefont {Fitch}},
  \bibinfo {author} {\bibfnamefont {B.~C.}\ \bibnamefont {Jacobs}}, \bibinfo
  {author} {\bibfnamefont {J.~D.}\ \bibnamefont {Franson}}, \bibinfo {author}
  {\bibfnamefont {P.}~\bibnamefont {Kok}}, \bibinfo {author} {\bibfnamefont
  {H.}~\bibnamefont {Lee}}, \ and\ \bibinfo {author} {\bibfnamefont {J.~P.}\
  \bibnamefont {Dowling}},\ }\href@noop {} {\bibfield  {journal} {\bibinfo
  {journal} {IEEE Journal of selected topics in quantum electronics}\ }\textbf
  {\bibinfo {volume} {9}},\ \bibinfo {pages} {1478} (\bibinfo {year}
  {2003})}\BibitemShut {NoStop}%
\bibitem [{\citenamefont {Wagenknecht}\ \emph {et~al.}(2010)\citenamefont
  {Wagenknecht}, \citenamefont {Li}, \citenamefont {Reingruber}, \citenamefont
  {Bao}, \citenamefont {Goebel}, \citenamefont {Chen}, \citenamefont {Zhang},
  \citenamefont {Chen},\ and\ \citenamefont
  {Pan}}]{wagenknecht2010experimental}%
  \BibitemOpen
  \bibfield  {author} {\bibinfo {author} {\bibfnamefont {C.}~\bibnamefont
  {Wagenknecht}}, \bibinfo {author} {\bibfnamefont {C.-M.}\ \bibnamefont {Li}},
  \bibinfo {author} {\bibfnamefont {A.}~\bibnamefont {Reingruber}}, \bibinfo
  {author} {\bibfnamefont {X.-H.}\ \bibnamefont {Bao}}, \bibinfo {author}
  {\bibfnamefont {A.}~\bibnamefont {Goebel}}, \bibinfo {author} {\bibfnamefont
  {Y.-A.}\ \bibnamefont {Chen}}, \bibinfo {author} {\bibfnamefont
  {Q.}~\bibnamefont {Zhang}}, \bibinfo {author} {\bibfnamefont
  {K.}~\bibnamefont {Chen}}, \ and\ \bibinfo {author} {\bibfnamefont {J.-W.}\
  \bibnamefont {Pan}},\ }\href@noop {} {\bibfield  {journal} {\bibinfo
  {journal} {Nature Photonics}\ }\textbf {\bibinfo {volume} {4}},\ \bibinfo
  {pages} {549} (\bibinfo {year} {2010})}\BibitemShut {NoStop}%
\bibitem [{\citenamefont {Wang}\ \emph {et~al.}(2020)\citenamefont {Wang},
  \citenamefont {Neuman},\ and\ \citenamefont {Narang}}]{wang2020dipole}%
  \BibitemOpen
  \bibfield  {author} {\bibinfo {author} {\bibfnamefont {D.~S.}\ \bibnamefont
  {Wang}}, \bibinfo {author} {\bibfnamefont {T.}~\bibnamefont {Neuman}}, \ and\
  \bibinfo {author} {\bibfnamefont {P.}~\bibnamefont {Narang}},\ }\href@noop {}
  {\bibfield  {journal} {\bibinfo  {journal} {arXiv preprint arXiv:2004.13725}\
  } (\bibinfo {year} {2020})}\BibitemShut {NoStop}%
\bibitem [{\citenamefont {Bock}\ \emph {et~al.}(2016)\citenamefont {Bock},
  \citenamefont {Lenhard}, \citenamefont {Chunnilall},\ and\ \citenamefont
  {Becher}}]{Bock_16_Highly}%
  \BibitemOpen
  \bibfield  {author} {\bibinfo {author} {\bibfnamefont {M.}~\bibnamefont
  {Bock}}, \bibinfo {author} {\bibfnamefont {A.}~\bibnamefont {Lenhard}},
  \bibinfo {author} {\bibfnamefont {C.}~\bibnamefont {Chunnilall}}, \ and\
  \bibinfo {author} {\bibfnamefont {C.}~\bibnamefont {Becher}},\ }\href
  {\doibase 10.1364/OE.24.023992} {\bibfield  {journal} {\bibinfo  {journal}
  {Opt. Express}\ }\textbf {\bibinfo {volume} {24}},\ \bibinfo {pages} {23992}
  (\bibinfo {year} {2016})}\BibitemShut {NoStop}%
\bibitem [{Ben(2020)}]{BeneWatts_19_Code}%
  \BibitemOpen
  \href@noop {} {\enquote {\bibinfo {title} {Online code repository},}\ }
  (\bibinfo {year} {2020}),\ \bibinfo {note}
  {https://github.com/aBeneWatts/macroscopic-bell}\BibitemShut {NoStop}%
\end{thebibliography}%



\end{document}